\documentclass[10pt, oneside]{article}

\input ./math_headers.sty

\title{A Physical Origin for\\ Singular Support Conditions in\\ Geometric Langlands Theory}
\author{Chris Elliott and Philsang Yoo }
\date{\today}

\DeclareMathOperator{\EOM}{EOM}
\newcommand{\vac}{\mc{V} \text{ac}}
\newcommand{\del}{\partial}
\def\d{{\rm d}}

\newcommand{\IC}{\mathrm{IndCoh}}
\renewcommand{\flat}{\mathrm{Flat}}
\newcommand{\arth}{\mathrm{Arth}}
\newcommand{\sing}{\mathrm{Sing}}
\newcommand{\dR}{\mathrm{dR}}

\begin{document}
\maketitle

\begin{abstract}
We explain how the nilpotent singular support condition introduced into the geometric Langlands conjecture by Arinkin and Gaitsgory arises naturally from the point of view of 4-dimensional $\mc N=4$ supersymmetric gauge theory.  We define what it means in topological quantum field theory to restrict a category of boundary conditions to the full subcategory of objects compatible with a fixed choice of vacuum, both in functorial field theory and in the language of factorization algebras.  For B-twisted $\mc N=4$ gauge theory with gauge group $G$, the moduli space of vacua is equivalent to $\hh^*/W$, and the nilpotent singular support condition arises by restricting to the vacuum $0 \in \hh^*/W$.  We then investigate the categories obtained by restricting to points in larger strata, and conjecture that these categories are equivalent to the geometric Langlands categories with gauge symmetry broken to a Levi subgroup, and furthermore that by assembling such for the groups $\GL_n$ with ${n\geq 1}$ one finds a hidden factorization algebra structure for the geometric Langlands theory.
\end{abstract}

\begingroup
\def\addvspace#1{}
\tableofcontents
\endgroup

\section{Introduction} \label{introduction}

\subsection{Goals} \label{goals_section}
In this paper we will use ideas from derived algebraic geometry to describe a novel structure associated to topological quantum field theories: the derived moduli space of vacua.  Our main motivation for studying this structure comes from one family of examples, namely the 4-dimensional topologically twisted field theories introduced by Kapustin and Witten \cite{KW} in order to describe a physical origin for the geometric Langlands conjecture.  In Kapustin and Witten's approach the categories appearing in the geometric Langlands conjecture arise as categories of boundary conditions in these topologically twisted field theories.  However, this analysis does not explain the essential refinement to the conjecture that was introduced by Arinkin and Gaitsgory \cite{ArinkinGaitsgory}.  We will explain how Kapustin and Witten's categories of boundary conditions fiber over the derived moduli space of vacua, so that the Arinkin--Gaitsgory refined category arises by imposing a natural support condition.

As this summary suggests, the goals of this paper are two-fold: to introduce and motivate the idea of the derived moduli space of vacua and the relationship between this structure and categories of boundary conditions, and to explain why applying this general notion to examples coming from Kapustin and Witten's work explains the refinement of Arinkin and Gaitsgory.  In order to set up these goals we will discuss the context for quantum field theory with which we will work, and explain why derived algebraic geometry will play an essential role.

\begin{remark}
We'll use a number of definitions and ideas from the theory of derived algebraic geometry and homotopical algebra in this paper.  In order to make this paper accessible to readers who aren't familiar with these fields we will endeavor to briefly introduce these ideas when they first appear.  We will refer the reader to the relevant literature for more complete and comprehensive explanations; for overall references we recommend the excellent survey articles of To\"en \cite{ToenOverview,ToenSurvey}, and the book of Gaitsgory and Rozenblyum \cite{GRvol1, GRvol2}.
\end{remark}

To begin with, we should emphasize that there is not yet a satisfactory mathematical formalism for quantum field theory that includes the examples we are most interested in (topologically twisted gauge theories such as those described by Kapustin and Witten) and models the categories of boundary conditions in these theories.  We will work with a mathematical model for a \emph{perturbative} topological quantum field theory along with a single dg category modelling boundary conditions.  More explicitly we will use the following framework.
\begin{itemize}
 \item There is a very general formalism available for the study of \emph{perturbative} quantum field theory on $\RR^n$ given by the theory of factorization algebras as developed by Costello and Gwilliam \cite{CostelloGwilliam1, CostelloGwilliam2}.  We will introduce this formalism in Section \ref{fact_alg_background_section}, but in brief we assign a cochain complex $\obs(U)$ of local observables to each open subset $U \sub \RR^n$, along with a linear operator product map $\obs(U_1) \otimes \cdots \obs(U_k) \to \obs(V)$ to every configuration of disjoint open sets $U_1 \sqcup \cdots U_k \inj V$.  We'll be particularly interested in a special class of factorization algebra, where the local observables associated to homotopy equivalent open sets are quasi-isomorphic: these are called $\bb E_n$-algebras, and are associated to \emph{topological} field theories.
 \item In addition to this perturbative data, we'll specify a dg category modelling \emph{boundary conditions} for the theory on the manifold $M^{n-2} \times \RR \times \RR_{\ge 0}$.  The idea that boundary conditions in a topological field theory should be represented by a dg category is explained, for example, by Kapustin \cite{KapustinNote}; in brief the idea is that possible boundary conditions are represented by objects of the category, and for any pair of objects $\mc B_1$ and $\mc B_2$ the Hilbert space of the theory on $M^{n-2} \times \RR \times [0,t]$ with the two given boundary conditions is represented by the hom-space $\mr{Hom}(\mc B_1, \mc B_2)$.  In our model we'll additionally specify an action of the $\bb E_n$-algebra $\mc A$ of local observables on the category $\mc B$ of boundary conditions.  By this we mean an action of the monoidal \footnote{The monoidal structure we'll use depends on a choice of plane $\RR^2 \sub \RR^n$.  We'll explain this in Section \ref{restriction_background_section}.} dg category $\mc A\text{-mod}$ on the category $\mc B$. 
 Why is this sort of action of the local observables natural from a field-theoretic point of view?  There are two ways of motivating this action: there's a construction on the level of dg categories that passes through the dg category of line operators, and there's a construction on the level of algebras using an equivalent characterization of an $\mc A\text{-mod}$ action on $\mc B$ via a universal property.  We defer an explanation until Section \ref{construction_intro_section}.
\end{itemize}
It should be possible to associate the data of an $\bb E_n$-algebra acting on a dg category to any choice of a topological quantum field theory and an $(n-2)$-manifold $M$.  What's more this data describes a non-trivial part of the topological field theory on $M^{n-2} \times \RR \times \RR_{\ge 0}$: the complete structure of the perturbative theory in the bulk, as well as the possible boundary conditions as a categorical ``module'' for these perturbative local observables, suggesting that this model is not purely a contrivance but a useful model for understanding aspects of topological field theory.  Let's discuss what this structure looks like in some examples.

\begin{examples}
\begin{enumerate}
 \item The simplest family of examples that fit into this formalism are 2d fully extended TQFTs in the functorial sense (so, in particular, $M$ is a point).  In 2d TQFTs there is a single dg category $\mc C$ of boundary conditions, and the algebra of local observables is modelled by the algebra of Hochschild cochains of $\mc C$.  The Hochschild cochains of a dg category $\mc C$ have a canonical $\bb E_2$ structure and a canonical action on $\mc C$ (for more details see Section \ref{restriction_background_section}).  For example, let $X$ be a smooth projective variety, and consider the dg category $\coh(X)$ of coherent complexes of sheaves on $X$.  This is the category of boundary conditions in the 2d topological B-model with target space $X$.  The algebra of local observables in this theory (or rather, its cohomology) can be described as the Dolbeault cohomology $\mr H^{0,\bullet}(X; \wedge^\bullet T^{1,0}_X)$ of $X$ with coefficients in the bundle of holomorphic polyvector fields (see e.g. \cite[Section 4.3]{BBZBDN}).
\item Another 2d example is provided by the Landau--Ginzburg B-model with target $\CC^n$ and superpotential $W \colon \CC^n \to \CC$.  The category of boundary conditions in this theory can be described as the category $\mr{MF}(W)$ of matrix factorizations of $W$ \cite{KapustinLi}. According to a theorem of Dyckerhoff \cite{Dyckerhoff}, the algebra of local observables (or rather, its cohomology) is equivalent to the Jacobian ring $\mr{Jac}(W)$ of $W$, i.e. to the quotient of the polynomial ring $\CC[z_1, \ldots, z_n]$ by the partial derivatives of $W$.
\item To give a more informal example, 2d A-models are expected to have, as their categories of boundary conditions, versions of the Fukaya category of the target.  For example, Ganatra \cite{GanatraThesis} studied the Hochschild cohomology (so the algebra of local observables in an A-model type theory) of wrapped Fukaya categories associated to an exact symplectic manifold $X$.  He proved that the Hochschild cohomology was equivalent to an invariant known as the symplectic cohomology of $X$.
\item Higher dimensional examples are typically more difficult to define rigorously.  As well as the four-dimensional theory that is the main focus of this paper, let us briefly describe a three-dimensional example, that of Rozansky--Witten theory.  The algebra of local observables of this theory (or rather, its cohomology) is discussed at length in \cite[Section 5]{BBZBDN}.  The category of boundary conditions on the circle has also been described in the physics literature by Kapustin, Rozansky, and Saulina \cite{KapustinRozanskySaulinaI, KapustinRozanskyII}, though to our knowledge the action of the algebra of observables on the category of boundary conditions has not yet been studied.
\end{enumerate}
\vspace{-10pt}
\end{examples}

The main goal of this paper is to study this structure -- of the dg category of boundary conditions and the action of local observables -- for a particular topological field theory: the 4d B-twisted gauge theory relevant for the geometric Langlands program.  We'll argue that by considering the action of the algebra of local observables on Kapustin and Witten's dg category we can explain, from the point of view of quantum field theory, an important and previously physically unexplained structure in the geometric Langlands program: the condition of nilpotent singular support.  In order to explain our approach we'll first review the highlights of the geometric Langlands program, the theories of Kapustin and Witten and our model for the geometric Langlands twists using derived algebraic geometry from \cite{EY1}.

\begin{remark}
Let us briefly remark on the categorical context in which we'll be working.  By a ``category'', in this paper, unless otherwise specified, we will always mean an $\infty$-category. In particular we can view dg categories over the complex numbers as stable $\infty$-categories, for instance as discussed in \cite[I.1.10]{GRvol1}: all of the main results in this paper will take place in that context.  Likewise, when we talk about associative algebras, or more generally algebras over any operad, we'll always mean algebras in the category of cochain complexes over $\CC$. 
\end{remark}

\subsection{Algebraic Moduli Spaces in Kapustin--Witten Gauge Theory} \label{KW_intro_section}
In this subsection let us recall some aspects of Kapustin and Witten's approach, and explain the context in which we interpreted Kapustin and Witten's theory in our previous work \cite{EY1}.  We'll start with the idea of a topological twist.  Recall that topological twisting is a procedure that extracts a topological field theory from a quantum field theory with an odd symmetry $Q$ that squares to zero (usually obtained from the action of a supersymmetry group).  This formalism can be interpreted using the language of derived algebraic geometry: this interpretation is the main object of study of \cite{EY1}, based on insights of Costello \cite{CostelloSH}.

\begin{remark} \label{derived_EOM_remark}
Let us remark on the role played by derived geometry in field theory, from the point of view of this paper.  When we talk about, for instance, the derived space of solutions to the equations of motion (or the derived critical locus of the action functional) we mean the following.  The critical locus of the action functional $S$ in the space of fields $\Phi$ of a classical field theory can be thought of as the intersection in the cotangent space $T^*\Phi$ of the zero section $\Phi$ with the graph $\Gamma_{\d S}$ of the one-form $\d S$ -- in other words the locus in $\Phi$ where $\d S = 0$.  The \emph{derived} critical locus is an enhancement of this usual intersection: geometrically one can interpret the derived intersection $X_1 \cap^h X_2$ as modelling intersection points after a small perturbation of $X_1$ that makes the intersection transverse, but where we remember the choice of perturbation as part of the data.  More concretely, if $X_1 = \spec R_1 $ and $X_2 = \spec R_2$ are affine subschemes of $Y = \spec S$, the derived intersection is obtained from the derived tensor product:
\[\spec R_1 \cap^h \spec R_2 = \spec (R_1 \otimes^{\bb L}_S R_2),\]
meaning that the $S$-module $R_1$, say, is replaced by an appropriate resolution.  In classical field theory this derived intersection is described locally near a fixed classical solution by the classical BV-BRST formalism.
\end{remark}

We applied the formalism for twisting in terms of derived algebraic geometry to a family of examples first studied by Kapustin and Witten.  Let $G_c$ be a compact connected Lie group with complexification $G$.  Kapustin and Witten described a $\bb{CP}^1$-family of topological twists of a gauge theory with gauge group $G_c$ and $\mc N=4$ supersymmetry.  We can summarize some of the most important properties of this family as described in \cite{KW} as follows.
\vspace{-10pt}
\begin{itemize}
\item There is a $\bb{CP}^1$ parameter denoted $\Psi$, a rational combination of the choice of twisting supercharge and the coupling constant, and the twisted quantum field theory depends only on $\Psi$.  These theories have the property that S-duality interchanges the twisted theory with gauge group $G_c$ at the parameter $\Psi$ with the twisted theory with Langlands dual gauge group $G_c^\vee$ at the antipodal parameter $-1/\Psi$.
\item Compactify the twisted theory with gauge group $G_c$ at $\Psi = 0$ on a Riemann surface $C$. The resulting 2d theory is equivalent to the A-model with target $\higgs_G(C)$ -- the moduli stack of $G$-Higgs bundles.  The category of boundary conditions in this theory is equivalent to the category $\text{D}(\bun_G(C))$ of D-modules on $\bun_G(C)$ -- the moduli stack of holomorphic $G$-bundles.
\item Compactify the twisted theory with gauge group $G_c^\vee$ at $\Psi = \infty$ on a Riemann surface $C$. The resulting 2d theory is equivalent to the B-model with the target $\loc_{G^\vee}(C)$ -- the moduli stack of $G^\vee$ local systems.  The category of boundary conditions in this theory is equivalent to the category $\QC(\loc_{G^\vee}(C))$ of quasi-coherent sheaves on $\loc_{G^\vee}(C)$.
\item S-duality induces an equivalence between these categories of boundary conditions, and therefore an equivalence
\[\text{D}(\bun_G(C)) \iso \QC(\loc_{G^\vee}(C)).\]
\end{itemize}

There is an obstacle that arises when we try to compare this story to the geometric Langlands story as it appears in geometric representation theory: there is no visible dependence on the algebraic structure of $C$.  As we'll discuss in the next subsection, the usual geometric Langlands equivalence depends on this algebraic structure.  Most visibly the stack appearing on the right-hand side is not the stack $\loc_{G^\vee}(C))$ of $G^\vee$-local systems on $C$, but the stack $\flat_{G^\vee}$ of $G^\vee$-bundles with flat connection.  These two stacks are analytically equivalent, but algebraically quite different.  The stack of local systems only depends on the topology of $C$, whereas the stack of flat connections depends on its structure as an algebraic curve.

\begin{remark}
Throughout this paper we use the notation $\flat_G(C)$ for the (derived) moduli stack of algebraic flat $G$-bundles on $C$.  In the geometric representation theory literature this is usually denoted by $\mr{LocSys}_G(C)$ or $\loc_G(C)$, a notation we prefer to reserve for the moduli stack of $G$-local systems (distinguished as above).
\end{remark}

\begin{remark}
In addition to the algebraic version of the geometric Langlands conjecture there is a topological version -- the \emph{Betti Langlands} conjecture -- introduced by Ben-Zvi--Brochier--Jordan \cite{BZBJ} and Ben-Zvi--Nadler \cite{BZNBettiLanglands}, which they call ``Betti Langlands'' (see also work by Nadler and Yun \cite{NadlerYun}). 
\end{remark}

In our paper \cite{EY1} we applied a general construction of twisting to the example of $\mc N=4$ super Yang--Mills theory in order to describe the (derived) moduli spaces of solutions to the equations of motion in the Kapustin--Witten family of twists.  We can identify these moduli spaces as (derived) algebraic stacks in a natural way in order to obtain the algebraic moduli stacks appearing in the usual geometric Langlands conjecture.  We'll summarize these results in Section \ref{twisting_section}.

\subsection{Geometric Langlands and Singular Support} \label{GL_intro_section}
The original Langlands program consists of a tantalizing set of conjectures which relate harmonic analysis, representation theory, algebraic geometry, and algebraic number theory. In particular, the reciprocity conjecture, in its simplest form, expects a close relationship between certain automorphic representations of $G(\bb A_\bb Q)$ and homomorphisms from the Galois group $\mr{Gal}(\overline{\bb Q}/\bb Q )$ to the Langlands dual group $G^\vee$.

Through Weil's well-known ``Rosetta stone'' \cite{WeilLetters}, there is an analogy between number fields and function fields of curves over a finite field, from which one can make a corresponding version of the Langlands reciprocity conjecture for function fields. Moreover, over function fields one can geometrize the problem via Grothendieck's function-sheaf correspondence. If we restrict our attention only to the unramified case, one replaces automorphic functions of $G(\bb A_F)$ for a number field $F$ by automorphic sheaves over $\bun_G(C)$, and $G^\vee$-representations of $\pi^{ \text{\'et}}_1( F)= \mr{Gal}(\overline{ F} / F )$ by $G^\vee$-local systems on $C $. It is a form of this conjecture that Drinfeld proved for $G=\GL_2$ \cite{DrinfeldGL2}, after which Laumon formulated a precise conjecture for $\GL_n$ \cite{LaumonFunctionFields}. Later the work of Frenkel, Gaitsgory, and Vilonen \cite{FGV} combined with a result of Gaitsgory \cite{GaitsgoryVanishing} explicitly constructed a cuspidal automorphic sheaf associated to an irreducible local system in unramified cases.

Once formulated in such geometric terms, one can also make such a conjecture for the function field of a curve $C$ over the field $\bb C$ of complex numbers, rather than over a finite field, and directly translate the results discussed above to this complex algebraic context. On the other hand, since the geometry becomes easier, one can hope to prove a stronger result (at least in the unramified case): instead of just a set-theoretic relationship between cuspidal sheaves on the moduli stack $\bun_G(C)$ of algebraic $G$-bundles and irreducible local systems on $C$, one would like to construct an equivalence of \emph{categories} in which those objects naturally sit. 

In the complex geometric setting, an algebraic local system can be replaced by a de Rham local system.  For instance, $G^\vee$-local systems can be thought of as the skyscraper sheaves on the moduli stack $\Flat_{G^\vee}(C)$ of flat $G^\vee$-bundles.  On the other side, an automorphic sheaf, in the unramified case, can be described as a D-module on $\bun_G(C)$: the category of constructible sheaves makes sense over $\bb C$, and can then be embedded into the larger, more algebraic, category of D-modules by the Riemann--Hilbert correspondence. With this in mind, Beilinson and Drinfeld imagined the following beautiful guiding picture.

\begin{metaconjecture}[``Best Hope'' Conjecture]
There is an equivalence of dg categories \[ \bb L_G \colon \text{D}(\bun_G(C)) \simeq \QC( \Flat_{G^\vee}(C))\] compatible with the actions of natural symmetries on both sides.
\end{metaconjecture}

This statement should be viewed as motivational but imprecise.  When $G$ is abelian, the conjecture is literally true.  This was proved independently by Laumon \cite{Laumon} and Rothstein \cite{Rothstein} using a twisted Fourier--Mukai transform adapted to send D-modules on an abelian variety $A$ to quasi-coherent sheaves on a twisted form of $T^*A^\vee$.   One must slightly augment Laumon and Rothstein's equivalence in order to include stackiness on the left-hand side (the automorphisms of the trivial line bundle) and ``derived-ness'' on the right-hand side (the derived fiber product $\pt \times_{\gg} \pt$ representing the possible transition data allowing one to glue together the zero flat connection over an affine cover of $C$).  This augmentation is straightforward, and is explained for instance in \cite[Remark 11.2.7]{ArinkinGaitsgory}.

As soon as $G$ is non-abelian the two categories, as stated above, fail to be equivalent. It turns out that the category on the A-side is ``too big''; that there are D-modules which correspond to no sheaf on the B-side. An example of this was verified by V. Lafforgue \cite{Lafforgue} for the curve $\bb P^1$, who demonstrated an equivalence between $\mr{QC}(\Flat_{G^\vee}(\bb P^1))$ and a proper subcategory of $\mr{D}(\bun_G(\bb P^1))$.  As such, in order to rescue the geometric Langlands conjecture we need to either enlarge the category on the B-side, or shrink the category on the A-side.

Arinkin and Gaitsgory \cite{ArinkinGaitsgory} proposed a category enlarging $\QC(\Flat_{G^\vee}(C))$ that has the potential to repair the geometric Langlands conjecture.  This enlargement takes place in two steps.  First one makes the following observation.  The stack $\flat_{G^\vee}(C)$ is not smooth, which means that the categories of perfect and coherent complexes of sheaves do not coincide.  The category of quasi-coherent sheaves is compactly generated by the category of perfect complexes.  If one instead forms the category compactly generated by all coherent complexes one obtains a strictly larger category of \emph{ind-coherent sheaves}.  More concretely this is the ind-completion of the category of coherent complexes, which means that objects can be thought of as formal filtered colimits of coherent complexes of sheaves.  We denote this category by $\IC(\Flat_{G^\vee}(C))$.

One way to see that this is natural is to consider a piece of structure in the geometric Langlands program called \emph{geometric Eisenstein series functors} -- a kind of functorial compatibility when we vary the group $G$.  Let $P \sub G$ be a parabolic subgroup with Levi quotient $L$.  The geometric Eisenstein series functor associated to $P$ is a functor
\[\mr{Eis}_P \colon \text{D}(\bun_L(C))\to \text{D}( \bun_G(C),\]
defined by pulling back a D-module to $\bun_P(C)$ then pushing it forward to $\bun_G(C)$.  Likewise on the spectral side we can try to define a corresponding functor
\[\mr{Eis}^{\mr{spec}}_{P^\vee} \colon \QC( \Flat_{L^\vee}(C) ) \to \QC( \Flat_{G^\vee}(C)).\]
When we do this, however, we run into a problem.  We expect geometric Langlands duality to intertwine these two functors (up to tensoring by a line bundle), but the functor on D-modules is well-behaved (it preserves compact objects) while the functor on quasi-coherent sheaves is not.  This problem is fixed by replacing quasi-coherent by ind-coherent.

However the category $\IC(\Flat_{G^\vee}(C))$ turns out to be too big.  For instance one can check this directly in the example $C = \bb{CP}^1$ using the calculation of V. Lafforgue.  Arinkin and Gaitsgory propose that one should fix this by considering the smallest possible subcategory of $\IC(\Flat_{G^\vee}(C))$ where all the geometric Eisenstein series functors for all the parabolic subgroups of $G^\vee$ are well-behaved.  The main theorem of \cite{ArinkinGaitsgory} says that this subcategory is precisely the category $ \IC_{\mc N_{G^\vee} } ( \Flat_{G^\vee}(C))$ consisting of objects satisfying the nilpotent singular support condition.

This leads one to the following corrected form of the best hope conjecture.

\begin{conjecture}[Arinkin--Gaitsgory]
There is an equivalence of dg categories \[ \bb L_G \colon \text{D}(\bun_G(C)) \simeq \IC_{\mc N_{G^\vee}} ( \Flat_{G^\vee}(C  ))\] compatible with the actions of natural symmetries on both sides.
\end{conjecture}

\subsection{Constructions and Results} \label{construction_intro_section}

With this background in hand, we'll move on to introduce the main new constructions appearing in this paper and their applications.  We'll use the action of local observables on the category of boundary conditions in order to describe natural subcategories of the category of boundary conditions.  The fundamental concept that we'll use is that if the algebra $\mc A$ of local observables acts on the category $\mc B$ of boundary conditions then we can pick out full subcategories of $\mc B$ by restricting to objects supported at a point $v \in \spec \mc A$.  In this section we'll outline this construction, explain what it means from the point of view of field theory, and apply the construction to our motivating example of twists of 4d $\mc N=4$ gauge theories.

To begin with we will explain the meaning of the action of local observables on the category of boundary conditions.  As we mentioned in Section \ref{goals_section} there are two ways of motivating this action.

\vspace{-10pt}
\begin{enumerate}
 \item The first motivating story passes through the action of the category of \emph{line operators} on the category of boundary conditions.  Line operators in a topological quantum field theory form a monoidal category that acts on the category of boundary conditions along a manifold $M$ of codimension 2 (see e.g. \cite{KapustinICM}).  Specifically, given a boundary condition for a TQFT on $M \times \RR \times \RR_{\ge 0}$, consider a line operator inserted along the line $\{p\} \times \RR \times \{\eps\}$.  By moving the line very close to the boundary (making $\eps$ very small) we can view this configuration as a single boundary condition (see Figure \ref{line_insertion_fig}).
 
 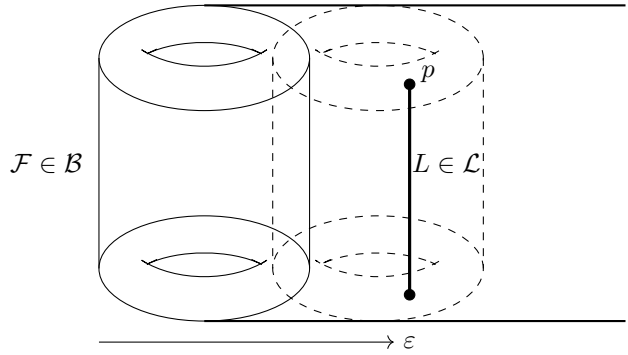
\begin{figure} [!h]
 \centering
\begin{tikzpicture}[scale=.7]
\begin{scope}
\draw[rounded corners=28pt] (-1.1,.1)--(0,-.6)--(1.1,.1);
\draw[rounded corners=24pt] (-.9,0.17)--(0,.6)--(.9,0.17);
\draw(0,0) ellipse (2 and 1);
\end{scope}

\begin{scope}[yshift=-4cm]
\draw[rounded corners=28pt] (-1.1,.1)--(0,-.6)--(1.1,.1);
\draw[rounded corners=24pt] (-.9,0.17)--(0,.6)--(.9,0.17);
\draw(0,0) ellipse (2 and 1);
\end{scope}

\draw (-2,0) -- (-2,-4);
\draw (2,0) -- (2,-4);
\node at (-3,-2) {$\mc F \in \mc B$};

\draw[thick] (0,1) -- (8,1);
\draw[thick] (0,-5) -- (8,-5);

\begin{scope}[xshift=3.3cm]
\draw[dashed,rounded corners=28pt] (-1.1,.1)--(0,-.6)--(1.1,.1);
\draw[dashed,rounded corners=24pt] (-.9,0.17)--(0,.6)--(.9,0.17);
\draw[dashed](0,0) ellipse (2 and 1);
\end{scope}

\begin{scope}[yshift=-4cm, xshift=3.3cm]
\draw[dashed,rounded corners=28pt] (-1.1,.1)--(0,-.6)--(1.1,.1);
\draw[dashed,rounded corners=24pt] (-.9,0.17)--(0,.6)--(.9,0.17);
\draw[dashed](0,0) ellipse (2 and 1);
\end{scope}

\draw[dashed] (1.3,0) -- (1.3,-4);
\draw[dashed] (5.3,0) -- (5.3,-4);

\draw[fill] (3.9,-0.5) circle (0.1);
\draw[fill] (3.9,-4.5) circle (0.1);
\draw[very thick] (3.9,-0.5) -- (3.9,-4.5);
\node at (4.25,-0.3) {$p$};
\node at (4.6,-2) {$L \in \LL$};
\draw[->] (-2,-5.4) -- (3.6,-5.4);
\node at (3.9,-5.4) {$\eps$};

\end{tikzpicture}
\caption{Insertion of a line operator $L \in \LL$ near a boundary condition $\mc F \in \mc B$ defines a new composite boundary condition.}
\label{line_insertion_fig}
\end{figure}
 
 Let's write $\LL$ for the monoidal category of line operators, and let $1_\LL$ be the monoidal unit: the trivial line operator.  There is a canonical monoidal functor $\eend_\LL(1_\LL)\text{-mod} \to \LL$ (for any monoidal category: this is a purely formal construction, see Proposition \ref{E_2_adjunction_prop}). We also have an $\bb E_2$-algebra equivalence $\eend_\LL(1_\LL) \to A$: the hom-spaces in the category $\LL$ of line operators are modelled by local operators inserted at the junction of two line operators (Kapustin explains this in \cite{KapustinICM}: see in particular his Figure 5), endomorphisms of the trivial line operator are therefore equivalent to local operators.  Therefore, there is a canonical monoidal functor 
 \[\mc A\text{-mod} \to \eend_\LL(1_\LL)\text{-mod} \to \LL,\]
 and therefore a canonical action of $\mc A\text{-mod}$ on the category $\mc B$ of boundary conditions.

 \item Alternatively there's a picture that motivates our construction purely on the level of algebras: the action of local observables is implemented by the ``whistle'' cobordism in open-closed TQFT.  Let's explain what we mean.  In two dimensions we consider the following situation (as depicted in Figure \ref{boundary_condition_fig}).
 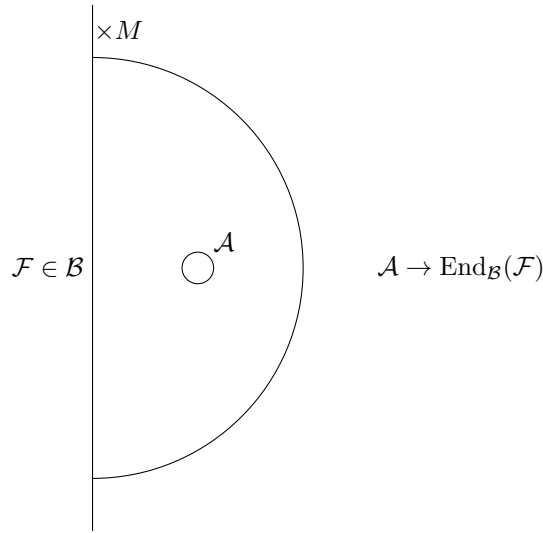
\begin{figure}[!hb]
 \centering
\begin{tikzpicture}[scale=.7]
 \node at (-0.85, 0) {$\mc F \in \mathcal B$};
 \clip (-.008, -5) rectangle (10,5);
 \draw (0,0) circle (4);
 \node at (0.5,4.5) {$\times M$};
 \draw(0,-5) -- (0,5);
 \draw (2,0) circle (.3);
 \node at (2.5,0.5) {$\mc A$};
 \node at (7, 0) {$\mc A \to \mathrm{End}_{\mathcal B}(\mathcal F)$};
\end{tikzpicture}
\caption{The action of $\mc A$ on $\mc B$ is mediated by a whistle cobordism with the given boundary condition.}
\label{boundary_condition_fig}
\end{figure}
Consider the manifold $M \times \RR \times \RR_{\ge 0}$ with a boundary condition $\mc F$ at the boundary, described by an object of $\mc B$.  Take the product with $M$ of a contractible open neighborhood $U$ in the half plane including an interval on the boundary, and choose a point in the interior of this patch.  This setup describes an ($\bb E_2$) algebra map from local observables to observables on the open patch $M \times U$.  So what is this latter algebra?  We will model observables on $M \times U$ with the boundary condition $\mc F$ by the algebra of endomorphisms $\mathrm{End}_{\mathcal B}(\mathcal F)$.  

In order to justify this claim, first, one can think of $\eend_{\mc B}(\mc F)$ as the \emph{quantum phase space} of our topological field theory coupled to the boundary condition $\mc F$.\footnote{This is called the phase space because it is the space of quantum states on a codimension 1 manifold $Y\times [0,1]$; that it has the structure of an algebra is an additional property arising from considering the same boundary condition at both ends.} Indeed, in general, the hom space $\hom_{\mc B}(\mc F_1, \mc F_2)$ can be interpreted as the space of states on a strip $Y^{n-2} \times [0,1]$ with boundary conditions $\mc F_1$ and $\mc F_2$ on the two boundary components.  In particular we obtain $\eend_{\mc B}(\mc F)$ by putting the boundary condition $\mc F$ at both boundary components. Alternatively we can view this as the space of observables on $Y$ times a half disk $D = \{(x,y) \colon x^2+y^2 < 1, x \geq 0\}$, with boundary condition $\mc F$ on the closed edge (as in Figure \ref{boundary_condition_fig}). That these descriptions -- the phase space and the space of observables -- are equivalent is a version of the state-operator correspondence for topological field theories.

Now, as we discussed in the first subsection, whenever we have an inclusion of open sets $U_1 \to U_2$ there's an associated map from local observables on $U_1$ to local observables on the larger open set $U_2$.  In particular, using the setup above, there's a map from observables on a contractible subset of $U \times M$ to observables on all of $U \times M$, i.e. a map $\mc A \to \mathrm{End}_{\mathcal B}(\mathcal F)$.  In fact this is a map of associative algebras, because it's compatible with this factorization structure in the vertical $\RR$ direction.  Finally, there is a universal property characterizing actions of an algebra $\mc A$ on a dg category $\mc B$: such an action is equivalent to maps $\mc A \to \mathrm{End}_{\mathcal B}(\mathcal F)$ for every object $\mc F \in \mc B$ (we'll discuss this in Section \ref{restriction_background_section} when we talk about Hochschild cochains).
\end{enumerate}
\vspace{-10pt}

With this story in hand, let's move on to discuss support conditions.  First of all, support conditions are controlled by closed subsets of an affine scheme we refer to as the moduli space of vacua of the topological field theory.  Let's briefly explain what we mean by that and justify the choice of name.  In general, a vacuum state of a quantum field theory can be modelled as a functional on the algebra of local observables $ \phi \colon \mr H^\bullet(\obsq(B^n)) \to k$ where $k$ is a ground field.  Specifically a \emph{vacuum state} is a translation invariant functional $\phi$ that satisfies the \emph{cluster decomposition property}. That is, for all local observables $\OO_1$ on $B_{r_1}(0)$ and $\OO_2$ on $B_{r_2}(0)$, we have
\[(\OO_1\ast \tau_x(\OO_2)) (\phi) - \OO_1(\phi) \OO_2(\phi) \to 0 \text{ as } x \to \infty\]
where $\tau_x$ denotes the translation of an observable by $x \in \RR^n$.  In the case where our theory is topological the condition of translation invariance is automatic, and the cluster decomposition property simply says that $\phi$ induces a ring homomorphism from the cohomology of local observables to $\RR$, or equivalently a closed point in $\mr{Spec}$ of the cohomology of local observables.

\begin{definition} \label{intro_vacua_def}
The \emph{moduli space of vacua} of the topological field theory with algebra $\mc A$ of local observables is the affine derived stack $\vac = \spec \mr H^\bullet(\mc A)$. 
\end{definition}

For an explanation of the term ``affine derived stack'', and the meaning of $\spec A$ where $A$ is a not necessarily connective differential graded algebra, see \cite[Section 3.2]{BZNLoopSpaces1} and \cite[Section 4.4]{LurieDAGVIII} (where the term ``coaffine'' is used).

\begin{remark}
Our definition of the moduli space of vacua depends not only on $\mr H^0(\mc A)$ -- the degree 0 local observables defined up to gauge equivalence -- but on observables of all degrees.  This typically includes strictly more information.   For instance, in the B-model with connected target $X$, if we only included $\mr H^0(\mc A)$ we would have a trivial underived space of vacua: $\spec \mr H^0(X) = \pt$.  If we include the full cohomology then the \emph{derived} moduli space of vacua captures the full cohomology ring of $X$.  In particular the underived space of vacua doesn't ``see'' the topology of $X$, but the derived space does.  More relevantly for this work, as we'll see, in the geometric Langlands A- and B-twisted theories, the underived moduli space of vacua is just a point, but the full, derived moduli space of vacua is a positive-dimensional space (as physically expected, see e.g. \cite[Section 2.11]{WittenNB}, in which the non-trivial finitely generated algebra of local operators of the B-twisted theory is discussed).

Furthermore, it is necessary to consider the full algebra of observables and not only $\mr H^0(\mc A)$ in order to see all of its expected structure.  Recall that in a topological theory the algebra of observables for a TQFT is expected to be an $\bb E_n$-algebra.  As its classical limit, one should obtain a shifted Poisson algebra -- with Poisson bracket of degree $n-1$ -- so, if $n \ge 2$, it can only be non-trivial when we include terms in non-zero degree.  This is discussed, for instance, in the recent paper \cite{BBZBDN}.
\end{remark}

As we discussed above, using the map $\mc A \to \mathrm{End}_{\mathcal B}(\mathcal F)$ for any boundary condition $\mc F$, we obtain a quasi-coherent sheaf over the moduli space $\vac$ of vacua.  This is the main construction that we'll use in this paper: specifically we'll study the \emph{support} of this sheaf, the set of points where the stalk (or the localization of the module) is non-zero, a closed subset of $\vac$.  In other words, if $v$ is not in the support of the module $\mathrm{End}_{\mathcal B}(\mathcal F)$ of bulk-boundary local observables, then formally locally -- near $v$ -- this module vanishes.  

\begin{remark}
We can spell the algebraic idea of support out a little more in terms of field theory in the following way.  If the module $\mathrm{End}_{\mathcal B}(\mathcal F)$ is supported at a point $v \in \vac$, then all the observables $\OO \in \mc A$ that map to zero \footnote{The same argument works if observables map to something nilpotent; these two notions lead to scheme-theoretic support and set-theoretic support, respectively. For their comparison, look at Section \ref{set_vs_scheme_support_rmk}.} in $\mathrm{End}_{\mathcal B}(\mathcal F)$ vanish at the point $v$, i.e land in the prime ideal $\mf p_v$ corresponding to $v$.  Equivalently, we can extend $v$ from a linear functional on $\mc A$ to a linear functional on the module $\mathrm{End}_{\mathcal B}(\mathcal F)$:
\[\xymatrix{
 \mf p_v \ar[r] &\mc A \ar[r] \ar[dr]^v & \mathrm{End}_{\mathcal B}(\mathcal F) \ar[d]^{v'} \\&&\CC,
}\]
since the difference between two choices of lift of an element of $\mathrm{End}_{\mathcal B}(\mathcal F)$ to $\mc A$ becomes an element of the ideal $\mf p_v$ where the functional $v$ vanishes.  In other words, the module $\mathrm{End}_{\mathcal B}(\mathcal F)$ being supported at $v$ means that the vacuum state $v$ can be extended to a vacuum state $v'$ of the bulk-boundary system with boundary condition $\mc F$.
\end{remark}

Thus, the idea is that one can consider the ``restriction'' of the category of boundary conditions to a point in the moduli space of vacua.  We say that a boundary condition is \emph{compatible} with a choice of vacuum if the algebra of bulk-boundary observables $\eend_{\mc B}(\mc F)$ is supported at the vacuum $v$.  From this point of view we can define a restricted \emph{category} of boundary conditions as the full subcategory of boundary conditions compatible with exactly one vacuum state $v$.  To summarize:

\begin{definition}
A boundary condition $\mc F \in \mc B$ is \emph{compatible} with a vacuum state $v \in \vac$ if the set-theoretic support of the $\mc A$-module $\eend_{\mc B}(\mc F)$ contains the point $v$.  The \emph{restricted} category of boundary conditions at the point $v$ is the full subcategory $\mc B_v$ of objects compatible precisely with the vacuum state $v$.
\end{definition}

Now, we'll apply this definition to the B-twist of $\mc N=4$ super Yang--Mills theory -- the topological twist associated to the spectral side of the geometric Langlands equivalence. The restriction of the category of boundary conditions has a mathematically clean interpretation in the case where the category $\mc B$ is the category $\IC(\mc X)$ of ind-coherent sheaves on a stack $\mc X$: one can characterize the sheaves that survive restriction to a vacuum as a condition on their singular support.

\begin{remark}
Singular support conditions are still interesting if we only look at the dg category of coherent, as opposed to ind-coherent, complexes of sheaves.  This would be closer to the story governing categories of boundary conditions in the B-model as usually presented in physics, for instance in \cite{Douglas, Sharpe}.  It is, however, more convenient to work with the ind-completion of this category for technical reasons: the theory of dg categories in \cite{GRvol1}, for instance, typically assumes this cocompleteness. This is harmless because the two constructions determine one another:\footnote{At least for nice enough (``QCA'') stacks, a class of stacks including $\flat_G(C)$.  See \cite[Theorem 3.3.5]{DrinfeldGaitsgory}.} the compact objects of the dg category ind-coherent sheaves recover coherent complexes, and ind-coherent complexes are realized as the ind-completion of the dg category of coherent complexes. 
\end{remark}

Our main result is the following, which gives a gauge theoretic explanation for the nilpotent singular support condition in geometric Langlands.
\begin{theorem} [Theorem \ref{main_theorem}]
The moduli space of vacua in the B-twisted $\mc N=4$ supersymmetric Yang--Mills theory with gauge group $G$ is equivalent to $\hh^*/W$, where $\hh$ is the Cartan subalgebra and $W$ is the Weyl group.  The algebra $\OO(\hh^*/W)$ canonically acts on the category $\IC(\flat_G(C))$, and the restriction to the vacuum $0$ is equivalent to the category $\IC_{\mc N_G}(\flat_G(C))$ of sheaves with nilpotent singular support.
\end{theorem}

It's natural to ask what happens for other points in $\hh^*/W$.  We can describe the categories compatible with more general vacua in terms of singular support conditions, and we conjecture that the categories obtained from this admit a very natural, symmetrical description.  An analysis of the geometry of the relevant singular support conditions motivates the following conjectural description.
\begin{conjecture}
The restriction of the category $\IC(\flat_G(C))$ to the point $v \in \hh^*/W$ is equivalent to the category $\IC_{\mc N_L}(\flat_L(C))$ where $L$ is the stabilizer of a semisimple lift of $v$ to $\mathfrak{g}^*$ in $G$. 
\end{conjecture}

We will give some pieces of evidence for the conjecture of both a mathematical and a physical nature.  In particular we can verify the conjecture for the curve $C=\bb P^1$.

\subsection{Summary of the Paper}  \label{summary}

We begin with Section \ref{boundary_vacua_section} where we provide the necessary background and state the basic definitions used in the rest of the paper.  The first subsection, Section \ref{motivation_and_background_section}, contains the background material we'll use concerning topological field theory, from the factorization and the functorial points of view.  In particular we explain how homotopical algebra connects with such models for topological field theory.  We then review the main ideas of support conditions for triangulated and dg categories from Benson--Iyengar--Krause \cite{BIK} and Arinkin--Gaitsgory \cite{ArinkinGaitsgory} in Section \ref{restriction_background_section}.  We particularly emphasize Arinkin and Gaitsgory's conditions for global complete intersection stacks, in which case singular support conditions admit a particularly nice description.  The proof of this result, however, is deferred to Appendix \ref{GCI_appendix}.  Finally in Section \ref{loc_and_vacuum_section} we introduce the moduli space of vacua and explain what it means to restrict a dg category to a point in this moduli space.  We'll use the factorization algebra model for our main examples, but we also explain an alternative construction using functorial field theories (which is less general but provides additional motivation).

In Section \ref{SS_section} we apply these ideas to the examples of Kapustin--Witten twisted $\mc N=4$ gauge theory, with a focus on the B-twisted theory.  We begin by reviewing our construction of these theories \cite{EY1} in Section \ref{twisting_section}, by the general procedure of the topological twist.  In Section \ref{geometry_of_flat_section} we recall from \cite{ArinkinGaitsgory} important facts about the moduli stack of flat $G$-bundles on a curve and the geometric Satake equivalence.  Then in Sections \ref{restricting_KW_theories_section} and \ref{degree_shift_section} we compute the moduli of vacua in the B-twisted theory, and calculate the restriction of the dg category $\IC(\flat_G(C))$ to the point $0$ in this moduli space, recovering the dg category of sheaves with nilpotent singular support.  This requires taking some care -- the relationship between the calculation of the moduli space of vacua and the twist itself is somewhat subtle.  We explain this issue in Section \ref{degree_shift_section}.

We conclude with Section \ref{consequences_section}, in which we investigate what happens when one moves away from $0$ in the moduli space of vacua.  In Section \ref{sym_breaking_section} we conjecture a natural answer -- that one obtains the nilpotent singular support category but with the gauge symmetry group broken to a Levi subgroup -- and provide some evidence for this conjecture.  Then finally in Section \ref{fact_categories_section} we conjecture that these categories should satisfy a natural factorization condition on the moduli space of vacua, and describe how this is related to factorization structures on cohomological Hall algebras as investigated by Kontsevich and Soibelman.

\subsection{Acknowledgements}
This project began with a suggestion of David Nadler and David Ben-Zvi, and would not have been possible without David Ben-Zvi generously sharing his ideas about the nature of vacua in the context of functorial quantum field theories.  We would like to thank Kevin Costello for many helpful conversations about all aspects of this work. We would also like to thank Dima Arinkin, Davide Gaiotto, Dennis Gaitsgory, Saul Glasman, Owen Gwilliam, Pavel Safronov, Claudia Scheimbauer, and Brian Williams for useful conversations and comments, and the anonymous referees for detailed comments and criticism.

This research was supported in part by Perimeter Institute for Theoretical Physics. Research at Perimeter Institute is supported by the Government of Canada through Industry Canada and by the Province of Ontario through the Ministry of Economic Development \& Innovation. CE acknowledges the support of IH\'ES.  The research of CE on this project has received funding from the European Research Council (ERC) under the European Union's Horizon 2020 research and innovation programme (QUASIFT grant agreement 677368). PY acknowledges the support of IH\'ES during his visit in 2017 and funding from ERC QUASIFT grant 677368.

\section{Boundary Conditions and Vacua} \label{boundary_vacua_section}

In this section we'll give an abstract definition, motivated by a construction in terms of locally constant factorization algebras and a parallel construction in functorial topological quantum field theory, of a support condition for a dg category $\mc C$ which we understand as the category of boundary conditions. Physically speaking, the support condition corresponds to compatibility with a choice of a point in the moduli space of vacua: each point in the moduli space yields an interesting full subcategory of $\mc C$ by picking out boundary conditions compatible with such a choice.

We'll begin by recalling  the factorization algebra model for quantum field theories developed by Costello and Gwilliam \cite{CostelloGwilliam1, CostelloGwilliam2} and the formalism of extended functorial QFTs.  After an exposition on these ideas, with an emphasis on the topological case, we'll then introduce the moduli space of vacua in each of these contexts, and explain what it means to restrict a dg category to a point in this space. 

\begin{remark}
We should emphasize that when we refer to ``topological'' quantum field theory we don't necessarily mean something equivalent to a functorial TQFT; we'll often have in mind something more general.  For instance, topological twists of supersymmetric field theories need not be topological in the strict mathematical sense, but can depend on additional geometric structure (see the example of Donaldson--Witten theory as discussed by Moore and Witten \cite{WittenMoore}, and geometric Langlands twisted $\mc N=4$ gauge theories as we discussed in our previous paper \cite{EY1}).  As we explained in the introduction, when we start making arguments we'll use the model of an $\bb E_n$-algebra of local observables acting on a fixed dg category of boundary conditions.
\end{remark}

\subsection{Factorization Algebras and Functorial Field Theory} \label{motivation_and_background_section}

In this subsection, we will describe mathematical models for quantum field theories that appear in the literature, with a focus towards the topological case.  We'll begin with reviewing the classical BV formalism for a Lagrangian description of a classical field theory from the modern perspective of derived geometry. Then we provide an exposition of functorial topological field theory and how one should think of it as arising from the data provided by the classical BV formalism.  We also explain the formalism of factorization algebras and mention a result of Costello and Gwilliam on constructing factorization algebras from such data. Finally we discuss how these two constructions are supposed to be compared for topological field theories. 

\subsubsection{The Classical BV Formalism} \label{BV_section}

We begin with some abstract motivation regarding the way in which derived geometry appears when one studies classical field theory; this perspective is called the classical BV formalism.  We'll explain the basic idea, and briefly explain the idea of a shifted symplectic structure, which naturually finds a home in the classical BV story.  We already outlined, in Remark \ref{derived_EOM_remark}, the idea of the \emph{derived critical locus} of an action functional.  That is, if $S$ is an action functional acting on a space $\Phi$ of fields on $M$, the space $\EOM(M)$ is defined as the derived intersection
\[\EOM(M) = \Phi \cap^h_{T^*\Phi} \mr{Graph}(\d S)\]
of the graph of $\d S$ with the zero section in the symplectic space $T^*\Phi$.  To pin things down a little more, let's suppose that $\Phi$ has the structure of an algebraic stack (for instance, the quotient of a vector space by an algebraic group of gauge transformations).  In the classical BV formalism one introduces a new piece of structure associcated to this derived intersection called a \emph{$(-1)$-shifted symplectic structure}.  In the physics literature this structure is called the \emph{antibracket}, but it fits into a more general formalism of shifted symplectic structures on derived stacks introduced by Pantev, To\"en, Vaqui\'e, and Vezzosi \cite{PTVV}.

The general idea behind shifted symplectic structures is the following.  Derived stacks admit (when they're sufficiently well-behaved: Artin stacks) a version of the cotangent space, and consequently a theory of differential forms.  If $\mc X$ is any derived Artin stack and $k$ is an integer one can define the $k$-shifted cotangent space $T^*[k]\mc X$ of $\mc X$.  This is a new derived stack with a map $\pi \colon T^*[k]\mc X \to \mc X$ down to $\mc X$, and with the property that the fiber $\pi^{-1}(x)$ over a closed point in $\mc X$ is the derived affine space (i.e. cochain complex) $\bb L_{\mc X,x}[k]$. Just as the cotangent spaces at $x$ glue to form a cotangent bundle, these complexes glue to form what is called the \emph{cotangent complex} $\bb L_{\mc X}$ shifted down in degree by $k$.  The \emph{tangent complex} $\bb T_{\mc X}$ is defined to be the linear dual of the cotangent complex.

Using this complex one can define a double complex of differential forms: one forms the sum over non-negative integers $p$ of the complexes $\wedge^p(\bb L_{\mc X})$.  This is bigraded by, on the one hand the internal grading on the cotangent complex, and on the other hand by the form degree $p$.  It's made into a double complex by defining a version of the \emph{de Rham differential} from $p$-forms to $(p+1)$-forms.  A $\d_{\dR}$-closed $p$-form of degree $k$ is an (inhomogeneous) element, which is closed for the total differential on this double complex, whose lowest degree term is a $p$-form of degree $k$. 

\begin{definition}
A \emph{$k$-shifted symplectic structure} on $\mc X$ is a section $\omega$ of the sheaf of $\d_{\dR}$-closed 2-forms of internal degree $k$ which is \emph{non-degenerate}, in the sense that $\omega$ induces an equivalence between the tangent complex $\bb T_{\mc X}$ and $k$-shifted cotangent complex $\bb L_{\mc X}[k]$.
\end{definition}

In order to bring this back into dialogue with the classical BV formalism, we introduce one of the main sources in \cite{PTVV} for symplectic structures.  This builds on the idea of a \emph{Lagrangian} $f \colon L \to \mc X$ in a $k$-shifted symplectic stack: a derived enhancement of the usual notion of a Lagrangian embedding (although the morphism $f$ doesn't have to be an embedding).
\begin{theorem}[{\cite[Theorem 2.9]{PTVV}}]
If $L_1$ and $L_2$ are two Lagrangians in a $k$-shifted symplectic derived stack $\mc X$, the derived intersection
\[L_1 \cap^h_{\mc X} L_2\]
is canonically $(k-1)$-shifted symplectic.
\end{theorem}

In particular, in our classical field theoretic story, the cotangent space $T^*\Phi$ is 0-shifted symplectic and the 0-section and the graph of a 1-form are both canonically Lagrangian, which means the derived intersection $\EOM(M)$ is canonically $(-1)$-shifted symplectic.

Now, a classical field theory is not just given by a single $(-1)$-shifted symplectic space that captures the moduli space of global solutions to the equations of motion, one can also consider solutions that are only defined locally on a subset of $M$. It is expected that for a codimension $k$ closed submanifold $Y$, there is a sense in which the moduli space of germs of solutions near the submanifold, which we denote by $\EOM(Y)$, has a canonical $(k-1)$-shifted symplectic structure. In fact, for a topological field theory, one expects this to be literally true.  We discussed this structure in the examples of geometric Langlands twisted theories in \cite{EY1}.  We will review some of these examples in Section \ref{twisting_section}.

An important class of examples of topological field theories capturing this additional data consists of those can be succinctly described, from the classical BV point of view, in terms of the AKSZ construction. These are classical field theories that can be described, with their $(-1)$-shifted symplectic structure, using the following fact. 

\begin{theorem}[{\cite[Theorem 2.5]{PTVV}}] \label{PTVV_AKSZ}
If $\mc X$ is an $n$-shifted symplectic derived stack and $Y$ is a ``compact $k$-oriented'' derived stack, meaning in particular that there is a degree $-k$ perfect pairing on the cochains $C^\bullet(Y, \OO_Y)$, then there is a canonical $(n-k)$-shifted symplectic structure on the mapping stack $\ul{\mr{Map}}(Y,\mc X)$.
\end{theorem}

In the examples below, we'll use two examples of compact $k$-oriented derived stacks.  Firstly, any compact oriented $k$-manifold $Y$ defines a compact $k$-oriented derived stack, the ``Betti stack'' denoted $Y_B$.  Secondly, any smooth compact Calabi--Yau variety $Y$ of dimension $k$ is, in particular, a compact $k$-oriented derived stack.

\begin{example}
\begin{itemize}
\item [(1)] (Topological classical mechanics) Let $X$ be an (ordinary) symplectic space, and let $Y$ be a compact oriented 1-manifold. Then topological classical mechanics with the phase space $X$ is described by the mapping stack $\EOM_{\mr{TCM}}(Y) = \underline{\mr{Map}}(Y_B, X)$.
\item [(2)] (B-model) Let $X$ be a smooth Calabi--Yau variety and let $Y$ be a Riemann surface. The B-model with target $X$ is described by the mapping space $\EOM_B(Y) =  \underline{\mr{Map}}( Y_B, T_{\mr{form}}^*[1]X)$ with target given by the 1-shifted cotangent space $T^*[1]X$. 
\item [(3)] (Chern--Simons theory) Let $G$ be a reductive group. Then Chern--Simons theory with gauge group $G$ is $\EOM_{\mr{CS}}(Y) =  \underline{\mr{Map}}( Y_B,   BG )$.
\end{itemize}
In the three examples, ($d$) = (1), (2) and (3) indicates a $d$-dimensional classical field theory, but one obtains $(m-1)$-shifted symplectic stacks by replacing $Y$ by a submanifold of codimension $m$. For instance, for the B-model, $Y = C $, $Y= S^1$, and $Y=\pt$ yield $(-1)$-shifted, 0-shifted, and 1-shifted symplectic spaces, respectively.
\end{example}

\begin{remark}
We've described these theories from the AKSZ point of view (as mapping spaces): these are rephrasings of the descriptions one would find elsewhere in the literature (for instance in terms of the action functional).
\end{remark}

In fact, more is expected.  A classical field theory on $M$ is supposed to define a sheaf of ``$(-1)$-shifted symplectic'' stacks on $M$ given by $U\mapsto \EOM(U)$. However, it is not easy to construct this sheaf in practice, because the notion of $(-1)$-shifted symplectic structure on an open set should be defined with respect to Verdier duality, rather than linear duality.\footnote{For instance, abelian Chern--Simons theory should have $(\Omega^\bullet(U), d_{\mr{dR}})$ on each open set $U$ as the space of fields, but $\Omega^1(U)$ is not linear dual to $\Omega^2(U)$.} This sheaf structure can however be realized in the context of perturbative field theories, as described in \cite{CostelloGwilliam2}.

\begin{remark}[$\bb E_n$- and $\bb P_n$-algebras] \label{E_n_P_n_remark}
For future reference, let us discuss the structure of the algebra of functions on a shifted symplectic space, and of the quantization of such an algebra. Just as the algebra of functions on a symplectic manifold has a canonical Poisson structure, the algebra of functions on an $n$-shifted symplectic space is a $\bb P_{n+1}$-algebra. Here a $\bb P_{n+1}$-algebra is a commutative differential graded algebra together with a Poisson bracket of cohomological degree $-n$. The numbering is made in such a way that $\bb P_d$-algebras are related to $d$-dimensional classical field theories. In particular, a $\bb P_1$-algebra recovers the usual notion of a Poisson algebra relevant to classical mechanics understood as 1-dimensional classical field theory.

Let us also briefly recall the notion of an \emph{$\bb E_n$-algebra}: a quantization of a $\bb P_n$-algebra; we'll discuss what we mean by quantization in Section \ref{comparing_background_section}.  An $\bb E_n$-algebra is an algebra $\mc A$  in cochain complexes over the operad of little $n$-disks.  Briefly, that means that for every embedding of $n$-balls, say $B_{r_1}(x_1) \sqcup \cdots B_{r_k}(x_k) \to B_r(x)$, there is a $k$-fold operation $A^{\otimes k} \to A$, where if $k=1$ the associated operation is a quasi-isomorphism, and where the operations are compatible with respect to composition.  The cohomology of an $\bb E_n$-algebra $\mc A$ has the structure of a $\bb P_n$-algebra -- in particular, it has a commutative product -- as long as $n \ge 2$; see Theorem \ref{thm:Cohen}.
\end{remark}

\subsubsection{Factorization Algebras} \label{fact_alg_background_section}
The principal approach to understanding quantum field theory taken in this article is motivated by the Heisenberg picture of quantum mechanics, where one has the algebra of operators as the main character. This is the subject of factorization algebras as developed by Costello and Gwilliam \cite{CostelloGwilliam1, CostelloGwilliam2}. Although the formalism of factorization algebras doesn't a priori encode the information of the Hilbert space of a quantum field theory, it works well in the generality of any quantum field theory admitting a Lagrangian description, contrary to the model of functorial field theory which doesn't yet have as firm a constructive foundation for non-topological field theories.  In this section we'll summarize some basic concepts from the theory of factorization algebras, referring the reader to Costello and Gwilliam for details.

For a manifold $M$, a \emph{prefactorization algebra} $\mc F$ on $M$ is an assignment of a cochain complex $\mc F(U)$ to every open subset $U \subset M$ together with some structure maps: if $U_1, \cdots, U_n \subset V$ are disjoint open subsets of $V$ for an open subset $V\subset M$, then we have a map $\mc F(U_1) \otimes \cdots \otimes \mc F(U_n) \rightarrow \mc F(V)$ satisfying a natural associativity condition that for $\amalg_{i\in I} U_i \subset \amalg_{j \in J} V_j \subset W$, the associated diagram
\[ \xymatrix{
\bigotimes_{i\in I} \mc F(U_i) \ar[r] \ar[dr] & \bigotimes_{j \in J} \mc F(V_j) \ar[d]\\
& \mc F(W)
}\]
commutes. A \emph{factorization algebra} is a prefactorization algebra satisfying a descent axiom, just like a sheaf is a presheaf with the gluing axiom satisfied. As we only need the notion for the sake of motivation, we won't articulate the technical definition here but take for granted that those prefactorization algebras arising from field theory are indeed factorization algebras.

One can construct factorization algebra structures from the data of a classical field theory by a version of deformation quantization.  Suppose we have a classical field theory given by an assignment $U \mapsto \EOM(U)$ for open subsets $U \subset M$. Then a collection of observations that only depend on behavior of those solutions to the Euler--Lagrange equations on $U\subset M$ is by definition a function on the space $\EOM(U)$, which we denote by $\mr{Obs}^{\mr{cl}}(U)$ and refer to as the space of classical observables on $U$. As $\EOM (U)$ is $(-1)$-shifted symplectic, $\mr{Obs}^{ \mr{cl} }(U)$ has the structure of a $\bb P_0$-algebra, namely, a commutative algebra with a Poisson bracket of cohomological degree $1$. The assignment $U \mapsto \mr{Obs}^{\mr{cl} }(U)$ is a cosheaf of $\bb P_0$-algebras. In particular, this forms a factorization algebra valued in $\bb P_0$-algebras.

What's more interesting is the algebra of quantum observables. In the quantum world, given an open subset $U$ of a spacetime manifold $M$, the space $\mr{Obs}^{\mr{q}}(U)$ of quantum observables is never a commutative algebra. However, one can still make observations happening at \emph{different points} of spacetime, leading to the structure maps of a factorization algebra $\obsq (U_1) \otimes \cdots \otimes \obsq (U_n) \rightarrow \obsq(V)$ for \emph{disjoint} open subsets $U_1, \cdots, U_n \subset V$ for an open subset $V\subset M$. These are encoding the \emph{operator product} of observables. We'll write the map \[ \mr{Obs}^{\mr{q}}(U_1) \otimes \mr{Obs}^{\mr{q}}(U_2)\to \mr{Obs}^{\mr{q}}(V)\] as $O_1 \otimes O_2 \mapsto O_1 \ast O_2$.

One has to further specify the relation between $\obscl$ and $\obsq$ in the quantization process. Recall the usual deformation quantization problem that the observables of classical mechanics form a Poisson algebra, or a $\bb P_1$-algebra, $(A, \{ \;,\;\})$ and we want to quantize it to obtain the observables of quantum mechanics realized as an associative algebra, or an $\bb E_1$-algebra, $(\widetilde{A}, \ast)$ over $\RR[[\hbar]]$ such that there is an isomorphism $\widetilde{A} \otimes_{\RR[[\hbar]]} \RR \simeq A$ as associative algebras and for $a,b\in A$, one has \[\lim_{\hbar\to 0} \frac{a\ast b - b\ast a}{\hbar} = \{a,b\}.\] Likewise, a quantization of $\bb P_0$-algebra $A$ is an $\bb E_0$-algebra $\widetilde{A}$ over $\RR[[\hbar]]$ such that $\widetilde{A}\otimes_{\RR [[\hbar]]} \RR \simeq A$ and the induced bracket from the $\bb E_0$ structure coincides with the given $\bb P_0$ structure (for more detail and background, see \cite[Chapter 2]{CostelloGwilliam2}). 

\begin{theorem}[Costello--Gwilliam]\cite{CostelloGwilliam2}
For a classical field theory $\obscl$ and a choice of BV quantization, there exists a factorization algebra $\obsq$ over $\RR[[\hbar]]$ so that for every open subset $U\subset M$, $\obsq(U)$ is a quantization of the $\bb P_0$-algebra $\obscl(U)$ in the above sense.
\end{theorem} 

One constructs such a factorization algebra quantizing a classical field theory by studying the obstruction-deformation complex associated to the given action functional.

\begin{remark} 
This theorem is formulated and proved in the context of \emph{perturbative} field theory, which is why we have $\hbar$ as a formal parameter as opposed to a number. Accordingly, what we write as $\EOM$ should be understood as a sheaf (on the spacetime manifold $M$) of $(-1)$-shifted symplectic \emph{formal moduli} spaces and $\obscl$ as a cosheaf of certain Chevalley--Eilenberg cochain complexes of \emph{$L_\infty$-algebras}. Moreover, there are technical issues to be dealt with coming from the infinite-dimensionality of the spaces involved, which can be done by using the Wilsonian effective field theory philosophy and renormalization techniques following Costello \cite{CostelloBook}.
\end{remark}

Now we want to understand some examples of factorization algebras.

The easiest class of examples is the class of locally constant factorization algebras.  A factorization algebra $\mc F$ is called \textit{locally constant} if, for every inclusion of disks $D \hookrightarrow D'$, the induced structure map $\mc F(D) \rightarrow \mc F(D')$ is a quasi-isomorphism. A theorem due to Lurie \cite{LurieCobordism} classifies such factorization algebras.

\begin{theorem}[Lurie]
A locally constant factorization algebra on $\bb R^n$ has a natural structure of $\bb E_n$-algebra.
\end{theorem}

In particular, as the factorization algebra of a topological field theory is by definition locally constant, an $n$-dimensional topological field theory yields an $\bb E_n$-algebra as its algebra of observables.

The formalism of factorization algebras can capture another interesting class of examples. If we have a factorization algebra $\mc F$ on $\CC$ which is holomorphically translation-invariant, that is, we suppose that $\mc F(D(z,r)) \simeq \mc F(D(0,r))$ holds for all $z\in \CC$ and $r \in \RR_{>0}$ and that the vector field $\frac{\del}{\del \overline{z} }$ acts homotopically trivially on $\mc F$, then under certain technical conditions the cohomology of $\mc F$ yields a vertex algebra (for more detail, see \cite{GwilliamThesis} and \cite[Chapter 5]{CostelloGwilliam1}). For instance, nicely blending this framework with Gelfand--Kazhdan formal geometry, Gorbounov, Gwilliam, and Williams constructed a sheaf of vertex algebras known as chiral differential operators on $X$ from the classical field theory data of the curved $\beta\gamma$-system with target $X$ \cite{GGW}.

\begin{remark}
In this paper we'll use the phrase ``local observables'' in a TQFT to refer to the complex the factorization algebra of observables assigns to a contractible open set.  This is a purely perturbative notion, in the sense that it is only sensitive to the formal neighborhood of a point in the moduli space of solutions to the equations of motion, and not to the global geometry of this moduli space (for instance it won't be able to see the difference between the stack $\EOM(B^n)$ of local solutions to the equations of motion and its affinization).  The phrase ``local operators'' of a TQFT can sometimes refer to something more general that also depends on ``non-perturbative'' data.
\end{remark}

\subsubsection{Functorial Field Theories} \label{functorial_background_section}
An alternative approach to the formalisation of quantum field theory is given by generalizing the Schr\"odinger picture of quantum mechanics, where we have the Hilbert space and time evolution operator as the main characters, by identifying quantum mechanics as the one-dimensional case of quantum field theory. That is, one assigns to a codimension one submanifold of spacetime its space of states, and to a bordism between such submanifolds the time evolution operator. Realizing the idea rigorously in complete generality is a daunting task, but there have been extensive developments for those theories that depend only on the topology of the spacetime: a definition of functorial TQFT was first proposed by Atiyah \cite{AtiyahTQFT}, following Segal's study on conformal field theory \cite{SegalConformal}, and later extended to incorporate boundary conditions for states along codimension two manifolds and to higher order boundary conditions, notably by Lurie \cite{LurieCobordism} following ideas developed by Freed, Quinn \cite{FreedQuinn, FreedExtended, FreedHigher} and Baez--Dolan \cite{BaezDolan}.

We'll consider the following definition of an $n$-dimensional functorial TQFT incorporating categories of boundary conditions:
\begin{definition} \label{2_extended_TQFT_def}
A 2-extended \emph{$n$-dimensional TQFT} is a symmetric monoidal functor of 2-categories
\[Z \colon \bord^2_n \to \mr{dgCat}\]
where $\bord^2_n$ is a 2-category of $n$-dimensional bordisms (originally defined by Schommer-Pries \cite{SchommerPriesThesis} as a symmetric monoidal bicategory; Later Calaque and Scheimbauer \cite{CalaqueScheimbauerBordisms} constructed a fully extended bordism category as an $(\infty,n)$-category whose truncation to an $(\infty,2)$-category is homotopy equivalent to Schommer-Pries's construction \cite[Proposition 8.17]{CalaqueScheimbauerBordisms}).
\end{definition}

\begin{remark} \label{motivation_dg_remark}
It will be useful to briefly recall the motivation for this functorial definition. First of all, according to this motivation, an $n$-dimensional TQFT $Z$ assigns to a closed $(n-1)$-manifold $M$ a cochain complex\footnote{From a physical point of view we can think of the differential graded structure here as arising from the BRST formalism: the space of states is graded by ghost number and the BRST operator acts as a differential.} which we regard as the space of quantum states on $M$. Now, the claim is that for a closed $(n-2)$-manifold $Y$, the assigned dg category $Z(Y)$ can naturally be interpreted as the category of boundary conditions along $Y$. In order to justify this, note that for a spacetime of form $M = Y \times [0,1]$, one cannot make sense of $Z(M)$ directly as the cochain complex of quantum states on $M$ as $M$ is not closed, but it does admit an interpretation as the space of quantum states on $M$ that satisfy \emph{boundary conditions} along the two boundary components $Y\times \{0\}$ and $Y\times \{1\}$. Namely, for such a pair of boundary conditions $\mc F_1$ and $\mc F_2$ (i.e. objects of the category $Z(Y)$), the space of compatible quantum states, which we write as $\mc H_{\mc F_1,\mc F_2} := \mr{Hom}_{Z(Y)}(\mc F_1,\mc F_2) $, is a cochain complex. Moreover, the composition map $\mc H_{\mc F_1,\mc F_2} \otimes \mc H_{\mc F_2,\mc F_3} \to \mc H_{\mc F_1,\mc F_3}$, which is a cochain map by the axioms defining a dg category, tells us how to glue a state in $\mc H_{\mc F_1,\mc F_2}$ and a state in $\mc H_{\mc F_2,\mc F_3}$ together to obtain a state on the manifold $Y \times [0,2]$ satisfying boundary conditions $\mc F_1$ and $\mc F_3$, using the fact that the spaces of states on $Y \times [0,1]$ and on $Y \times [0,2]$ are canonically equivalent in a TQFT. Furthermore, the axioms tell us that in each space of states $\mc H_{\mc F_1,\mc F_1}$ there's a ``unit state'', so that gluing it onto a state in $\mc H_{\mc F_1,\mc F_2}$ has no effect. This is why a dg category appearing in a topological field theory is modeling boundary conditions of a theory.
\end{remark}

\begin{example}
Consider a specific 2-dimensional quantum field theory: the B-model with the target $X$ a Calabi--Yau manifold.  A standard model for the category of boundary conditions in the B-model is $Z(\pt) = \mr{Coh}(X)$ -- the derived category of coherent sheaves on $X$ \cite{Sharpe}. The cochain complex of quantum states on an interval $I$ with boundary conditions $\mc F_1, \mc F_2 \in \mr{Coh}(X)$ is $\mc H_{\mc F_1, \mc F_2} = \mr{Hom}(\mc F_1, \mc F_2)$. This is also known as the space of open string states with the given boundary conditions.
\end{example}

In particular we can consider the dg category such a functorial TQFT assigns to the $(n-2)$-sphere.  This dg category $\mc L : = Z(S^{n-2}) $ is an $\bb E_{n-1}$-algebra object in dgCat by applying the functor to $n$-dimensional pairs of pants. We call the $\bb E_{n-1}$-monoidal dg category $\mc L = Z(S^{n-2})$ the category of \emph{line operators} of the theory $Z$. For an explanation of why this deserves its name, look at Section 7 of the recent paper of Beem--Ben-Zvi--Bullimore--Dimofte--Neitzke \cite{BBZBDN}. To be more explicit about the $\bb E_{n-1}$-algebra structure, consider an $(n-1)$-manifold of form $M = B_1(0) \bs (B_\eps(x_1) \sqcup \cdots \sqcup B_\eps(x_k))$, where $B_r(x)$ is an $(n-1)$-dimensional ball of radius $r$ around the point $x \in \RR^{n-1}$ and $\eps$ is sufficiently small. By definition, it defines a functor
\[Z(M) \colon \mc L^{\otimes k} \to \mc L.\] 
As the definition of this functor depended on a point in a configuration space of $(n-1)$-balls, we obtain a $\bb E_{n-1}$-algebra structure. 

One important feature possessed by the line operators is that they act on the dg category of boundary conditions along an $(n-2)$-manifold $Y$.  Indeed, the bordism $(Y \times [0,1]) \bs B_\eps(y,t)$ defines a functor $Z(S^{n-2}) \otimes Z(Y) \to Z(Y)$, or equivalently a monoidal functor $\mc L \to \mr{Fun}(Z(Y),Z(Y))$, where $\mr{Fun}(\mc C, \mc C)$ is the category of functors from $\mc C$ to itself with the composition monoidal structure. Indeed, the embedding of the ball $B_\eps(y,t)$ into $Y \times [0,1]$ picks out a particular $\bb E_1$-category, i.e. monoidal category, structure on $\LL$ in the direction of $[0,1]$ so that $Z(Y)$ is made into a module for $\LL$ as a monoidal category.

\begin{remark}
The category of line operators can become interesting for $n \geq 3$; if $n=2$, the monoidal action just amounts to 
\[\mr{id} \colon \mr{Fun}( Z(\pt) , Z(\pt) ) \to  \mr{Fun}( Z(\pt) , Z(\pt) )\] 
for $Z(S^0) = Z(\pt \sqcup \pt^* ) =  \mr{Fun}( Z(\pt) , Z(\pt) )$ and $Y=\pt$.
\end{remark}

An even more fundamental feature of a quantum field theory is its the algebra of local operators.  The \emph{local operators} in a TQFT $Z$ are elements of the cochain complex $Z(S^{n-1})$.  By $n$-dimensional bordisms analogous to those we used to define the $\bb E_{n-1}$ structure on the category of line operators, this complex becomes an $\bb E_n$-algebra.

In the case $n=2$, there is a whistle (open-closed) bordism from $S^1$ to $I$: see Figure \ref{whistle_fig}. Note that this bordism is actually depicted in the Introduction as Figure \ref{boundary_condition_fig}, if one sets $M=\pt$.  We interpreted the figure a little differently in that discussion, thinking of the factorization structure on operators instead of the time-evolution of states encoded by the functorial TQFT.  One should think of these stories as related by a state-operator correspondence.

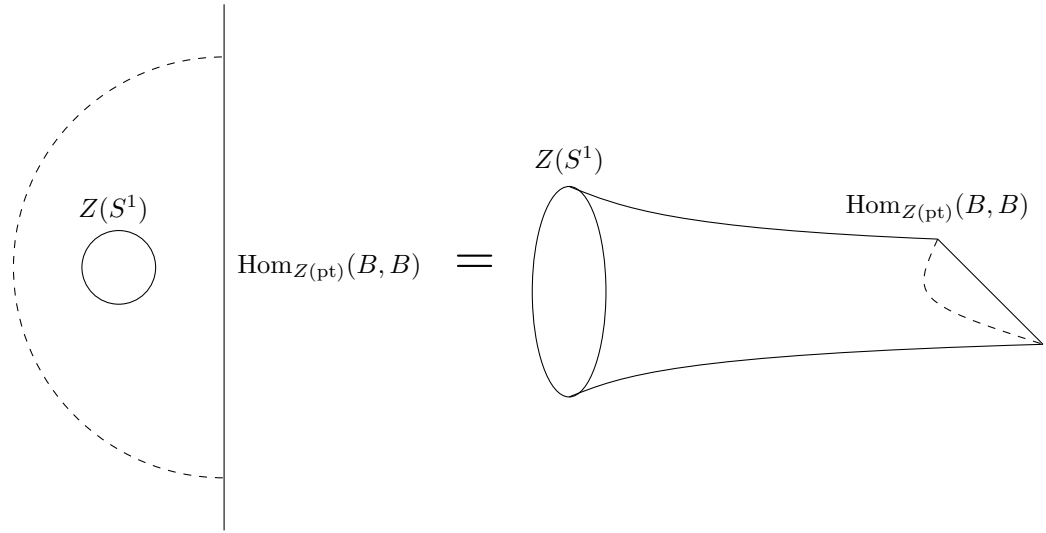
\begin{figure} [!h]
 \centering
 \begin{minipage}{.3\textwidth}
 \begin{tikzpicture}[scale=.7]
 \node at (2, 0) {$\hom_{Z(\pt)}(B, B) $};
 \clip (.008, -5) rectangle (-5,5);
 \draw[dashed] (0,0) circle (4);
 \draw(0,-5) -- (0,5);
 \draw (-2,0) circle (.7);
 \node at (-2.1,1.1) {$Z(S^1)$};
\end{tikzpicture}
\end{minipage}\hspace{1.4cm}\scalebox{2}{$=$}\hspace{.4cm}%
 \begin{minipage}{.3\textwidth}
\begin{tikzpicture}[scale=.7]
 \draw (1,0) ellipse (.7 and 2);
 \draw (8,1) -- (10,-1);
 \draw (1,2) ..  controls (2,1.6) and (3,1.2) .. (8,1);
 \draw (1,-2) ..  controls (2,-1.6) and (3,-1.2) .. (10,-1);
 \draw[dashed] (8,1) .. controls (7.5,-.1) and (7.5,-.2) .. (10,-1);
 \node at (1,2.5) {$Z(S^1)$};
 \node at (8,1.6) {$\hom_{Z(\pt)}(B, B)$};
\end{tikzpicture}
\end{minipage}
\caption{A pair of equivalent perspectives on the whistle cobordism in open-closed field theory defines a morphism $Z(S^1)\to \hom_{Z(\pt)}(B, B)$.}
\label{whistle_fig}
\end{figure}

Hence, for any choice of a boundary condition, say $B \in Z(\pt)$, the whistle gives rise to a morphism $Z(S^1)\to \hom_{Z(\pt)}(B, B)$. Then, in arbitrary dimension, the product of the whistle with an $(n-2)$-manifold $Y$ gives a morphism $Z(S^1 \times Y) \to  \hom_{Z(Y)}(\mc F, \mc F)$ for each $\mc F \in Z(Y)$. Composing with the complement of a point $\{0\} \times \{y\} \in B^2 \times Y$ understood as a bordism from $S^{n-1}$ to $S^1 \times Y$, an $n$-dimensional TQFT provides us with an operator of the form 
\[\alpha_{y,\mc F} \colon Z(S^{n-1}) \to \hom_{Z(Y)}(\mc F, \mc F)\] 
for each choice of a point $y \in Y$ and a boundary condition $\mc F$ along $Y$. This makes the space $\hom_{Z(Y)}(\mc F, \mc F)$ into a module for $Z(S^{n-1})$. The physical interpretation of this structure is that the local operators in the bulk theory, i.e. operators supported away from the boundary, act on states on $Y\times [0,1]$ satisfying a fixed boundary condition without altering it. From the construction, the action is through a particular $\bb E_2$-algebra structure which a priori depends on the choice $y\in Y$.

Finally let us explain one perspective on functorial field theory from the perspective of the classical BV formalism (not the only possible perspective, we should mention for instance the work of Cattaneo, Mn\"ev, and Reshetikhin \cite{CMR}, which takes a different point of view). In one sentence, one obtains a functorial field theory through (categorified) geometric quantization of a global description of the solutions to the classical equations of motion. This way of thinking may not be explicitly written in the literature, but it is implicit in many works as cited below.

Recall that in the classical BV formalism for a codimension 1 submanifold, we would obtain a (0-shifted) symplectic structure, as expected for a phase space. The Hilbert space is given by the geometric quantization of this space. This process of geometric quantization is not well-understood in this generality, but when $\EOM(N) = T^* \mc X_N$ is of the form of a cotangent bundle it admits a standard choice of geometric quantization by taking functions on the base, namely, $\OO(\mc X_N)$. Similarly, for a codimension 2 manifold $Y$, if we had $\EOM(Y) = T^*[1] \mc X_Y$, then its categorified geometric quantization can be taken to be sheaves on the base, namely, $\mr{Coh}(\mc X_Y)$, or better, its ind-completion $\IC(\mc X_Y)$ as the \emph{Hilbert} category: see Wallbridge \cite{Wallbridge}, Calaque \cite{Calaque} in the case of AKSZ type theories with Betti source, and forthcoming work of Safronov.

\begin{example}[B-model]
Recall that the B-model with target Calabi--Yau variety $X$ is described by $\EOM_B(Y) = \underline{\mr{Map}}(Y_{B} , T_{\mr{form}}^*[1] X )$. For a codimension 1 manifold $S^1$, we obtain $\EOM_B(S^1) = \LL( T_{\mr{form}}^*[1] X ) = T_{\mr{form}}^*(\LL X )$, where $\LL X$ is the derived loop space: the derived fiber product $X \times_{X \times X} X$. The Hilbert space $Z(S^1)$ is given by its geometric quantization, which is $\OO(\LL X) = \mr{PV}(X)$, the algebra of polyvector fields on $X$. For a codimension 2 manifold $\pt$, we have $\EOM_B(\pt) = T_{\mr{form}}^*[1] X$ which yields $Z(\pt)  = \mr{Coh}(X)$.  Note that by the cobordism hypothesis $Z(S^1)$ is determined by $Z(\pt)$, i.e. $Z(S^1) = \mr{HC}^\bullet(X)$. As a result, this example recovers some of the fundamental relationships between loop spaces, polyvector fields, and Hochschild (co)chains in the context of TQFT. For a different, but related, take on why this formal process encodes physically meaningful operations, see Sections 1 and 4 of \cite{BBZBDN}.  \footnote{For the purpose of this paper, the essential examples end up being versions of the B-model. The category of boundary conditions of B-model with target $X$ is usually modelled by $\coh(X)$; without loss or gain of information we can alternatively consider the category $\indcoh(X)$. One may just use this as the category of boundary conditions, without relying on the framework of categorified geometric quantization.}
\end{example}

\begin{remark}\label{rmk:de Rham algebra}
While this is perhaps the most convenient formalism in which to understand the category of boundary conditions and the action of $Z(S^{n-1})$ on it, we'll be interested principally in examples which it doesn't fully capture.  The reason is that the functorial TQFT formalism is not capable of capturing the de Rham algebraic structure of geometric Langlands twisted theories that we discussed in our previous paper \cite{EY1}. In other words, from the functorial TQFT perspective, it is natural to expect that one assigns a version of $D(\bun_G(C))$ or $\IC(\flat_G(C))$ as $Z(C)$ for a 4-dimensional TQFT $Z$ for a smooth compact \emph{topological} surface $C$, but our description of geometric Langlands gauge theories does depend on the choice of \emph{complex} structure of a smooth proper curve $C$.  We'll therefore consider TQFTs modelled by the factorization algebra description, along with the further data of a dg category $Z(M)$ modelling boundary conditions on a codimension 2 manifold $M$, equipped with an action of the factorization algebra.
\end{remark}

\subsubsection{Comparisons for Topological Field Theory} \label{comparing_background_section}
Now we would like to compare the two formalisms for topological field theory, starting from the Lagrangian data defining a classical field theory. An issue is that for constructing a functorial field theory it is essential to have a description of the full moduli space of solutions to the classical equations of motion (including its global geometry), whereas the formalism of factorization algebras has been developed only in a perturbative context, in the sense that it only depends on the quantization of the formal neighborhood of a fixed point in the moduli space of solutions. However, in the context of \emph{topological} field theories, there is a shortcut that allows us to construct global quantizations.

We compare the two perspectives using the \emph{state-operator correspondence}. The idea is that for an $n$-dimensional classical topological field theory on $M$, its associated functorial field theory $Z \colon \mr{Bord}_n\to \mr{Vect}_n$ and factorization algebra $\obsq$ satisfy the relation $\obsq (U) \simeq Z(\del U)$ for an open ball $U$ that is a subset of $M$, understood as an $n$-manifold. In order to say something more precise we will recall the following general definition of the space of states in a factorization algebra, following Costello and Gwilliam \cite[Section 4.9]{CostelloGwilliam1} (here the word state is used in the sense of operator algebras).

\begin{definition} \label{states_def}
Let $\obs^\mr{q}$ be a translation-invariant factorization algebra on $\RR^n$, over a ring $R$ (for instance, $R=k$ or $k[[\hbar]]$ for $k=\bb R$ or $\bb C$).  A \emph{state} for $\obs^\mr{q}$ is a graded linear map of degree 0 $\langle - \rangle \colon \mr H^\bullet(\obs^\mr{q}(\RR^n)) \to R$.  In a topological field theory, one can replace $\RR^n$ by an open ball $B^n$; there is a canonical equivalence between the spaces of observables on $B^n$ and on $\RR^n$. 
\end{definition}

Having this notion in mind, let us discuss what we should do in order to describe the algebra of local quantum observables starting from the global data of a classical field theory $U\mapsto \EOM(U)$. The space of classical local observables is just the space of functions on the global moduli space $\EOM(\RR^n)$, which forms a $\bb P_0$-factorization algebra. If the theory in question is an $n$-dimensional topological field theory, then its factorization structure is that of an $\bb E_n$-algebra via local constancy. From the Poisson additivity result $\bb P_0\otimes \bb E_n =\bb P_n$ (first announced by Rozenblyum, but not yet publically available -- a different proof is due to Safronov \cite{Safronov}, c.f. the Dunn additivity result $\bb E_1\otimes\bb E_n= \bb E_{1+n}$ \cite{Dunn, HigherAlgebra}), this structure is equivalent to the structure of a $\bb P_n$-algebra. 
 
Now, the process of BV quantization promotes the $\bb P_0$-algebra structure on $\obscl(U)$ to an $\bb E_0$-algebra structure on $\obsq(U)$ for each open set $U\subset M$. Accordingly, for an $n$-dimensional topological field theory, the factorization algebra is described by an $\bb E_n$-algebra. Conceptually speaking, this $\bb E_n$-algebra structure should be equivalent to the one coming from the identification $Z(S^{n-1})^* = \obsq(B^n)$ (see Remark \ref{state_operator_E_n_comparison_remark} for more discussion of this point).

In this sense, one should be able to describe the BV quantization problem for (nonperturbative) $n$-dimensional topological field theory in terms of the quantization problem going from a $\bb P_n$-algebra to an $\bb E_n$-algebra. In order to elucidate this translation, let us try to discuss their relationship.

For an $\bb E_n$-algebra $\mc A$, by definition we have a map $\mr{Emb}(\coprod_I B^n,B^n ) \times \mc A^I \rightarrow \mc A$. For $I=\{1,2\}$, we have a map $S^{n-1} \rightarrow \mr{Emb}( B^n \coprod B^n,B^n )$ by considering the first disk fixed at the origin. This gives a map $S^{n-1} \times \mc A^2 \rightarrow\mc A$. Taking cohomology, we get a map $\mr H^\bullet(S^{n-1}) \otimes \mr{H}^\bullet(\mc A)^{\otimes 2} \rightarrow \mr{H}^\bullet (\mc A)$. Thinking of the nontrivial class in $\mr H^{n-1}(S^{n-1})$, we have a map $\mr H^\bullet(\mc A)^{\otimes 2}[n-1] \rightarrow \mr{H}^\bullet(\mc A)$, or $(\mr H^\bullet(\mc A)[n-1] )^{\otimes 2} \rightarrow \mr{H}^\bullet(\mc A)[n-1]$.

\begin{theorem}[Cohen \cite{Cohen}]\label{thm:Cohen}
Let $\mc A$ be an $\bb E_n$-algebra. Then the above map on $\mr H^\bullet(\mc A)$ induces a Lie bracket of degree $1-n$ on $\mr H^\bullet(\mc A)$. Moreover, if $n \ge 2$, then $\mr H^\bullet(\mc A)$ is a $\bb P_n$-algebra.
\end{theorem}

Surprisingly, if $n$ is at least 2, then it's possible to go the other way: lifting $\bb P_n$-algebra structures to $\bb E_n$-algebra structures. This follows from the formality of the $\bb E_n$ operad, which is a theorem of Tamarkin \cite{Tamarkin} for $n=2$ and of Kontsevich \cite{KontsevichOperadsMotives} and Lambrechts--Voli\'c \cite{LambrechtsVolic} for $n\geq 3$. For our purpose, we will use the following theorem concerning the spaces of deformations of an $\bb E_{n}$ and of a $\bb P_{n}$ structure, which is based on the formality result.

\begin{theorem}\cite[Corollary 5.4]{ToenBranes} \label{thm:Toen}
For $n\geq 2$, if $X$ is a derived stack, then the dg Lie algebra $C^{\bb E_{n}}(X)[n]$ is non-canonically quasi-isomorphic to the dg Lie algebra $\mr{Pol}(X,n-1)[n]$. The quasi-isomorphism depends on the choice of a formality equivalence $\bb E_{n}\simeq \bb P_{n}$.
\end{theorem}
 
Here $C^{\bb E_{n}}(X)$ is the $\bb E_{n}$-Hochschild cochain complex, which is an $\bb E_{n+1}$-algebra, and hence $C^{\bb E_{n}}[n]$ has the structure of a dg Lie algebra. If $X=\spec A$, this controls the deformations of the $\bb E_{n}$-algebra structure on $A$. For the right-hand side, $\mr{Pol}(X,n-1) := \OO( T^* [n]X )$ is a $\bb P_{n+1}$-algebra and hence $\mr{Pol}(X,n-1)[n]$ has the structure of a dg Lie algebra controlling the deformations of $\bb P_{n}$-algebra structure on $A$.

The only case of importance for our arguments will be the following simple situation. Suppose a $\bb P_n$-algebra $A^{\mr{cl}}$ comes from a commutative algebra structure with a trivial Poisson bracket, and that it does not admit any deformations as a $\bb P_{n}$-algebra. Then it doesn't admit any $\bb E_{n}$-algebra deformations either, i.e. $A^{\mr{cl}}$ does not admit a non-trivial quantization.  In other words, any quantization $A^{\mr{q}}$ is homotopic to the $\bb E_{n}$-algebra whose structure is the trivial one coming from the commutative algebra structure.  In particular, as commutative algebras, $\mr H^\bullet(A^{\mr q})$ and $A^{\mr{cl}}$ are necessarily quasi-isomorphic. This is the claim that we'll want to use, in a specific example with $n=4$ (see Section \ref{restricting_KW_theories_section}, and in particular Proposition \ref{no_deformations_prop}).

\begin{remark}[Quantum mechanics and deformation quantization]
The original deformation quantization problem asks if, given a Poisson manifold $M$, one can promote the $\bb P_1$-algebra $(C^\infty(M), \{-,-\})$ to an associative algebra $(C^\infty(M)[[\hbar]],\ast)$, or $\bb E_1$-algebra, where $\ast$ is an $\hbar$-linear associative product satisfying $\lim_{\hbar\to 0}\frac{a\ast b - b\ast a}{\hbar } = \{a,b\}$ for $a,b\in C^\infty(M)$. A Poisson structure on $M$ corresponds to a Maurer--Cartan solution of the dg Lie algebra $\mr{Pol}_{C^\infty}(M,0)[1]$. The original formality theorem of \cite{KontsevichPoisson} provides a homotopy equivalence of dg Lie algebras $C^{E_1}(C^\infty(M))[1] \simeq \mr{Pol}_{C^\infty}(M ,0)[1]$ which solves the relevant deformation quantization problem. Note that here we're using the notation $\mr{Pol}_{C^\infty}(M,0)$ to denote the analogue of $\mr{Pol}(X,0)$ for a smooth manifold $M$ instead of a derived stack $X$; for more detail one should refer to the original paper.
\end{remark}

The claim is that from the data of a classical field theory $U\mapsto \EOM(U)$, one can construct the algebra of local operators as well as the dg category of boundary conditions. That is, functions on $\EOM(\pt)$ form a $\bb P_n$-algebra and its quantization forms an $\bb E_n$-algebra capturing the factorization algebra structure of local quantum observables. In codimension 2, $\EOM(Y^{n-2} )$ is 1-shifted symplectic, and we obtain the dg category of boundary conditions by categorified geometric quantization (in suitable examples).  Our main examples -- topologically twisted $\mc N=4$ super Yang--Mills -- yield examples of categories with $\bb E_2$-algebra actions by exactly this procedure.

\subsection{Support Conditions in dg Categories} \label{restriction_background_section}
In this section we'll recall some of the main constructions in the work of Benson, Iyengar, and Krause \cite{BIK} and of Arinkin and Gaitsgory \cite{ArinkinGaitsgory}.  We'll explain, as we alluded to in Section \ref{construction_intro_section} in the introduction what it means for an $\bb E_2$-algebra to act on a dg category, and how to think about the support of an object with respect to this action.  We'll begin by introducing, first informally and then in more detail, the notion of the \emph{scheme of singularities} of a derived stack.  We'll then analyse this notion in an important family of examples.

We should warn the reader that this subsection contains some details which are undeniably technical.  We suggest reading the motivational Section \ref{SS_motivation_section} first, then only coming back to the rest of this section in order to understand the proof of the main Theorem \ref{main_theorem}.  The most technical details occur in the proof of Proposition \ref{SS_complete_intersection_prop}; these details are not important for the remainder of the paper, so have been deferred to Appendix \ref{GCI_appendix}. 

\subsubsection{Motivation for Singular Support} \label{SS_motivation_section}

Recall from Section \ref{GL_intro_section} the ``best hope'' form of the geometric Langlands conjecture. This says, in particular, that there is an equivalence of dg categories 
\[\text{D}(\bun_{G^\vee}(C)) \simeq \QC( \Flat_{G}(C)),\] 
satisfying a number of additional conditions.  Experts knew even in the 90s that the conjecture cannot be true as written: the dg category $\text{D}(\bun_{G^\vee}(C))$ is too big (see for instance the example of $C = \bb{CP}^1$ studied by Lafforgue \cite{Lafforgue}.  One point of view on the problem is that the dg category $\QC( \Flat_{G}(C))$ is ``not big enough'' to properly capture the information at the singular points of $\Flat_{G}$. Arinkin and Gaitsgory described a technique that enlarged the category in order to resolve this shortcoming by studying support conditions in dg categories.

In general, when a stack $\mc X$ is not smooth, the natural functor $\perf(\mc X) \to \coh(\mc X)$ is not an equivalence\footnote{The natural functor is defined only when $\mc X$ is eventually coconnective, that is, its structure sheaf should have cohomology in bounded degrees in the derived direction}, but only a fully faithful embedding. Here a perfect complex is a complex finitely generated by vector bundles on $\mc X$, whereas a coherent complex is finitely generated by coherent sheaves on $\mc X$. The inclusion of the ind-completions yields the natural functor $\QC(\mc X) \to \indcoh(\mc X)$.  Recall the following elementary example, showing that this inclusion is strict.

\begin{example}
Consider $A=k[\eps]/(\eps^2)$. The $A$-module $k$ is coherent, but not perfect; a resolution of $k$ as an $A$-module involves infinite terms $\xymatrix{
\cdots \ar[r]^-\eps & A \ar[r]^-\eps &  A \ar[r]^-\eps & k}$. 
\end{example}

As a next step, one may try to compare $\indcoh(\Flat_G)$, rather than $\QC(\Flat_G)$ with the dg category on the other side of the Langlands equivalence, but then it turns out to be too big. One has to find -- as we discussed in the introduction -- an intermediate dg category between these two examples.

Arinkin and Gaitsgory described intermediate dg categories between $\QC(\mc X)$ and $\indcoh(\mc X)$ for derived stacks $\mc X$ which are \emph{quasi-smooth}. Being quasi-smooth means that the tangent complex is perfect and has Tor-amplitude in degrees $\leq 1$. More concretely, at least for a derived scheme, being quasi-smooth means that $\mc X$ can be written locally as a complete intersection. Namely, $\mc X$ can be Zariski locally written as a derived fiber product as follows:
\[\xymatrix{
\mc X \ar[r] \ar[d] & \bb A^n \ar[d]\\
\pt \ar[r] &  \bb A^m.
}\]

The following is the most important example.
\begin{example} 
Consider the derived fiber product
\[\xymatrix{
\spec k[\eta] \ar[r] \ar[d] &  \pt \ar[d] \\
\pt \ar[r] & \bb A^1}\]
where $\eta$ is a parameter of degree $-1$. The affine derived scheme $\mc X=\spec k[\eta]$ is, therefore, quasi-smooth. Its tangent complex is $\bb T_\mc X = k[-1]$ generated by $\frac{\del}{\del \eta}$. One can consider $\bb A^m$ to get derived schemes of the form $k[\eta_1,\cdots,\eta_m]$ as well.
\end{example}

\begin{remark}
This is the first time derived schemes have appeared in this paper.  The category of affine derived schemes is just the opposite category to the category of commutative dgas concentrated in degrees $\le 0$.  More generally, there's a notion of a not necessarily affine derived scheme, which consist of those derived stacks admitting a Zariski atlas of affine derived schemes.  Derived schemes form a subcategory of the category of derived stacks; for a detailed reference explaining this theory we refer the reader to the book of Gaitsgory and Rozenblyum \cite{GRvol1}.
\end{remark}

To get an idea of how one can build subcategories of $\indcoh(X)$ for quasi-smooth stacks $X$, let us recall a familiar idea from the representation theory of a simple Lie algebra $\mf g$. If we want to understand the theory of $U\mf g$-modules, then we can consider the center $Z\mf g = Z (U \mf g)$; the representation theory of $U\mf g$ can be decomposed into the study of the algebraic geometry of the affine scheme $\spec Z \mf g \cong \mf h / W \cong \bb A^r$ together with the representation theory of the fiber, that is, the category $ U\mf g\text{-mod}_\chi$ for $\chi \in \spec Z \mf g$, that one might study using Beilinson--Bernstein localization. 

In light of that, we would like to describe the ``center'' of the category $\indcoh(X)$. In fact, one can make a definition that makes sense at the level of a general dg category, not only a category of ind-coherent sheaves.  This is the theory of the Hochschild cochains of a dg category.  Let's recall the construction, following \cite[Appendix E]{ArinkinGaitsgory}. 

We'll start our description by recalling a fundamental category-theoretic fact.

\begin{prop}[{\cite[Section 4.8.5]{HigherAlgebra}}] \label{E_2_adjunction_prop} 
There is an adjunction 
\[L \colon \{\bb E_2 \text{-algebras}\}\rightleftarrows \{\text{monoidal dg categories}\} \colon R\]
where -- as an $\bb E_1$-algebra -- $R(\LL) = \eend_{\LL}(1_{\LL})$, and where -- as a dg category -- $L(\mc A) = \mc A^{\mr{op}} \text{-mod}$.  Recall that our convention is that ``algebra'' always means an algebra in cochain complexes.
\end{prop}

The additional product on the $\bb E_1$-algebra $\eend_{\LL}(1_{\LL})$ is inherited from the monoidal structure on $\LL$.  Indeed, as Arinkin and Gaitsgory explain, the above adjunction should be thought of as a specific example of an adjunction that makes sense internally in any $\infty$-category, applied here internally to $\bb E_1$-algebras.  Using this adjunction we can define Hochschild cochains as an $\bb E_2$-algebra.
 
\begin{definition}
The \emph{Hochschild cochains} $\mr{HC}^\bullet(\mc C)$ of a dg category $\mc C$ are the $\bb E_2$-algebra $R(\eend_{\mr{dgCat}_{\mr{cont}}}(\mc C))$ obtained from Proposition \ref{E_2_adjunction_prop} -- that is, the algebra $\eend_{\eend_{\mr{dgCat}_{\mr{cont}}}(\mc C)}(\id_{\mc C})$.  The cohomology of this $\bb E_2$-algebra is the \emph{Hochschild cohomology} $\HH^\bullet(\mc C)$ of $\mc C$.
\end{definition}

Roughly speaking, an element $a \in \mr{HC}^\bullet(\mc C)$ is a combatible collection of endomorphisms $(a_{\mc F} \colon \mc F\to \mc F)_{\mc F \in \mc C}$. The dg category structure of $\mc C$ induces the differential. One can easily check that if $\mc C=A\text{-mod}$ for an ordinary associative algebra $A$, then the center $Z(A)$ is the same as $\HH^0(\mc C)$.

Now, we would like to have a space to use in order to study $\mc C$ locally, as in our representation theoretic analogy, but its not yet clear how to relate $\mr{HC}^\bullet(\indcoh(X))$ to the geometry of $X$.  In order to draw this bridge, we'll use as motivation the following fundamental theorem.

\begin{theorem}[Hochschild--Kostant--Rosenberg]
If $X$ is a smooth affine scheme, then one has an isomorphism $\HH^\bullet(\coh(X)) \cong \Gamma(X, \wedge^\bullet T_X)  $.
\end{theorem}

More generally, if $X$ is an eventually coconnective affine derived scheme, then one has an equivalence \cite[Corollary G.2.7]{ArinkinGaitsgory}
\[\mr{HC}^\bullet(\indcoh (X)) \simeq  \Gamma (X , U_{\OO_X} ( \bb T_X[-1]) ).\] 
We'll discuss this equivalence a bit more in the next subsection.  For now, one should just be aware of the non-obvious fact that the $(-1)$-shifted tangent complex admits a natural Lie bracket \cite{Hennion,Kapranov}, which we're using to make sense of the enveloping algebra appearing in this expression. When $X$ is a smooth affine scheme, then the Lie algebra $\bb T_X[-1]$ is necessarily abelian (for degree reasons) and we recover the Hochschild--Kostant--Rosenberg theorem.

When $X$ is only quasi-smooth, $\bb T_X[-1] $ is still very simple. One has $\mr H^0(X, \OO_X)\to \HH^0(\indcoh(X))$. Moreover, one has $\mr H^2(X, \bb T_X[-1] ) \to \HH^2(\indcoh(X))$ as a module over $H^0(X,\OO_X)$. Recalling the defining property of the shifted cotangent space, i.e. $T^*[-1]X = \spec \left(\sym_{\OO_X} ( \bb T_X[1])  \right)$, we can make the following definition.

\begin{definition} \label{scheme_sings_def}
The \emph{scheme of singularities} of a quasi-smooth derived stack $\mc X$ is defined to be the classical truncation $\sing(\mc X) = (T^*[-1] \mc X)^{\mr{cl}}$ of the shifted cotangent space.  This classical truncation is a functor that produces an underived stack for a derived stack: the intuition is that one ``truncates'' the structure sheaf -- a cochain complex in non-positive degrees -- at degree 0; more rigorously one can define the functor as a left Kan extension (\cite[I.2]{GRvol1}).
\end{definition}

\begin{remark}
So what does this mean, and why is it called the scheme of singularities?  The basic idea is that, for quasi-smooth stacks $\mc X$, the degree 1 part of the tangent complex, or the degree $-1$ part of the cotangent complex, controls the singularities of $\mc X$.  In other words, there is a natural projection map $\pi \colon \mr{Sing}(\mc X) \to \mc X^{\mr{cl}}$ to the classical truncation of $\mc X$, and the fibers $\pi^{-1}(x)$ over points in $\mc X^{\mr{cl}}$ are the vector spaces $\mr H^{-1}(\bb L_{\mc X,x})$, and control ``how singular'' $\mc X$ is at the point $x$.  In particular, the fiber $\pi^{-1}(x)$ over $x$ is just a point if and only if $\mc X$ is smooth at $x$.
\end{remark}

The above story was intended to motivate the idea, that the scheme of singularities captures some of the data of the Hochschild cochains of $\indcoh(X)$ in a geometric way.  In particular, there is a canonical map $\mc O_{\sing X} \to \HH^{\mr{even}}(\indcoh (X))$. From this, one can state Arinkin and Gaitsgory's main definition of singular support conditions.

\begin{definition}
Let $X$ be quasi-smooth. For a conical (i.e. invariant under the action of $\CC^\times$ by scalar multiplication on the fibers of the map $\pi \colon \mr{Sing}(X) \to X^{\mr{cl}}$) Zariski-closed subset $Z\subset \sing X$, we define $\indcoh_Z(X) \subset \indcoh(X)$ to be the full dg subcategory consisting of objects whose set-theoretic supports are contained in $Z$.
\end{definition}

\begin{example}
For the fundamental example of $X=\spec k[\eta_1,\cdots,\eta_m]$, the scheme of singularities is $\sing(X) = \bb A^m$.  One can identify quasi-coherent sheaves as $\QC(X) = \indcoh_0(X)$. For any $0 \subsetneq Y \subsetneq \bb A^m$, we obtain an intermediate dg category between quasi-coherent and ind-coherent sheaves..
\end{example}

The discussion so far can be summarized by the following table (Table \ref{analogy_table}) of analogies relating the decomposition of the category of representations of $\gg$ (over the center) and of the dg category of ind-coherent sheaves on a quasi-smooth scheme $X$ (over the scheme of singularities):

\begin{table}[h]
\centering
\begin{tabular}{|c|c|}
\hline
 Simple Lie algebra $\mf g$ & Quasi-smooth scheme $X$ \\
\hline\hline
$A\text{-mod} = U\mf g\text{-mod}$ & $\mc C=  \indcoh(X)$ \\
\hline
 $Z(A ) = Z \mf g$ & $ \mr{HC}^\bullet(\mc C) = \mr{HC}^\bullet(X)$\\
\hline 
$ \spec Z \mf g \simeq  \mf h / W \simeq \bb A^r$ & $\sing X = (T^*[-1]X )^{\mr{cl}}$ \\
\hline
$ U\mf g\text{-mod}_\chi$ & $\indcoh_Z(X)$ \\
\hline
\end{tabular}
\caption{Analogy between singular support conditions and the decomposition of the category of representations of a simple Lie algebra.}
 \label{analogy_table}
\end{table}

\subsubsection{Singular Support Conditions} \label{main_SS_section}

With that somewhat informal motivation in hand, let's give a more formal introduction to singular support conditions.  Along the way we'll set up some notation that we'll use later in the paper. For the rest of this section, $\mc C$ will denote a dg category, and $\mc A$ will denote an $\bb E_2$-algebra.  As we've been discussing, at its heart, the idea of singular support conditions stems from the action of the Hochschild cohomology of the category: they are set-theoretic support conditions for this action.  In order to make sense of this, and to expand on some of the physical motivation for our story that we introduced in Section \ref{construction_intro_section}, we'll need to explain what ``support'' means in a general context.  This formalism was originally developed by Benson, Iyengar, and Krause \cite{BIK} in the setting of triangulated categories.  

\begin{definition}
A (left) \emph{action} of an $\bb E_2$-algebra $\mc A$ on $\mc C$ is an action of the monoidal dg category $\mc A\text{-mod}$ on $\mc C$, that is, a monoidal functor $\mc A\text{-mod} \to \eend_{\mr{dgCat}_{\mr{cont}}}(\mc C)$.
\end{definition}

\begin{prop}
There is a canonical action of $\mr{HC}^\bullet(\mc C)^{\mr{op}}$ on $\mc C$, where here the superscript op indicates the opposite $\bb E_1$-algebra \footnote{Arinkin and Gaitsgory write this as $\mc A^{\text{int-op}}$, where ``int'' stands for ``internal'', they write $\mc A^{\text{ext-op}}$ for the opposite using the other (``external'') $\bb E_1$-structure.}.
\end{prop}

\begin{proof}
Given any monoidal dg category $\LL$ the counit of the adjunction from \ref{E_2_adjunction_prop} is a monoidal functor 
\[\eps_{LR} \colon \eend_{\LL}(1_\LL)^{\mr{op}}\text{-mod} \to \LL.\] 
Consider the monoidal dg category $\LL= \eend_{\mr{dgCat}_{\mr{cont}}}(\mc C)$ to obtain the desired functor.
\end{proof}

From this point of view, an alternative characterization of the action of $\mc A$ on $\mc C$ is as an $\bb E_2$-homomorphism $\mc A \to \mr{HC}^\bullet(\mc C)^{\mr{op}}$, or as a monoidal functor $\mc A^{\mr{op}}\text{-mod} \to \mr{HC}^\bullet(\mc C)\text{-mod}$.  We can use this to give a suitable definition of support conditions for a dg category.  

\begin{definition}
In what follows we'll write $\mr{H}^{\mr{even}}(\mc A)$ for the even part of the cohomology of $\mc A$, viewed as a commutative algebra upon forgetting the grading.  In particular, after passing to the even cohomology we can ignore the superscript ``op''.
\end{definition}

As we discussed in our section on motivation, the idea behind singular support conditions for a dg category $\mc C$ was that, if we can find a commutative algebra $A$ acting on $\mc C$, then we can decompose $\mc C$ into subcategories indexed by points in $\spec A$.  According to the above discussion, we can produce a commutative algebra acting on $\mc C$ from an $\bb E_2$-algebra $\mc A$ by taking the even cohomology -- if we took the full cohomology it would only be graded commutative.

\begin{definition}
Let $\mc A$ be an $\bb E_2$-algebra, and let $Y$ be a closed and conical -- invariant for the $\CC^\times$ action on $\mc A$ coming from the grading -- subset of $\spec \mr{H}^{\mr{even}}(\mc A)$.  The \emph{restriction} $\mc A\text{-mod}_Y$ to the subset $Y$ is the full subcategory of $\mc A\text{-mod}$ consisting of modules $M$ where $M$, or rather its image in the homotopy category, is set-theoretically supported on $Y$ as a $\mr{H}^{\mr{even}}(\mc A)$-module.
\end{definition}

\begin{remark} \label{conical_def_remark}
The condition that $Y$ was conical was necessary for the restriction to be well-defined as a dg category.  The reason is that the algebra $\mr{H}^{\mr{even}}(\mc A)$ does not act merely as a commutative algebra, but as a graded commutative algebra in even degrees.  As a result, the only localizations compatible with the grading are localizations by homogeneous ideals: closed subsets compatible with the grading, or equivalently, conical subsets.  For comparison, see the discussion in \cite[Section 3.6]{ArinkinGaitsgory}. 
\end{remark}

\begin{remark}
We define restriction in terms of the set-theoretic support, but in principle we could equally well consider a stricter notion where we restrict the scheme-theoretic support of modules. This scheme-theoretic notion, however, does not appear to produce non-trivial restricted categories in natural examples: see Section \ref{set_vs_scheme_support_rmk} for a more detailed discussion.  One could define more general support conditions by taking the scheme-theoretic support conditions with respect to ideals which are not maximal, having the effect of allowing deformations in some but not all derived directions.  While this may be mathematically interesting in some examples, we don't currently have a physically meaningful reason to consider such constructions.
\end{remark}

\begin{definition} \label{restriction_def}
Suppose $\mc A$ is an $\bb E_2$-algebra acting on a dg category $\mc C$, and let $Y$ be a closed subset of $\spec \mr{H}^{\mr{even}}(\mc A)$.  The category $\mc C_Y$ obtained by \emph{restricting} $\mc C$ to the subset $Y$ is the tensor product
\[\mc C_Y =\mc A\text{-mod}_Y \otimes_{\mc A\text{-mod}} \mc C\]
with respect to the action of $\mc A\text{-mod}$ on $\mc C$.
\end{definition}

\begin{definition}
The \emph{support} of an object $c \in \mc C$ in the affine scheme $\spec \mr{H}^{\mr{even}}(\mc A)$ is the minimal closed subset $Y$ of $\spec \mr{H}^{\mr{even}}(\mc A)$ such that $c$ lies in $\mc C_Y$.
\end{definition}

The restriction to $Y$ is a full subcategory of $\mc C$ -- we think of it in the following way.  The Hochschild cohomology $\mr{HH}^\bullet(\mc C)$ satisfies a universal property: it maps to the endomorphism algebra $\eend_{\mr{Ho} (\mc C)}(c)$ for every object $c$ in $\mc C$.  Note that here we're abusing notation slightly, and writing $c$ for the image of $c \in \mc C$ in the homotopy category $\mr{Ho}(\mc C)$.  As we explained above, if $\mc A$ is any $\bb E_2$-algebra acting on $\mc C$, then we obtain an (ungraded) map $\mr{H}^{\mr{even}}(\mc A)\to \mr{HH}^\bullet(\mc C)$.  Thus, composing these two maps, for every object $c$ of $\mc C$ there is a morphism
\[\mr{H}^{\mr{even}}(\mc A) \to \eend_{\mr{Ho} (\mc C)}(c).\]
Then the restricted category can be viewed as the full subcategory spanned by objects $c$ such that $\eend_{\mr{Ho}(\mc C)}(c)$ has support in $Y$ with respect to this composite action.

Now, let us return to the story from Section \ref{SS_motivation_section} and consider the special case where $\mc C = \IC(X)$ for a derived stack $X$.  We'll define support conditions using the universal $\bb E_2$-algebra acting on $\IC(X)$, namely $\mc A = \mr{HC}^\bullet(\IC(X))$ -- the Hochschild cochains of the category $\IC(X)$.  The singular support of a sheaf on $X$ is essentially going to be the support with respect to the action of the even Hochschild cohomology, but we'll be able to identify this more geometrically using the scheme of singularities of $X$ (as defined in Definition \ref{scheme_sings_def}).

\begin{remark}
From now on, when $X$ is a derived stack we'll simply write $\mr{HC}^\bullet(X)$ for the Hochschild cochain complex $\mr{HC}^\bullet(\IC(X))$.
\end{remark}

The relationship between the Hochschild cochains of a derived stack and its scheme of singularities is derived from the following version of the Hochschild--Kostant--Rosenberg theorem.  We'll use the fact that for any derived Artin stack \footnote{There's also a smallness condition: they should be locally finitely presented.} the shifted tangent complex $\bb T_X[-1]$ has the structure of a Lie algebra in the dg category $\QC(X)$.  In this generality the result is due to Hennion \cite{Hennion}, though the idea goes back at least to Kapranov \cite{Kapranov}. 

\begin{lemma}[{\cite[Corollary G.2.7]{ArinkinGaitsgory}}]
If $Z$ is an eventually coconnective affine derived scheme then there is a canonical isomorphism of associative algebras (in cochain complexes)
\[\Gamma(Z; U_{\OO_Z}(\bb T_Z[-1])) \to \mr{HC}^\bullet(Z)\]
from the derived global sections of the universal enveloping algebra of the Lie algebra object $\bb T_Z[-1]$ to the underlying associative algebra of Hochschild cochains of $Z$.
\end{lemma}

This induces, for any (eventually coconnective) affine derived scheme $Z$, an isomorphism
\[\mr H^\bullet(Z; U_{\OO_Z}(\bb T_Z[-1])) \to \HH^\bullet(Z),\]
therefore, in particular, a map $\Gamma(Z; \OO_{Z^{\mr{cl}}}) \to \HH^0(Z)$ of commutative algebras, and a map $\Gamma(Z; \mr{H}^0(\bb T_Z[1])) \iso \Gamma(Z; H^1(\bb T_Z)) \to \HH^2(Z)$ of $\Gamma(Z;\OO_{Z^{\mr{cl}}})$-modules if $Z$ is quasi-smooth.  Taking the symmetric algebra, this defines a canonical map of commutative algebras
\[\Gamma(\mr{Sing}(Z); \OO_{\mr{Sing}(Z)}) \to \HH^{\mr{even}}(Z).\]
Therefore we can define singular support in the following way.

\begin{definition}
The \emph{singular support} of a sheaf $\mc F \in \IC(Z)$ for a quasi-smooth affine derived scheme $Z$ is a closed subspace of $\mr{Sing}(Z)$ defined as the support of $\eend(\mc F)$ as a module over $\Gamma(\mr{Sing}(Z); \OO_{\mr{Sing}(Z)})$ with respect to the composition. The category of sheaves with singular support in $Y \sub \mr{Sing}(Z)$ is denoted by $\IC_Y(Z)$.
\end{definition}

Now, suppose $X$ is a quasi-smooth derived stack.  We can define the singular support of a sheaf $\mc F \in \IC(X)$ using a smooth atlas of affine derived schemes.
\begin{definition}
If $Y$ is a closed subspace of $\mr{Sing}(X)$, we define the category $\IC_Y(X)$ to be the limit
\[\IC_Y(X) = \lim_{Z \to X} \IC_{Y \times_X Z}(Z)\]
over smooth maps from affine derived schemes to $X$.  Here the inclusion $Y \inj \mr{Sing}(X)$ defines a map $Y\to X$ since $X$ is quasi-smooth.
\end{definition}

To conclude this section, we'll describe singular support conditions for a special sort of stack: those which are not just locally described by a complete intersection, but by a coherent family of complete intersections parameterized by a base stack $\mc X$.  The main result about such stacks closely follows results in \cite[Section 9]{ArinkinGaitsgory}, though we'll describe a slightly modified version here: introducing an auxiliary parameter that will be important in Section \ref{degree_shift_section}.  We'll use notation that is consistent with the notation of Arinkin and Gaitsgory.

The following definition should be compared to the definition of a local complete intersection, in the derived geometry language, that we gave in Section \ref{SS_motivation_section}.  Note that we're working in a much more general setting than we described there: we won't just allow global derived intersections of the form $\pt \times_{\bb A^m} \bb A^n$.  We'll replace the two affine spaces by any smooth stacks, and rather than simply taking the derived fiber, we'll allow families of fibers parameterized by a smooth base stack. 

\begin{definition}
A derived stack $\mc Z$ is a \emph{global complete intersection} if it can be obtained as a derived fiber product of smooth stacks 
\[\xymatrix{
\mc Z \ar[r] \ar[d] &\mc U \ar[d] \\
\mc X\ar@<0.6ex>[r]&\mc V \ar@<0.6ex>[l]
}\]
where $\mc X \to \mc V$ is a section of a smooth schematic map \footnote{A map of derived stacks is schematic if base changing by a derived scheme produces a derived scheme.} $\mc V\to \mc X$.
\end{definition}

Singular support conditions for global complete intersections can be described using the action of an algebra associated to the base stack $\mc X$ of the complete intersection.
\begin{definition}
Form the derived fiber product $\mc G_{\mc X/ \mc V} = \mc X \times_\mc V \mc X$.  This fiber product has the structure of a groupoid over $\mc X$.  In particular, the category $\IC(\mc G_{\mc X/\mc V})$ has the structure of a monoidal dg category over $\QC(\mc X)$.  We define $\mr{HC}^\bullet(\mc X / \mc V)$ to be the $\bb E_2$-algebra $R(\IC(\mc G_{\mc X/\mc V}))$ in the dg category $\QC(\mc X)$ defined by Proposition \ref{E_2_adjunction_prop} (that is, the algebra of endomorphisms of the monoidal unit).
\end{definition}

The description of singular support conditions in terms of this algebra appears in \cite[Lemma 9.2.6]{ArinkinGaitsgory}.  We'll state a slightly modified version of this description that incorporates an auxilliary degree 2 parameter $t$.  This is a technical detail, but it will be essential for Section \ref{degree_shift_section}.

Thus, we let $t$ be a degree 2 parameter, and write $\IC(\mc Z)[t,t^{-1}]$ for the tensor product $\IC(\mc Z) \otimes \CC[t,t^{-1}]\text{-mod}$.  We'll describe the restricted category $\IC_{Y[t,t^{-1}]}(\mc Z)[t,t^{-1}]$ associated to a closed subset $Y \sub \sing(\mc Z)$, given by restricting to the closed subset generated by $Y$ under the action of $\CC[t,t^{-1}]$. This restriction is defined as follows.

\begin{definition} \label{t_parameter_support_def}
If $\mc Z$ is a quasi-smooth affine derived scheme we define $\IC_{Y[t,t^{-1}]}(\mc Z)[t,t^{-1}]$ to be the restriction of $\IC(\mc Z)[t,t^{-1}]$ with respect to the ideal $\mc I_Y[t,t^{-1}]$ in $\OO(\sing(\mc Z))[t,t^{-1}]$ generated by $\mc I_Y$ under the $\CC[t,t^{-1}]$ action, using the natural composite morphism \[\OO(\sing(\mc Z))[t,t^{-1}] \to \HH^\bullet(\IC(\mc Z))[t,t^{-1}] \to \HH^\bullet(\IC(\mc Z)[t,t^{-1}]).\]  We'll often refer heuristically to the ``closed set $Y[t,t^{-1}]$'' even though this does not literally make sense since $\CC[t,t^{-1}]$ is not connective. If $\mc Z$ is a more general quasi-smooth derived stack we define $\IC_{Y[t,t^{-1}]}(\mc Z)[t,t^{-1}]$ as the limit 
\[\IC_{Y[t,t^{-1}]}(\mc Z)[t,t^{-1}] = \lim_{Z \to \mc Z} \IC_{(Y \times_{\mc Z} Z)[t,t^{-1}]}(Z)[t,t^{-1}]\]
over smooth affine charts $Z \to \mc Z$.
\end{definition}

\begin{remark} \label{conical_t_remark}
In Arinkin and Gaitsgory's setup the supports of objects are necessarily conical (i.e. $\CC^\times$-invariant).  As we mentioned above in Remark \ref{conical_def_remark} this is because $\HH^{\mr{even}}(\mc Z)$ is naturally a graded object and it's important to remember that grading, and consider only homogeneous ideals.  Upon introducing the parameter $t$ this still applies, but now $\CC^\times$ acts not only on the Hochschild cohomology $\HH^{\mr{even}}(\mc Z)$ but also on the parameter $t$ with weight 2.  

This means that we can consider support conditions for any conical subset of $\sing(\mc Z)[t,t^{-1}]$ (speaking heuristically -- more precisely we work with a smooth affine cover and locally specify homogenous ideals in $\OO(\sing(\mc Z))[t,t^{-1}]$), and one can obtain such a subset from \emph{any} closed subset of the classical scheme of singularities $Y \sub \sing(\mc Z)$ provided $\CC^\times$ acts with even weights.  Indeed, let $A$ be a dg commutative algebra with $\CC^\times$-action and an ideal $I$, and choose generators $\{a_i\}$ where $a_i$ has weight $2w_i$.  Embed $A$ into $A[t,t^{-1}]$ by sending $a_i$ to the weight zero element $a_it^{-w_i}$.  The image of $I$ under this map is $\CC^\times$-invariant, or equivalently generates a homogenous ideal in $A[t,t^{-1}]$.  This ideal recovers $I$ when we ``set the parameter $t$ to 1'', i.e. when we invariants for the $\CC^\times$ action purely in the $t$ direction. 

The restricted categories we obtain by this procedure are somewhat more general than the restrictions to closed sets $Y[t,t^{-1}]$ described above.  However, using the modified embedding $A \inj A[t,t^{-1}]$ giving the generators weight zero -- in the case of singular support a modified embedding $\OO(\sing(\mc Z)) \to \OO(\sing(\mc Z))[t,t^{-1}]$ -- we can make our homogeneous ideals into ideals of the form $\mc I_Y[t,t^{-1}]$, and therefore find ourselves in the situation described in Definition \ref{t_parameter_support_def} once more.  Essentially, we modified the action of the algebra $\OO(\sing(\mc Z))[t,t^{-1}]$ using an automorphism that becomes trivial after composition with the evaluation map $t \mapsto 1$.
\end{remark}

Now that we've introduced our auxiliary parameter, we can state the main result that we'll need, characterizing singular support conditions for global complete intersection stacks. 

\begin{definition}
Let $V$ denote the total space of the vector bundle $\bb V \to \mc X$ freely generating $\mr{HC}^\bullet(\mc X / \mc V)$, in the sense that $\mr{HC}^\bullet(\mc X / \mc V) \iso \sym_{\OO_{\mc X}}(\bb V[-2])$ (see \cite[9.1.4]{ArinkinGaitsgory} for an explicit description).
\end{definition}

\begin{prop} \label{SS_complete_intersection_prop}
There is a canonical embedding $\sing(\mc Z) \to V \times_{\mc X} \mc Z$, and therefore a canonical morphism $f \colon \sing(\mc Z) \to V \times_{\mc X} \mc U$. If $\wt Y$ is a closed subset of $V$, and $Y = f^{-1}(\wt Y \times_{\mc X}\mc U)$ is the corresponding closed subset of $\mr{Sing}(\mc Z)$, then there is a canonical action of the category $\mr{HC}^\bullet(\mc X / \mc V)^{\mr{op}}\text{-mod}$ on $\IC(Z)$.  If we include the degree 2 parameter $t$ then there is an equivalence of dg categories
\[\IC_{Y[t,t^{-1}]}(\mc Z)[t,t^{-1}] \iso  \IC(\mc Z)[t,t^{-1}]  \otimes_{\mr{HC}^\bullet(\mc X / \mc V)[t,t^{-1}]^{\mr{op}}\text{-mod}}   \mr{HC}^\bullet(\mc X / \mc V)[t,t^{-1}]^{\mr{op}}\text{-mod}_{\wt Y[t,t^{-1}]}.\]
\end{prop}

We defer the proof to Appendix \ref{GCI_appendix}.

\begin{remark}
Let us try to explain the role this result will play in the present paper.  We'll use this result in order to calculate categories with a singular support condition, using the fact that the action of the groupoid $\mc G_{\mc X / \mc V}$, and therefore the action of the category $\mr{HC}^\bullet(\mc X / \mc V)^{\mr{op}}\text{-mod}$, can be described geometrically.  Specifically we'll be interested in the case (appearing in geometric Langlands) where the stack $\mc Z$ is the stack $\flat_G(C)$ of flat $G$-bundles on a curve $C$.  We'll describe in Section \ref{geometry_of_flat_section} how this can be written as a global complete intersection.  For this example, we'll be able to describe the groupoid $\mc G_{\mc X /\mc V}$ and its action on $\flat_G(C)$.  Even better, we'll be able to identify this action with the action of the category of line operators in the B-twisted $\mc N=4$ gauge theory.

The auxiliary parameter $t$ appears -- as we'll explain in Section \ref{degree_shift_section} -- when we calculate the topological twist of the full $\mc N=4$ theory.  Normally one kills this parameter when describing the twist, but by keeping it around we can use it to make sense of singular support conditions for non-conical subsets $Y$ of $\sing(\mc Z)$.
\end{remark}
 
\subsection{Restriction to a Vacuum} \label{loc_and_vacuum_section}
\subsubsection{Compatibility with a Vacuum in the Factorization Algebra Context} \label{fact_algebra_restriction_section}
In the functorial context, the notion of boundary conditions admits a precise mathematical model \emph{if} one is working with a purely topological theory.  However, many theories that one wishes to study are not actually topological in the sense of Atiyah--Segal. For instance, an $n$-dimensional topological quantum field theory in the sense of Atiyah--Segal would assign a number as a partition function to an $n$-dimensional closed manifold, but for a general field theory, there is not sufficient data to produce a number unless we first choose a vacuum state \cite[Section 4.9]{CostelloGwilliam1}. Accordingly, once we make a choice of vacuum state then only the boundary conditions compatible with the choice would be relevant, picking out a particular class of boundary conditions. The connections to geometric representation theory that we have in mind can be captured by the following slogan:

\begin{center}
\textbf{Support conditions for categories of boundary conditions in a TQFT arise by demanding compatibility with a choice of vacuum state.}
\end{center}

To put it another way, our main aim in this paper is to explain a physical motivation for certain kinds of condition that one can impose whenever one interprets a category as the category of boundary conditions in a TQFT. This will have non-trivial content because not every TQFT arising in physics is of Atiyah--Segal type: we will make the idea more precise. In this subsection we'll describe the notion of a ``vacuum state'' in a TQFT.  In the following subsection we'll observe that in a 2-extended TQFT of Atiyah--Segal type this notion is not very interesting, because the moduli of vacua is necessarily discrete.

Let $\mc A = \obsq(\RR^n)$ be a locally constant factorization algebra on $\RR^n$, or equivalently an $\bb E_n$-algebra.  Let $\mc B$ be a dg category such that $\mc A$ acts on $\mc B$, i.e. there exists an $\bb E_2$-algebra map $\mc A^{\mr{op}} \to \mr{HC}^\bullet(\mc B)$.  We'll usually construct such a setup by quantizing a \emph{classical} field theory: $\mc A$ will quantize the local classical observables and $\mc B$ will arise by categorical geometric quantization. That is, $\mc B$ is a category (geometrically) quantizing $\EOM(Y)$, the space of solutions to the equations of motion on an $(n-2)$-manifold $Y$, which is expected to be 1-shifted symplectic. The general theory of categorified geometric quantization in this case involves choosing a polarization and taking sections of a gerbe (see  \cite{Wallbridge} and forthcoming work of Safronov). On the other hand, often $\EOM(Y)$ will be a 1-shifted cotangent bundle $T^*[1]\mc X$, for which there's a natural choice of Lagrangian given by the zero-section $\mc X$ so that the categorified geometric quantization yields $\mc B=\IC(\mc X)$; in particular this will be the case for our B-twisted theory.

We begin by introducing the moduli space of vacua. We recall the following definition, specialized for topological field theories, from our discussion in Section \ref{construction_intro_section}.  First, recall from Definition \ref{states_def} that the space of \emph{states} in a quantum field theory is the graded vector space of graded linear functionals $\phi \colon \mr H^\bullet(\obsq(B^n)) \to \CC$ of degree zero.  In a topological quantum field theory, a state $\phi$ is called a \emph{vacuum state} if it is a graded ring homomorphism, meaning that vacuum states are in canonical bijection with maximal ideals in $\mr H^0(\obsq(B^n))$.  We will broaden this definition by considering families of vacuum states over a derived base, in order to probe the cohomology of the algebra of local observables outside degree zero.  To that end, if $\eps$ is a parameter of cohomological degree $k < 0$, a $\CC[\eps]$-point of the moduli stack of vacua will be a degree zero graded ring homomorphism $\phi \colon \mr H^\bullet(\obsq(B^n)) \to \CC[\eps]$.  This motivates the following definition.

\begin{definition}
The \emph{moduli space $\vac$ of vacua} is the affine derived stack $\spec \mr{H}^\bullet(\obs^{\mr{q}}(B^n))$ (where $\spec$ is used in the sense of \cite[Section 3.2]{BZNLoopSpaces1}). 
\end{definition}

\begin{remark}
We originally learned the related idea -- that for an $n$-dimensional functorial TQFT, the moduli space of vacua should be given by $\spec Z(S^{n-1})$ -- from David Ben-Zvi.  Forthcoming work of Ben-Zvi--Neitzke \cite{BenZviNeitzke} will discuss this idea in more detail.  Note that, ideally, we would not need to pass to the cohomology of the algebra of local observables, but would work with the $\bb E_n$-algebra directly.  In this paper we work with the cohomology in order to make a definition in terms of ordinary (dg) commutative algebras, rather than being forced to grapple with a version of ``$\bb E_n$ algebraic geometry''. 
\end{remark}

One can now define the restricted category $\mr{H}^\bullet(\mc A)\text{-mod}_{\{v\}}$ for a choice $v \in \vac$ of vacuum following the formalism developed above, i.e. by applying Definition \ref{restriction_def}.

\begin{definition} \label{vacuum_restriction_def}
The \emph{restriction} of the category $\mc B$ to a closed subset $V \sub \vac$ is the tensor product $\mc B \otimes_{\mc A\text{-mod}} \mc A\text{-mod}_V$.  As usual this can be thought of as the full subcategory of objects $\mc F \in \mc B$ such that $\eend_{\mc B}(\mc F)$ is supported at $V$ with respect to the action of $\mc A$ through Hochschild cohomology.
\end{definition}

We spent some time discussion the motivation behind this definition in Section \ref{construction_intro_section}.  The idea is that the restricted category of boundary conditions is the full subcategory consisting of objects $\mc F$ compatible with a vacuum $v$ in the specified subset $V \sub \vac$.  Here ``compatibility'' can be thought of as saying that the vacuum in the bulk theory (modelled by the factorization algebra $\mc A$) extends to the bulk-boundary theory with boundary condition $\mc F$ (modelled by the endomorphism algebra $\eend_{\mc B}(\mc F)$).  We could also think of the localized ring $\mr H^\bullet(\mc A)_V$ at the closed subset $V$ as modelling the local behavior of observables near the subset $V \sub \vac$, and the restricted category as the full subcategory on boundary conditions whose bulk-boundary system survives to a non-trivial module over the localized ring $\mr H^\bullet(\mc A)_V$. 

\begin{remark}
We need to take a moment to discuss conical versus non-conical subsets of $\vac$, and genuine dg categories vs 2-periodic dg categories.  In Definition \ref{vacuum_restriction_def} there is no requirement that the subset $V$ is conical, only that it's a closed subset in the graded sense (determined by a homogenous ideal in the graded ring $\mr{H}^\bullet(\mc A)$).  In Section \ref{degree_shift_section} we'll discuss restrictions corresponding to closed subsets of $\mr{H}^\bullet(\mc A)$ viewed as an \emph{ungraded} ring, so to not-necessarily homogenous ideals.  The corresponding restricted categories are only $\ZZ/2\ZZ$-graded dg categories, not ordinary $\ZZ$-graded dg categories.  We'll make this story precise by adjoining an auxiliary degree 2 parameter $t$, like we saw in Proposition \ref{SS_complete_intersection_prop} (indeed, that was exactly why we needed to include that parameter in our singular support discussions in Section \ref{main_SS_section}).
\end{remark}

\subsubsection{Set and Scheme Theoretic Support}\label{set_vs_scheme_support_rmk}
In this subsection we will discuss the distinction between set-theoretic support and scheme-theoretic support. This discussion is parenthetical; it is not necessary to understand the content here in order to follow the main flow of this paper.

Definition \ref{vacuum_restriction_def} used the set-theoretical notion of support, rather than the scheme-theoretical notion.  Being scheme-theoretically supported on $Y$ is a stronger condition, that imposes a restriction on the non-reduced directions in $\vac$, not just the ordinary support.  In terms purely of commutative algebra, the set of modules $M$ over a commutative ring $A$ where the localization $M_v$ of $M$ at a point $v \in \spec A$ is non-zero consists of the set of modules set-theoretically supported at $v$, not scheme-theoretically.

Often, if one uses a scheme-theoretical support condition instead of a set-theoretical support condition then the category of boundary conditions compatible with the vacuum often turns out to be trivial.  We'll explain this in a family of examples in a way which we hope is suggestive of a more general story.

For simplicity we'll consider the case where $A$ is a local ring, so $\spec A$ consists of a single closed point (but of course is generally non-trivial as an affine scheme). For instance, consider the topological B-model with target a compact Calabi--Yau manifold $X$, so $\mc B= \coh(X)$ is the category of boundary conditions and $A\cong \HH^\bullet(X)\cong\PV(X)$ is a finite-dimensional but not semi-simple ring. Then the category of boundary conditions is always compatible with the unique choice of a vacuum in our sense by definition. On the other hand, to understand the meaning of being scheme-theoretically supported at the point, let us investigate the map $\HH^2(\coh(X))\to \mr{Ext}^2_{\coh(X)}(\mc F,\mc F)$ for $\mc F \in \coh(X)$ and try to see for which $\mc F$ this map vanishes. It is possible to proceed rigorously, but it will be enough to provide a heuristic, but conceptual, explanation for this remark.

If $\mc B$ is a dg category, there is a fiber sequence
\[\mr{Def}_{\mc B}(\mc F) \to \mr{Def}(\mc B, \mc F)\to \mr{Def}(\mc B)\] 
of dg Lie algebras governing the corresponding deformation theories (see \cite[Proposition 4.3]{BlancKatzarkovPandit} for a statement in the context of spectral geometry), where $\mr{Def}_{\mc B}(\mc F)$ is the shifted tangent complex $ \bb T_{\mc F}[-1] \mc M_{\mc B} =  \mr{End}(\mc F)$ to the moduli $\mc M_{\mc B}$ of objects in $\mc B$ (see for instance \cite{ToenVaquieModuli}), $\mr{Def}(\mc B)$ is $\mr{HC}^\bullet(\mc B)[1]$ \cite[Section 3.3]{BlancKatzarkovPandit}, and $\mr{Def}(\mc B, \mc F)$ is their extension governing deformations of the category $\mc B$ together with the object $\mc F$. In particular, when $\mc B = \coh(X)$, we find the boundary map $\HH^2(X)\to \mr{Ext}^2(\mc F,\mc F)$, where $\HH^2(X)$ describes the deformation of $X$ and $\mr{Ext}^2(\mc F,\mc F)$ describes the obstruction of $\mr{Def}(\mc F)$. Now by the Hochschild--Kostant--Rosenberg theorem, we can decompose the Hochschild cohomology as 
\[\HH^2(\coh(X)) \iso H^0(X,\wedge^2 \bb T_X) \oplus H^1(X, \bb T_X) \oplus H^2(X,\OO_X);\] 
from the perspective of deformation theory, these three terms correspond to a noncommutative deformation of $\OO_X$, a complex structure deformation of $X$, and a gerbal deformation of $\coh(X)$, respectively.  Then the map $\HH^2(X)\to \mr{Ext}^2(\mc F,\mc F)$  would measure whether, for a given element of $\HH^2(X)$, corresponding to a deformation of $\coh(X)$, the object $\mc F$ survives as a non-trivial object of the deformed category. In particular, looking for an object scheme-theoretically supported on the unique closed point in the moduli space of vacua amounts to finding an object $\mc F$ which survives all such deformation of $\coh(X)$, which is an extremely strong constraint.

For instance, consider the category $\coh(X)$ for a K3 surface $X$. It is enough to check that vector bundles on irreducible subvarieties do not survive all three types of deformation. No curve class survives a complex structure deformation. A global vector bundle doesn't survive a gerbal deformation. Finally, the skyscraper sheaf at a point $x$ does not survive a noncommutative deformation.  To see this, one should note that by taking an open affine neighborhood of $x$, if the skyscraper sheaf were to survive the deformation, then it would still define a one-dimensional representation of the non-commutative deformation of the ring of functions on the open patch. This deformation is equivalent to the Weyl algebra, which doesn't admit any finite-dimensional representations. In sum, the category of boundary conditions on K3 surfaces scheme-theoretically supported on the unique point is trivial. 

\begin{example} \label{2d_TQFT_examples}
If the theory is semi-simple, which corresponds to the case where $\spec A$ consists of a finite number of points, then one can decompose the category $\mc C$ into the corresponding finite number of simpler subcategories living over each vacuum. For instance, this happens when one considers the A-model or B-model with target a Fano manifold, or a Landau--Ginzburg theory. One can think of the subcategory as describing the category of boundary conditions in the IR limit with the given vacuum (for a related discussion and more for the example of Fukaya--Seidel type categories, one should refer to the work of Gaiotto, Moore, and Witten \cite{GaiottoMooreWitten}; we expect their finite set $\bb V$ of vacua to be the vacuum moduli space in our sense and we interpret their work as an attempt to recover the entire category $\mc C$ from those simpler categories over each point of $\bb V$). 

There's an interesting example of a semi-simple fully extended 2d TQFT that can be considered from the point of view of support in the moduli space of vacua.  The 2d \emph{Landau--Ginzburg model} with $n$-dimensional target and superpotential $W$ has, as its category $\mc C$ of boundary conditions, the (2-periodic) dg category $\mr{MF}(\CC[z_1, \ldots, z_n], W)$ of matrix factorizations of $W$.  The $\bb E_2$-algebra of local observables is given by the Hochschild cochains of this category.  By a theorem of Dyckerhoff \cite[Corollary 6.4]{Dyckerhoff} the Hochschild cohomology of this category is quasi-isomorphic to the \emph{Jacobian ring} of $W$.  This ring is semi-simple but typically not simple, so one can restrict the category of matrix factorizations at a point in its spectrum and obtain a proper subcategory of the category of matrix factorizations.  According to a Orlov's comparison theorem \cite[Theorem 3.9]{Orlov} and its dg enhancement \cite[Theorem 2.49]{BRTV}, the category of matrix factorizations is equivalent to the singularity category of the zero fiber $W^{-1}(0)$, i.e. the quotient of the dg category of coherent complexes by the dg category of perfect complexes.  We speculate that the restriction of the category of matrix factorizations to a singular point -- i.e. a point in the spectrum of the Jacobian ring -- corresponds to restricting to those objects of the singularity category supported at said singularity.
\end{example}

\subsubsection{Motivation from Functorial Field Theory} \label{loc_functorial_section}
In this section we'll explain an alternative construction of the restricted category of boundary conditions in a topological field theory, from the functorial point of view.  Given some assumptions on the compatibility between the algebraic and functorial perspectives on quantum field theory we expect this definition to be a special case of Definition \ref{vacuum_restriction_def} in a precise sense, but for the present work we won't need a precise result of this nature; this section should be thought of as providing additional motivation for our constructions.  

Let $Z$ be a 2-extended $n$-dimensional TQFT as defined in Section \ref{functorial_background_section}.  If $Z$ comes from the restriction of a fully extended TQFT then full dualizability forces the algebra $Z(S^{n-1})$ to be cohomologically finite-dimensional, and therefore to have discrete spectrum.  In many interesting examples (including the A-model and B-model with compact Calabi--Yau targets) the spectrum of the algebra $Z(S^{n-1})$ consists of a single point, but for more general cases where it doesn't, we can try to investigate a procedure to \emph{restrict} the theory $Z$ at a point in this spectrum, thus splitting our topological field theory into sectors of this form.

More concretely, if we choose a compact $(n-2)$-manifold $Y$ we can consider the dimensional reduction of $Z$ along $Y$, that is, the fully extended 2-dimensional TQFT $Z_Y$ defined by $Z_Y(X) = Z(Y \times X)$.  We can now consider the restrictions of this theory with respect to the action of $Z(S^{n-1})$.  As we discussed in Section \ref{functorial_background_section} there is an $\bb E_2$-action of $Z(S^{n-1})$ on the dg category $Z(Y)$, where the $\bb E_n$-algebra $Z(S^{n-1})$ is broken down to $\bb E_2$ by the reduction along the $(n-2)$-manifold $Y$.  We can therefore apply the definitions of Section \ref{restriction_background_section}, where now the role of the space of vacua is played by $\mr{H}^{\mr{even}}(Z(S^{n-1}))$.

\begin{definition}
If $v$ is a closed point in $\spec\mr{H}^{\mr{even}}(Z(S^{n-1}))$, the space of boundary conditions along $Y$ \emph{compatible} with the point $v$ is defined to be the restricted category $Z(Y)_{v}$. 
\end{definition}

\begin{examples}
As we noted above, in the fully extended context the algebra $Z(S^{n-1})$ is finite-dimensional, so the moduli space $\mr{H}^{\mr{even}}(Z(S^{n-1}))$ is always discrete.  The theories we discussed in Example \ref{2d_TQFT_examples} are 2d examples of this form, with discrete moduli spaces of vacua.
\end{examples}

\begin{remark} \label{state_operator_E_n_comparison_remark}
According to the state-operator correspondence, we anticipate an equivalence of cochain complexes between $Z(S^{n-1})$ and the complex of local quantum observables $\obs^{\mr{q}}(B^n)$ in any topological context where both can be defined (as discussed in Section \ref{comparing_background_section})\footnote{We don't have a class of theories which have two such descriptions in mind.  While it's possible to construct factorization algebras of quantum observables explicitly for, e.g, topologically twisted supersymmetric theories, there isn't a known procedure for building corresponding functorial field theories assigning a number to a compact $n$-manifold.  It would be interesting to check this expectation for the example of Chern--Simons theory.}.  Our expectation is that this equivalence should be compatible with the $\bb E_n$-structures coming from $n$-dimensional pairs of pants (on $Z(S^{n-1})$) and from factorization (on $\obs^{\mr{q}}(B^n)$)  In particular, we conjecture that for 2-extended topological quantum field theories there is a graded algebra isomorphism $f$ making the following triangle commute:
\[\xymatrix{\mr H^\bullet(Z(S^{n-1})) \ar[r] \ar[d]_f & \HH^\bullet(Z(Y)), \\ \mr H^\bullet(\obs^{\mr{q}}(B^n)) \ar[ur] &}\]
meaning that the support conditions defined in terms of functorial field theory and in terms of factorization algebras can be made to exactly coincide.  We emphasize again that the parallel functorial story is intended to motivate some of the support condition constructions in this paper in parallel with the motivation described in Section \ref{construction_intro_section}, and our results don't rely on this conjectural equivalence.
\end{remark}

\begin{remark}
One can define these categories separately for each $(n-2)$-manifold $Y$; it's natural to ask whether there's a restricted 2-extended $n$-dimensional functorial field theory $Z_v$ which assigns $Z(Y)_v$ to each $Y$.  We hope to return to this question in future work. 
\end{remark}

\section{Vacua and Singular Supports} \label{SS_section}
Having described the abstract idea of restriction to a vacuum in a few different ways, we'll explain how to implement the definition for our main example, which is the B-twisted $\mc N=4$ gauge theory we constructed in \cite{EY1}.  We'll comment on how the definitions work for the A-twisted theory, but there are still a number of unanswered questions that must be addressed before that example can be treated in depth.

\subsection{Twisted $\mc N=4$ Gauge Theories} \label{twisting_section}
According to the definitions given in Section \ref{fact_algebra_restriction_section}, we have to describe a locally constant factorization algebra of quantum observables, along with the category of boundary conditions on a curve $C $ acted upon by the algebra of observables.  We'll motivate these objects in the B-twisted theory using our construction of the classical B-twisted theory, along with some formality results specifying the quantization.  The model we'll derive in this section will need a somewhat subtle correction, which we'll address in Section \ref{degree_shift_section} below.

Let's begin by recalling what it means to twist a supersymmetric field theory, and our construction of the classical geometric Langlands twisted theories after Kapustin and Witten.  Suppose we're given a perturbative classical field theory $E$ with an action of the super group $H=\CC^\times \ltimes \Pi \CC$ which we denote by $(\alpha, Q)$, where $\alpha$ denotes the action of $\CC^\times$ and $Q$ denotes the action of $\Pi \CC$. Following Costello \cite{CostelloSH}, we can give a construction of the \emph{twisted theory} as a family of perturbative classical field theories parametrized by $\CC[u,u^{-1}]$, where $u$ is a parameter of $\ZZ/2\ZZ \times \ZZ$ bidegree $(1,1)$.  Here the parameter $u$ is a generator for the algebra of algebraic functions on the classifying stack $B(\Pi \CC)$, or equivalently, the Chevalley--Eilenberg cochain complex of the super Lie algebra $\Pi \CC$ acting on our classical field theory.

\begin{remark}
Note that here we use the uncompleted version of the Chevalley--Eilenberg cochain complex which gives $\CC[u]$, not $\CC[[u]]$ (we can view this latter object as a pro graded algebra; it's not a graded algebra in the usual sense since it's not equivalent to the direct sum of its homogenous components). From the point of view of deformation theory, using the completed version is more natural, see for instance the analysis of Calaque and Grivaux \cite{CalaqueGrivaux}. However, for a nilpotent Lie algebra $\gg$ like $\Pi \CC$ the completion is less important than it is in the general case, since the formal moduli problem $B\gg$ can be represented by the uncompleted version of the Chevalley--Eilenberg complex, according to a theorem of Getzler \cite{Getzler}.  
\end{remark}

\begin{definition}
The \emph{twisted family} deforming $E$ associated to the twisting data $(\alpha,Q)$ is the perturbative classical field theory over the ring $\CC[u,u^{-1}]$ defined by
\[E^Q_u = E^{\Pi \CC}\otimes_{\CC[u]}  \CC[u,u^{-1}].\]
Here we note that the space $E^{\Pi \CC}$ of derived $\Pi \CC$ invariants lives over the classifying stack $B (\Pi \CC)$, so it has the structure of a module over $\OO( B (\Pi \CC) ) \iso \CC[u]$.
\end{definition}

Now, let's specialize this definition to a single twisted theory over $\CC$, without an auxiliary parameter.

\begin{definition}
The theory \emph{twisted} by the action $(\alpha, Q)$ is the ``generic fiber'' of this $\CC[u,u^{-1}]$ -family, that is, it is the space $E^Q=(E^{\Pi \CC}\otimes_{\CC[u]}  \CC[u,u^{-1}] )^{\CC^\times}$ of $\CC^\times$-invariants after inverting the twisting parameter $u$, or heuristically restricting to the odd shifted punctured line.
\end{definition}

\begin{remark}
Again, this definition is apparently different from the one of Costello \cite{CostelloSH}, because we use the uncompleted version of Chevalley--Eilenberg cochain complex. However, one can identify $(A \otimes \CC[u,u^{-1}])^{\CC^\times}\cong (A \otimes \CC(\!(u)\!))^{\CC^\times}$ for a $\CC(\!(u)\!)$-module $A$ as long as one has a weight decomposition $A=\bigoplus A_n$ of $A$ into the sum of its irreducible components under the $\bb C^\times$ action, which is true for the examples considered in \cite{EY1} (for instance, see Remark \ref{B_twist_u_family_remark} below). \vspace{-1em}

In alternative language, we can define the twisted family via a graded mixed complex, meaning a cochain complex $E$ with an additional $\ZZ \times \ZZ/2$-grading and a square-zero endomorphism $Q$ of degree $(1,1)$ and cohomological degree 0 \cite{PTVV}.  This in particular defines a $\ZZ/2\ZZ$-graded cochain complex $E_u$ with an action of $\CC[u]$, where now $u$ has $\ZZ/2\ZZ$-degree 1.  In examples coming from supersymmetric gauge theory $E$ will be defined to be the classical BV complex.  The module $E_u$ then has the property referred to above, of admitting a decomposition into the direct sum of its weight spaces, as long as $E$ has this property, which is the case since the space of fields decomposes into the sum of finitely many irreducible modules for the group of R-symmetries.
\end{remark}

\begin{remark}
As we'll see in Section \ref{degree_shift_section}, when we discuss exactly how to think about the moduli space of vacua in such a twisted theory, for some applications it is valuable to remember the entire twisted family, rather than taking $\CC^\times$-invariants.
\end{remark}

In \cite{EY1}, we extended this definition to define twists \emph{globally}, meaning twists of $(-1)$-shifted symplectic derived stacks with an action of the super group $H$.  Using this extended definition, we described the geometric Langlands A- and B-twists as deformations of a single, minimal, ``holomorphically twisted'' version of the $\mc N=4$ theory.  This is an important construction, since the holomorphic twist is the minimal twist of $\mc N=4$ super Yang--Mills theory which admits a purely algebraic description in four dimensions.  There is a natural origin story for this algebraic structure arising from the fact that $\mc N=4$ super Yang--Mills can be obtained by compactifying a six-dimensional theory (a gauge theory on twistor space) which admits a canonical algebraic structure.

Before we summarize our results from \cite{EY1}, we'll need a piece of terminology from derived geometry.

\begin{definition}
The \emph{de Rham stack} of a derived stack $\mc X$ is intended to capture the idea ``identify infinitesimally close points in $\mc X$''.  Formally, if we think of $\mc X$ in terms of its functor of points -- so for every commutative dga $R$ in degrees $\le 0$, the $R$-points of $\mc X$ define a simplicial set $\mc X(R)$ -- the de Rham stack is defined to be the functor $\mc X_{\dR}(R) = \mc X(R^{\mr{red}})$, where $R^{\mr{red}}$ is the quotient of $R$ by its nilradical.

One can think of $\mc X_{\dR}$ as appearing by deforming something a bit more geometric.  There is a derived stack $\mc X_{\mr{Hod}}$ over $\bb A^1$ called the \emph{Hodge stack} of $\mc X$.  This family over $\bb A^1$ is $\bb G_m$-equivariant, the fiber over 1 is $\mc X_{\dR}$, and the fiber over 0 is the formal completion $T_{\mr{form}}[1]\mc X$ of the 1-shifted tangent space along the 0-section.  The informal idea here is that, at least for a smooth complex variety $X$, the structure sheaf of $T_{\mr{form}}[1]\mc X$ looks like the complex $(\Omega^{\bullet, \bullet}(X), \ol \dd)$ of $(p,q)$-forms on $X$, and turning on the deformation means turning on the differential $\dd$ (so at $\lambda = 1$ we recover the de Rham complex of $X$).
\end{definition}

\begin{example}
The \emph{holomorphically twisted} $\mc N=4$ gauge theory on a smooth proper complex algebraic surface $X$ has the following $(-1)$-shifted symplectic moduli stack of solutions to its equations of motion
\[\EOM_{\mr{hol}}(X) = T_{\mr{form}}^*[-1]\higgs^{\mr{fer}}_G(X)\]
where $\higgs_G^{\mr{fer}}(X)$ is the moduli stack of $G$-Higgs bundles on $X$ where the Higgs field is fermionic: that is, the derived mapping stack $\ul{\mr{Map}}(\Pi TX, BG)$.  The notation $T_{\mr{form}}^*[-1]\mc X$ means the formal completion of the shifted cotangent stack $T^*[-1]\mc X$ along its zero section.  There are two natural actions of the supergroup $H$ on this moduli stack: $\CC^\times$ acts on the Higgs field with weight one, and the supertranslation either acts on the base or the cotangent fiber.  We proved that the corresponding global twists are given by
\begin{align*}
\EOM_B(X) &= T_{\mr{form}}^*[-1]\flat_G(X) \\
\text{and } \EOM_A(X) &= \higgs_G(X)_{\mr{dR}}
\end{align*}
respectively.

One can also compute the twist of the space of solutions to the equations of motion on a smooth proper curve $C$.  The stacks that arise are now 1-shifted symplectic: one finds
\begin{align*}
\EOM_B(C) &= T_{\mr{form}}^*[1]\flat_G(C) \\
\text{and } \EOM_A(C) &= \higgs_G(C)_{\mr{dR}}  = T^*_{\mr{form} }[1] (\bun_G(C)_{\mr{dR}} ) .
\end{align*}
One might note that the last identification can be further reduced to $\bun_G(C)_{\mr{dR}}$, since the tangent complex of a de Rham stack is always trivial. However, this above identification has nontrivial content since it extends to an identification compatible with the realization of the de Rham stack as a specialization of the Hodge stack; for a further discussion on this point see \cite[Section 4.3]{EY1}.

Finally, one can compute the twists of the space of solutions to the equations of motion on a point.  In this case one obtains 3-shifted symplectic stacks:
\begin{align*}
\EOM_B(\pt) &= T^*[3]BG \iso \gg^*[2]/G \\
\text{and } \EOM_A(\pt) &= (BG)_{\mr{dR}}.
\end{align*}
The latter statement follows because we can identify $\EOM_A(\pt) = \underline{\mr{Map}}( \bb D^2 , BG )_{\mr{dR}} = \underline{\mr{Map}}( \bb D^2 , BG_{\mr{dR}} )$, where $\bb D^2$ stands for a 2-dimensional formal disk: $\bb D^2 = \spec \CC[[x_1,x_2]]$. Since $\bb D^ 2= \colim \bb D^2_{(n)}$, where $\bb D^2_{(n)}$ is the $n^\text{th}$ infinitesimal neighborhood of a point on a smooth surface, this moduli space is equivalent to 
\[\underset {n \to \infty}\lim\, \underline{\mr{Map}}( \bb D^2_{(n)} , BG_{\mr{dR}} ) \iso \underset{n \to \infty} \lim \,\underline{\mr{Map}}( \pt , BG_{\mr{dR}} ) = \underline{\mr{Map}}( \pt , BG_{\mr{dR}} ) .\]

The factorization algebras of local observables are the $\bb P_4$-algebras of functions on these stacks. 
\end{example}

\begin{remark} \label{B_twist_u_family_remark}
Note that the $(-1)$-shifted tangent complex (i.e. the classical BV complex) of the twisted family associated to the B-twist can be calculated directly. In a smooth category, it is equivalent to the $\CC[u,u^{-1}]$-module 
\[E^Q_u = (\Omega^{\bullet,\bullet}(X; \gg_P \oplus \gg^*_P[-1])[u,u^{-1}], u \dd + \ol \dd).\]
\end{remark}

\begin{remark}
Note that 
\[\OO(\EOM_{A,G^\vee}(\pt)) \iso \mr{H}^\bullet_{\mr{dR}}(BG^\vee) \iso \sym( (\hh^\vee)^*[-2] )^W \iso \sym( \hh[-2] )^W \iso \OO( \hh^*[2]/W )\]
which is the same as $\OO( \EOM_{B,G}(\pt) ) = \OO(\hh^*[2]/W  )$ at the classical level.  One might think of this as a coincidence since there is no reason to expect such an equivalence at the classical level.  However, as we will see in Proposition \ref{no_deformations_prop}, these algebras of local classical observables do not admit any quantum corrections, so in fact we have an equivalence between the algebras of \emph{quantum} observables in the dual theories.  There \emph{is} a physical reason to expect a relationship between the algebras of local quantum observables in the A- and B-twisted theories, namely the fact that the two quantum field theories are expected to be S-dual to one another. 
\end{remark}

\subsection{Moduli Space of Flat Bundles and Hecke Symmetry}\label{geometry_of_flat_section}

In this section we'll recall some key properties of the derived stack $\flat_G(C)$ of flat $G$-bundles on a smooth complex curve $C $.  In the local case, where we consider a formal bubble $\bb B = \bb D \sqcup_{\bb D^\times} \bb D$ instead of $C $, we'll also recall the derived statement of the geometric Satake correspondence due to Bezrukavnikov and Finkelberg \cite{BezrukavnikovFinkelberg}.  We will deliberately omit the most technical details from derived algebraic geometry: instead we refer the reader to Sections 10 to 12 of \cite{ArinkinGaitsgory} for a fully detailed account.

\subsubsection{Geometry of $\mathrm{Flat}_G$} \label{flat_G_section}

We'll begin by describing the tangent and cotangent complexes of $\flat_G(C)$.  Recall that $\flat_G(C)$ is defined as a derived mapping stack: that is, as 
\[\flat_G(C) = \ul{\mr{Map}}(C_{\dR}, BG).\] 
We can, alternatively describe its $R$-points directly.  Let $\spec R$ be an affine derived scheme.  The $R$-points of $\flat_G(C)$ are given by the simplicial set of algebraic $G$-bundles on $\spec R\times C $ with a partial flat connection in the direction of $C$.  That is, they're given by the simplicial set $\mr{Map}(\spec R \times C_{\dR}, BG)$.  The \emph{tangent complex} at an $R$ point $(P, \nabla)$ can be identified as the complex
\begin{align*}
\bb T_{(P, \nabla)} \flat_G(C) &\iso (\Gamma(\spec R \times C _{\mr{dR}}; \gg_P)[1], \d_R) \\
&\iso (\Omega^\bullet(\spec R \times C ; \gg_P)[1], \d_{\nabla} + \d_R)
\end{align*}
where $\d_{\nabla}$ is the covariant derivative in the direction of $C$ and $\d_R$ is the internal differential on $R$ (or, more precisely,the differential on the dg $R$-module determined by $P$).  Similarly, the \emph{cotangent complex} can be described as
\[\bb L_{(P, \nabla)} \flat_G(C) \iso (\Omega^\bullet(\spec R \times C ; \gg^*_P)[1], \d_{\nabla} + \d_R).\]
Since $C$ is smooth and proper and $BG$ is 2-shifted symplectic, the derived stack $\flat_G(C)$ is 0-shifted symplectic by the AKSZ argument (i.e. by Theorem \ref{PTVV_AKSZ}). The shifted symplectic structure is easy to identify from this point of view: it's the isomorphism between the tangent and cotangent complexes induced from the Killing form isomorphism $\gg \iso \gg^*$.  We can also easily see the following, since for any $R$ the cohomology of the tangent complex is concentrated in degrees less than or equal to 1.

\begin{prop}
 The stack $\flat_G(C)$ is quasi-smooth.
\end{prop}

Let's now consider the shifted cotangent \emph{spaces}, i.e. the derived stacks $T^*[k]\flat_G(C)$, especially in the cases where $k = \pm 1$.  If $\spec R$ is a \emph{classical} affine scheme, an $R$-point of the classical part $(T^*[k]\flat_G(C))^{\mr{cl}}$ is an $R$-point of $\flat_G(C)$ -- i.e. a $G$-bundle $P$ on $C  \times \spec R$ with a partial flat connection $\Delta$ in the $C $ direction -- along with an element of $\mr H^{1+k}_\nabla(\spec R \times C ; \gg^*_P)$.  If $k=-1$, the classical truncation of this space is nothing but the \emph{scheme of singularities} that we defined in Definition \ref{scheme_sings_def}.  The scheme of singularities of $\flat_G(C)$ has another name.

\begin{definition}\cite[Section 1.1.5]{ArinkinGaitsgory}
The space of \emph{Arthur parameters} for the group $G$ and curve $C$ is the classical part of the $-1$-shifted cotangent stack, that is, $\mr{Arth}_G(C) := (T^*[-1]\flat_G(C))^{\mr{cl}}$, whose $R$-points consist of a $G$ bundle $P$ on $C  \times \spec R$ with a partial flat connection in the $C $ direction and a flat section $\phi$ of the coadjoint bundle $\gg^*_P$.
\end{definition}

The refined statement of the geometric Langlands conjecture involves sheaves with singular support in the following substack of $\arth_G(C)$.
\begin{definition}
The \emph{global nilpotent cone} is the subspace of $\mr{Arth}_G(C)$ where the Arthur parameter $\phi$ is \emph{nilpotent}.  In other words, there is a map $\mr{Arth}_G(C) \to \gg^*/G$ given by evaluating the flat section $\phi$ at a point $x \in C $, and one defines the global nilpotent cone as the fiber product
\[\mc N_G = \mr{Arth}_G(C) \times_{\gg^*/G} \mr{Nilp}/G \]
where $\mr{Nilp}$ is the nilpotent cone in $\gg^*$, which is defined as the image of the nilpotent cone in $\gg$ under the $G$-invariant isormorphism $\gg \cong \gg^*$ \footnote{We fixed such an identification when we specified the AKSZ shifted symplectic structure above: the AKSZ shifted symplectic structure required choosing a 2-shifted symplectic structure on $BG$: this structure is equivalent to the choice of a $G$-invariant isormorphism $\gg \cong \gg^*$.}.
\end{definition}

Note that this is independent of the choice of a point $x \in C $ because $\phi$ is a flat section.

\begin{remark}
In the original Langlands program, Langlands reciprocity involves not just a Galois representation but also an action of the so-called ``Arthur $\SL_2$''. In the geometric story, following Arinkin and Gaitsgory, we call the entire singularity space the space of Arthur parameters. In this case, allowing only nilpotent singular support condition corresponds to the action of a nilpotent element of $\SL_2$, so having the nilpotent singular support condition corresponds to having an Arthur parameter in the classical sense. 
\end{remark}

We'll conclude this subsection by noting that $\flat_G(C)$ forms a global complete intersection stack, in the sense that we described in Section \ref{main_SS_section}.  In brief, if one chooses a point $x \in C$, one can write $\flat_G(C)$ as a derived pullback $\flat_G^{\mr{RS}}(C) \times_{\gg/G} BG$, where $\flat_G^{\mr{RS}}(C)$ models flat $G$-bundles with regular singularity at the point $x$, and the map to $\gg/G$ sends such a flat bundle to its residue at $x$.  In other words, there is a derived pullback square (see also diagram \cite[(12.11)]{ArinkinGaitsgory}):
\[\xymatrix{
\flat_G(C) \ar[r] \ar[d] &\flat_G^{\mr{RS}}(C) \ar[d] \\
BG \ar@<1ex>[r]&\mc \gg/G. \ar@<1ex>[l]
}\]
For our purposes, we only need to know the existence of this pullback square in order to apply Proposition \ref{SS_complete_intersection_prop}.  So we need to understand the groupoid associated to the pullback square, and its action on $\flat_G(C)$, but we won't need to carefully understand the stack $\flat_G^{\mr{RS}}(C)$.  The important fact is, while the stack $\flat_G^{\mr{RS}}(C)$ isn't quite smooth, it is smooth in a Zariski neighborhood of $\flat_G(C)$ \cite[Proposition 10.6.6]{ArinkinGaitsgory}.  

The groupoid associated to this diagram -- which we'll use to describe singular support conditions -- is the groupoid $\mr{Hecke}^{\mr{spec}}_G = BG \times_{\gg/G} BG$.  We'll now discuss the action of this groupoid on $\flat_G(C)$ -- and the action of the category $\IC(\mr{Hecke}^{\mr{spec}}_G)$ -- in the context of the geometric Satake correspondence.

\subsubsection{Geometric Satake Equivalence} \label{satake_section}
This is one of the few parts of the paper where the group $G$ and its Langlands dual group $G^\vee$ can appear in the same expression.  In our notation we'll usually just write $G$, and write $G^\vee$ for the automorphic side (A-side), and $G$ for the Galois side (B-side) only when they are compared. Note that this notation is in contrast to the choice most often appearing in the geometric Langlands literature, (for instance, \cite[Section 12]{ArinkinGaitsgory}), but we prefer this choice because of our extensive focus on the Galois side.

Consider the affine Grassmannian $\mr{Gr}_{G,x} = G(\KK_x) / G(\OO_x) $ for a point $x\in C $. One of the main characters of the geometric Langlands theory is the \emph{spherical Hecke category} at a point $x\in C $ defined by $\mr{Sph}_{G,x} = \mr{D}(\bun_G(\bb B ))$, where $\bb B = \bb D \sqcup_{\bb D^\times} \bb D$ for a formal disk $\bb D$ and a formal punctured disk $\bb D^\times$. More precisely, we think of $\bun_G(\bb B)$ as the double quotient stack $G(\OO_x)\backslash G(\KK_x) / G(\OO_x)$ and $\mr{D}(\bun_G(\bb B )) = \mr{D}_{G(\OO_x)}( \mr{Gr}_{G,x})$ as the category of $G(\OO_x)$-equivariant D-modules on the affine Grassmannian $\mr{Gr}_{G,x}$. Its monoidal structure is given by the pull-push diagram 
\[\xymatrix{
&G(\OO_x)  \backslash  G(\KK_x) / G(\OO_x) & \\
&G(\OO_x) \backslash G(\KK_x) \times_{G(\OO_x)} G(\KK_x) / G(\OO_x) \ar[dl]_{p_1} \ar[dr]^{m} \ar[u]_{p_2} &  \\
G(\OO_x)  \backslash G(\KK_x) / G(\OO_x) &&  G(\OO_x)  \backslash  G(\KK_x) / G(\OO_x),
}\] 
that is by $M_1 \ast M_2 = m_*(p_1^*M_1 \otimes p_2^* M_2)$.  There are some technical issues involved in working with D-modules on general infinite-dimensional spaces, but one can make sense of them in this case by writing the affine Grassmannian as an ind-finite type ind-scheme which is $G(\OO_x)$-invariant at each level (for instance this guarantees that the spherical Hecke category is compactly generated).

For the na\"ive geometric Satake equivalence, we consider the heart $\mr{Sph}_{G,x}^\heartsuit := \mr{D}_{G(\OO_x)}( \mr{Gr}_{G,x} ) ^\heartsuit$ of the natural t-structure which is the abelian category of $G(\OO_x)$-equivariant D-modules on $\mr{Gr}_{G,x}$. It is a theorem of Lusztig \cite{Lusztig} that the monoidal structure described above preserves the heart $\mr{Sph}_{G,x}^\heartsuit$. The geometric Satake theorem in this context refers to the following fundamental theorem relating the group $G$ and its Langlands dual group $G^\vee$.

\begin{theorem}[Ginzburg \cite{Ginzburg2}, Mirkovic--Vilonen \cite{MV}] 
There exists a canonical equivalence of monoidal 1-categories \[\mr{Sat}_x^\heartsuit \colon \mr{Sph}_{G^\vee,x}^{\heartsuit} \simeq \mr{Rep}(G)^\heartsuit.\]
\end{theorem}

However, one should not immediately try to promote this to an equivalence of derived categories. For the unit object $\delta_1$ given by the delta function at $1 \in \mr{Gr}_{G,x}$, one computes its endomorphisms in the derived category to be
\begin{align*}
\mr{Hom}_{ \mr{Sph}_{G,x} } (\delta_1,\delta_1) &= \mr{End}_{\mr{D}(\pt)^{G(\OO_x)} }(\delta,\delta) \\ 
&\iso \mr{H}^\bullet(BG(\OO_x))\\ 
&\iso \mr{H}^\bullet(BG).
\end{align*} 
On the other hand, since $\mr{Rep}(G)$ for a reductive group $G$ is a semisimple category, one has $\mr{Hom}_{\mr{Rep}(G)}(1,1) =\bb C$. Therefore in order to sensibly promote the geometric Satake equivalence to an equivalence of derived categories one needs to enhance the right-hand side, so that one has an equivalence encoding all the higher Ext information on the left-hand side.

In order to do that, it is convenient to enlarge the category $\mr{Sph}_{G,x}$. This process is usually called ``renormalization'' (see e.g. \cite{ArinkinGaitsgory, BZGunningham}). We consider the full subcategory $\mr{Sph}_{G,x }^{ \mr{loc,c} } \inj \mr{Sph}_{G,x}$ of locally compact objects, that is, those objects which become compact, or equivalently coherent, under the forgetful functor $\mr{D}_{G(\OO_x)}( \mr{Gr}_{G,x} ) \to \mr{D}( \mr{Gr}_{G,x})$. Note that it is closed under the convolution monoidal structure, since the pullback and pushforward functors $p_i^*$ and $m_*$ defining the monoidal structure preserve coherence. Then we define the renormalized spherical Hecke category $\mr{Sph}_{G,x }^{\mr{ren}}$ to be the ind-completion of $\mr{Sph}_{G,x }^{\mr{loc,c}}$.

On the other side of the geometric Satake equivalence, one considers the dg category $\IC(\mr{Hecke}^{\mr{spec}}_G)$, where $\mr{Hecke}^{\mr{spec}}_G$ is the derived fiber product $BG \times_{\gg/G} BG$ (the groupoid mentioned in the previous subsection).  We can think of this groupoid as modelling flat $G$-bundles on the formal bubble $\bb B$: i.e. as modelling pairs of flat $G$-bundles on the formal disk $\bb D$ that agree with restricted to the formal punctured disk.  In other words
\begin{align*}
\mr{Hecke}^{\mr{spec}}_G &= BG \times_{\gg/G} BG \\
&\iso \flat_G(\bb D) \times_{\flat^{\mr{RS}}_G(\bb D^\times)} \flat_G(\bb D) \\
&\iso \flat_G(\bb D \sqcup_{\bb D^\times} \bb D) \\
&= \flat_G(\bb B).
\end{align*}
Here we're identifying $\gg/G$ with the stack $\flat^{\mr{RS}}_G(\bb D^\times)$ of flat $G$-bundles on $\bb D^\times$ with regular singularity at 0. The category $\IC(\mr{Hecke}^{\mr{spec}}_G)$ has a convolution monoidal structure whose unit is the skyscraper sheaf at the trivial flat bundle.

Here is what is referred to as the derived geometric Satake equivalence. 

\begin{theorem}[Bezrukavnikov--Finkelberg \cite{BezrukavnikovFinkelberg}]
There is an equivalence of monoidal dg categories 
\[\mr{Sph}_{G^\vee,x }^{\mr{ren}}  \simeq \IC(\mr{Hecke}^{\mr{spec}}_G).\]
\end{theorem}

The importance of these monoidal categories is that they act on the categories of interest on any curve $C $ as follows. Let us explain this point for the Galois side (the corresponding action on the automorphic side is defined verbatim). If we choose a point $x\in C $, then there is a natural action of the monoidal category $ \IC(\flat_G(\bb B))$ on the category $ \IC(  \Flat_G(C) )$ by pull-tensor-push along the natural maps
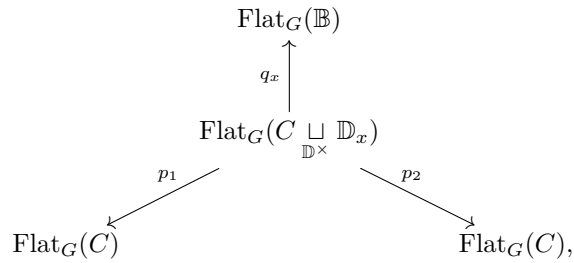
\[\xymatrix{
&\flat_G(\bb B)&\\
&\flat_G(C  \underset {\bb D^\times} \sqcup \bb D_x) \ar[u]^{q_x} \ar[dl]_{p_1} \ar[dr]^{p_2} &\\
\flat_G(C) &&\flat_G(C),
}\] 
all of which are induced from the natural inclusion maps on the corresponding curves; note that there are two different inclusion maps from $C $ to $C  \sqcup_{\bb D^\times} \bb D_x$.  These inclusion maps are shown in Figure \ref{Hecke_curve_fig}.

 \begin{figure} [!h]
 \centering
\begin{tikzpicture}[scale=.7]
\begin{scope}[yshift=3.5cm]
\draw (0,0) circle (.6);
\draw (-.6,0) .. controls (0,-0.2) .. (.6,0);
\draw[dashed] (-.6,0) .. controls (0,0.2) .. (.6,0);
\end{scope}

\begin{scope}
\draw[rounded corners=28pt] (-1.1,.1)--(0,-.6)--(1.1,.1);
\draw[rounded corners=24pt] (-.9,0.17)--(0,.6)--(.9,0.17);
\draw(0,0) ellipse (3 and 1.5);
\draw[fill=white] (0,1.1) circle (.6);
\draw (-.6,1.1) .. controls (0,0.9) .. (.6,1.1);
\draw[dashed] (-.6,1.1) .. controls (0,1.3) .. (.6,1.1);
\end{scope}

\begin{scope}[yshift=-2cm,xshift=-7cm]
\draw[rounded corners=28pt] (-1.1,.1)--(0,-.6)--(1.1,.1);
\draw[rounded corners=24pt] (-.9,0.17)--(0,.6)--(.9,0.17);
\draw(0,0) ellipse (3 and 1.5);
\draw (-.6,1.1) .. controls (0,0.9) .. (.6,1.1);
\draw[dashed] (-.6,1.1) .. controls (0,1.3) .. (.6,1.1);
\clip (-.6,1.1) rectangle (.6,.5);
\draw (0,1.1) circle (.6);
\end{scope}

\begin{scope}[yshift=-2cm,xshift=7cm]
\draw[rounded corners=28pt] (-1.1,.1)--(0,-.6)--(1.1,.1);
\draw[rounded corners=24pt] (-.9,0.17)--(0,.6)--(.9,0.17);
\draw(0,0) ellipse (3 and 1.5);
\draw (-.6,1.1) .. controls (0,0.9) .. (.6,1.1);
\clip (-.6,1.1) rectangle (.6,1.7);
\draw[fill=white] (0,1.1) circle (.6);
\draw[dashed] (-.6,1.1) .. controls (0,1.3) .. (.6,1.1);
\end{scope}

\draw[<-] (0,2) -- (0,2.5);
\draw[<-] (-2.9,-.8) -- (-3.9,-1.3);
\draw[<-] (2.9,-.8) -- (3.9,-1.3);

\end{tikzpicture}
\caption{The curve $C$ includes into $C \sqcup_{\bb D^\times} \bb D_x$ -- the curve with the formal neighborhood of the point $x$ doubled -- in two different ways.  There is also an inclusion of the formal bubble $\bb D \sqcup_{\bb D^\times} \bb D$.}
\label{Hecke_curve_fig}
\end{figure}
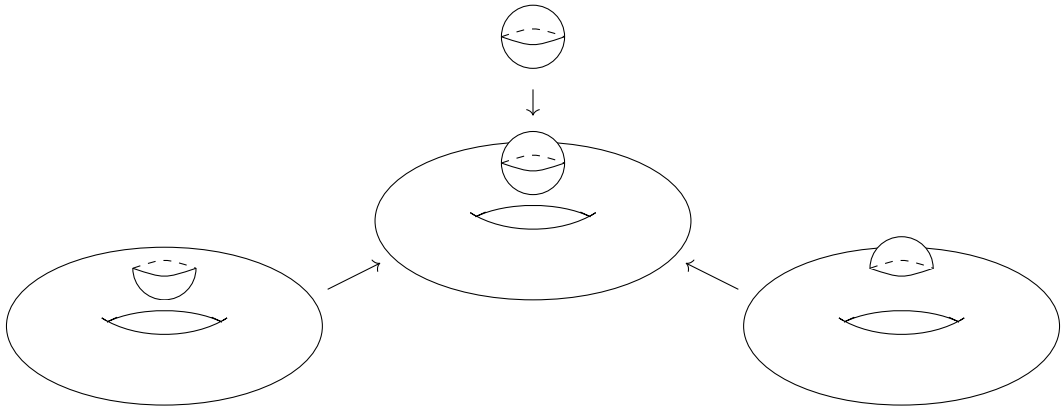

\begin{remark}
Indeed, this action preserves the subcategory $ \IC_{\mc N_{G}}(\Flat_G(C)) \subset  \IC(\Flat_G(C)) $ and Arinkin and Gaitsgory's version of the geometric Langlands correspondence asserts an equivalence of two dg categories \[D(\bun_{G^\vee}(C))\simeq \IC_{\mc N_{G}}(\Flat_G(C))\] as module categories over those equivalent monoidal categories.
\end{remark}

\subsection{Vacua for Twisted $\mc N=4$ Theories} \label{restricting_KW_theories_section}

Now we'd like to apply the constructions of Section \ref{loc_and_vacuum_section} to examples modelling the Kapustin--Witten twisted 4d gauge theories described in Section \ref{twisting_section}. According to the recipe, we need to specify an $\bb E_4$-algebra $\mc A$ of 4d local observables and a dg category $\mc B$ with an action of $\mc A$ modelling \emph{boundary conditions} along the curve $C $.  We'll build this data starting from a dg category $\mc B$ quantizing the space of solutions $\EOM(C)$ along a curve in the twisted theory, which is a module for a monoidal dg category $\LL$ quantizing the local space $\EOM(\bb B)$, so that $\eend_{\LL}(1_{\LL})$ quantizes the factorization algebra of local classical observables.

For most of this section we'll focus on the B-twisted theory (we'll discuss what is known in the A-twisted case in Remark \ref{rmk:A-twist}).  Set
\[\mc B = \IC(\flat_G(C))\]
viewed as the (na\"ive) categorical geometric quantization of the moduli space $\EOM_B(C) = T^*_{\mr{form}}[1]\flat_G(C)$ as discussed in Section  \ref{functorial_background_section}.  Likewise set $\LL = \IC(\flat_G(\bb B))$, which acts on $\mc B$ after choosing a point $x \in C $ as discussed above.

The action of $\LL$ on $\mc B$ leads to a map of $\bb E_2$-algebras $\mr{End}_{\LL}(1_\LL)\to \mr{HC}^\bullet(\mc B)$ as in Section \ref{restriction_background_section}. For the unit $1_\LL = \delta$ of the monoidal category $\LL = \IC(\flat_G(\bb B))$, one can compute its endomorphism algebra.
\begin{lemma}[{\cite[Section 12.4]{ArinkinGaitsgory}}] \label{endomorphism_unit_hecke_lemma}
There is an equivalence 
\[\hom_{\LL}(\delta,\delta) \iso \OO(\gg^*[2]/G) \iso \OO(\hh^*[2]/W)\]
as $\bb E_2$-algebras, where the algebra on the right-hand side is actually commutative, and therefore $\bb E_2$ in a trivial way.
\end{lemma}
Therefore there's a natural map from $\mr{End}_{\LL}(1_\LL) \iso \OO(\hh^*[2]/W)$ to the Hochschild cochains $\mr{HC}^\bullet(\mc B)$ . 

\begin{remark} \label{scheme_direction_HH_remark}
Heuristically, if $\flat_G(C)$ were a scheme, as opposed to a stack, then the HKR theorem would identify its Hochschild cohomology with $\OO(T^*[1]\flat_G(C))$.  In these terms the action of $\OO(\hh^*[2]/W) = \OO(\gg^*[2]/G)$ would be given by pullback along the map
\[\xymatrix{ T^*[1]\flat_G(C) \ar[r]^-{\ev_x} & \gg^*[2]/G,}\]
defined by evaluating the section at a point $x$ in $C $ to obtain an element of $\gg^*$ well-defined up to conjugation, where the degree $-2$ part of the cotangent fiber parametrizes sections of the coadjoint bundle.
\end{remark}

We need to quantize the $\bb P_4$-algebra of classical observables -- the $\bb P_4$ structure associated to the 3-shifted symplectic structure on $\gg^*[2]/G$ -- to an $\bb E_4$-algebra.  We'll do this using a very simple example of a general body of theory about $\bb E_n$-deformations.  We refer to To\"en's ICM address \cite{ToenICM} and Francis \cite{Francis} for details.

\begin{prop} \label{no_deformations_prop}
The commutative algebra $\OO(\hh^*[2]/W)$ with its trivial 3-shifted Poisson bracket doesn't admit any $\bb P_4$ deformations. 
\end{prop}

\begin{proof}
The $\bb P_4$ deformations of the given $\bb P_4$-algebra structure on $\OO(\hh^*[2]/W)$ correspond to Maurer--Cartan elements in the dg Lie algebra whose underlying graded vector space is the homotopy fiber of the natural map $\mr{Pol}(\hh^*[2]/W, 3)[4]\to \OO(\hh^*[2]/W  )[4]$. As $\mr{Pol}(\hh^*[2]/W, 3)[4]= \OO( T^*[4]( \hh^*[2]/W))$ is a symmetric algebra freely generated by elements of cohomological degree 2 and the natural map surjects onto $\OO(\hh^*[2]/W  )[4]$, there are no elements in degree 1 in the corresponding dg Lie algebra, so the claim immediately follows.
\end{proof}

Therefore by Theorem \ref{thm:Toen} and the discussion immediately following it, we should take $\mc A = \OO(\hh^*[2]/W)$ with the commutative algebra structure as the $\bb E_4$-algebra of quantum local observables of our theory.  Before defining the restriction, however, we'll need to perform this calculation in a slightly more sophisticated fashion which remembers the origin of the theory in question as a topological twist.

\begin{remark}
There's a natural interpretation of this action in terms of the map from 4d local operators to 2d local operators (after compactification along $C$) for the B-twisted theory.  Indeed, the classical B-twisted theory assigns to a closed curve $C$ the 1-shifted symplectic stack $T^*[1]\mr{Flat}_G(C)$, or alternatively the (sheaf of) $\bb P_2$-algebras $\OO(T^*[1]\mr{Flat}_G(C))$.  At the classical level, there's a map of algebras $\OO(\gg^*[2]/G) \to  \OO(T^*[1]\mr{Flat}_G(C))$ for every point $x \in C $ which we think of as $\obscl(B^4)\to \obscl( B^2 \times C)$, where the choice of a point $x\in C $ gives a map $B^4\to B^2\times C $.  We know that the algebra $\OO(\gg^*[2]/G)$ doesn't receive any quantum corrections. Let's suppose that this is also true for $\OO(T^*[1]\mr{Flat}_G(C))$, so this map is exact at the quantum level.  We claim that this map is precisely the map into the Hochschild cohomology described in our heuristic Remark \ref{scheme_direction_HH_remark}.  Indeed, if we assume a state-operator correspondence then in the 2d theory obtained from reduction along $C$ the local operators are precisely given by the Hochschild cochains of the category assigned to the point (since these are states on the circle with an appropriate framing).  Note that this interpretation ignores some of the subtleties present in our story: the difference between Hochschild cochains and functions on the shifted cotangent bundle occuring because of the stacky directions in $\flat_G(C)$ (and, for that matter, the stacky directions present in $\gg^*[2]/G$).

\end{remark}

\subsection{Extension to the Twisted Family} \label{degree_shift_section}

Using this moduli space of vacua, we can finally connect compatibility with a vacuum and Arinkin--Gaitsgory's nilpotent singular support condition. However, in order to find a more general setting with which we can find some new structures on the geometric Langlands theory, we still need to expand this definition slightly. In fact, this extension is a meaningful process in greater generality, but we mostly work with our main example.

In order to motivate this extension we'll return to the construction of the twisted theory.  Recall that in the perturbative setting (and hence in the nonperturbative setting as well) we took $\CC^\times$-invariants of the twisted family 
\[E ^{\Pi \CC}\otimes_{\CC[u]}  \CC[u,u^{-1}]\]
built from the B-supercharge, where $u$ is a parameter of bidegree $(1,1)$ in $\ZZ/2\ZZ \times \ZZ$. This process was taken as the definition of the B-twist. Based on this definition, we have seen that the resulting moduli space of vacua is equivalent to $\hh^*[2]/W$, as opposed to the underived space $\hh^*/W$ physicists might more naturally consider. This is not a mistake; the twisting procedure necessarily involves a meaningful degree shift; for instance, the correct cohomological grading of the Hochschild cohomology (or polyvector fields) as the local observables of the B-model topological open string theory is introduced precisely by the topological twisting construction for a 2-dimensional supersymmetric $\sigma$-model.  Without this particular choice of $\CC^\times$-action, one would not obtain the usual cohomological grading. 

On the other hand, it is possible to recover the ungraded space physicists think of as the moduli space of vacua. Our perspective will be that we needn't restrict to the $\CC^\times$-invariants for a particular choice of R-symmetry.  Twisted theories are sometimes described in this way in the physics literature; a priori twisted theories needn't admit a $\ZZ$-grading, instead they'll only be $\ZZ/2\ZZ$-graded (for instance this is the case for Baulieu's twist of 10d $\mc N=1$ super Yang--Mills \cite{Baulieu}).  One can promote a $\ZZ/2\ZZ$-graded theory to a $\ZZ$-graded theory by equipping it with a suitable $\CC^\times$-action, see for example \cite{KapustinSaulina}.  We'll discuss $\ZZ/2\ZZ$-graded theories a little more later in this subsection.

Because the twisted family is locally constant in the $u$ direction, the twisted field theory ``lives over'' $\CC[u,u^{-1}]$. For example, if we don't take $\CC^\times$ invariants, we have 
\[\obsq_u(B^4) : = \obsq(B^4)^{\Pi \CC} \otimes_{\CC[u]}  \CC[u,u^{-1}] \iso \OO(\hh^*[2]/W)[u,u^{-1}]\] 
as the full algebra of quantum observables.  Likewise, one can similarly understand the full category of boundary conditions along $C$ as 
\[\IC(\flat_G(C))[u,u^{-1}] := \IC(\flat_G(C)) \otimes_{\vect} \CC[u,u^{-1}]\text{-mod}.\]

The action of local observables in this family is the same as in the previous section.  That is, after choosing a point $x \in C$ the spectral Hecke category $\LL[u,u^{-1}] = \IC(\flat_G(\bb B))[u,u^{-1}]$ with the parameter $u$ adjoined acts on $\mc B[u,u^{-1}] = \IC(\flat_G(C))[u,u^{-1}]$, which induces a map from $\OO(\hh^*[2]/W)[u,u^{-1}]$ to the Hochschild cochains $\mr{HC}^\bullet(\mc B[u,u^{-1}])$ by taking endomorphisms of the unit.  Heuristically (as in Remark \ref{scheme_direction_HH_remark}) we can think of this as an evaluation map at the chosen point $x$.  That is, there is an evaluation map 
\[\ev_{x} \colon T^*[1]\flat_G(C) \to \gg^*[2]/G\]
which induces algebra maps
\begin{align*}
\ev_{x}^* \colon \OO(\hh^*[2]/W) &\to \OO(T^*[1]\flat_G(C)) \\
\text{and hence } \OO(\hh^*[2]/W)[u,u^{-1}] &\to \OO(T^*[1]\flat_G(C))[u,u^{-1}][s] 
\end{align*}
where now $s$ is a parameter of $\ZZ/2\ZZ \times \ZZ$ bidegree $(1,-2)$.  As before we should only think of this description as a heuristic since we can't literally compare $\OO(T^*[1]\mc X)$ and $\mr{HC}^\bullet(X)$ when $X$ is a stack and we don't have a version of the HKR theorem available.  For a precise definition of the action of local observables we'll use the description via the action of the spectral Hecke category.

\begin{remark} \label{odd_formal_disk_remark}
Geometrically we can think of the algebra $\OO(\hh^*[2]/W)[u,u^{-1}]$ of twisted local observables as the algebra of functions on a space of \emph{twisted vacua} of the form $\hh^*[2]/W \times \Pi \bb C^\times_1$, where $\Pi \bb C^\times_1 = \text{``}\spec\CC[u,u^{-1}]\text{''}$ is a ``shifted fermionic punctured line''.  Of course, this doesn't literally make sense because the algebra $\CC[u,u^{-1}]$ is not connective and hence does not define an affine derived scheme \footnote{One might try to define it as an affine stack in the sense of \cite{BZNLoopSpaces1, ToenAffine}, but we expect this will not be a productive way to think about it.  For instance there are no non-zero maps from $\CC[u,u^{-1}]$ to a connective cdga, so if one tries to define $\spec\CC[u,u^{-1}]$ as a functor in the na\"ive way, the result is the empty stack.}.  Nevertheless we think this is a useful perspective for understanding the twisting procedure -- we think of the twist as the restriction of a family of theories over a shifted fermionic line to the shifted fermionic punctured line.
\end{remark}

In order to reconnect with geometry and with the vacua story from Section \ref{loc_and_vacuum_section} we'd like to replace our nonexistent space $\hh^*[2]/W \times \Pi \bb C^\times_1$ of twisted vacua by something more meaningful.  The na\"ive thing to do from the point of view of the usual twisting story would be to take spec of the $\CC^\times$-invariants in our algebra of twisted local observables: in this case $u$ has weight 1 and $\OO(\hh^*[2]/W)$ has weight zero, so this recovers the na\"ive space $\hh^*[2]/W$ of vacua from the previous section.  Alternatively, we can get an actual scheme of vacua by considering the ordinary (that is, underived) directions in our $\CC[u,u^{-1}]$ family of derived stacks (which, of course, does ``see'' the derived directions in the algebra of twisted local observables).  In other words, we make the following definition.

\begin{definition} \label{refined_vacua_def}
The \emph{refined space of vacua} in a twisted field theory is the classical affine scheme
\[\mc V^{\mr{ref}} = \spec \mr{H}^0(\obsq_u(B^4) ) \]
obtained by taking the degree zero part of the total algebra of twisted observables as a $\CC[u,u^{-1}]$-module. 
\end{definition}

\begin{figure}[h]
\centering
\begin{tikzpicture}[scale=1]
 \draw [thin, white, shade, shading=axis, left color=lightgray, right color=white,shading angle=0] (0,0) rectangle (8,1);
 \draw [thin, white, shade, shading=axis, left color=lightgray, right color=white,shading angle=180] (0,1) rectangle (8,2);
 \draw [ultra thick] (0,1) -- (8,1);
 \draw [ultra thick] (0,0.1) -- (8,1.9);
 \draw [ultra thick, dashed, red] (2,0) -- (2,2);
 \draw [fill] (2,1) circle (0.1);
 \draw [red, fill] (2,0.55) circle (0.1);
 \node at (2.3,1.3) {$v$};
 \node [red] at (2.4,2.4) {$v[u,u^{-1}]$};
 \node at (8.5,2) {$\mc V^{\mr{ref}}$};
 \node at (8.3,1) {$\mc V$};
 \node at (-0.8,1.5) {\scalebox{1.1}{${\mc V[u,u^{-1}]}$}};
 \draw [<->] (9,0.5) -- (9,1.5);
 \node at (9.35,1) {$u^{\pm 1}$};
\end{tikzpicture}

\caption{The refined moduli space of vacua $\mc V^{\mr{ref}}$ is viewed geometrically as the classical part of the thickened space $\mc V[u,u^{-1}]$ of twisted vacua.  Here we've drawn $\mc V$ horizontally and the degree 2 direction with coordinate $t$ vertically.  We've indicated the closed subset $v[u,u^{-1}]$ inside of $\mc V[u,u^{-1}]$ associated to a point $v$ of $\mc V$ by a red dashed line.}
\label{refined_vacua_fig}
\end{figure}
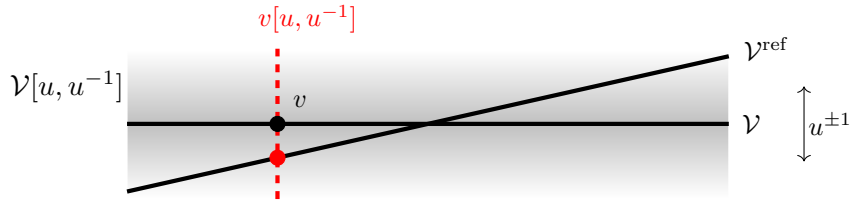

Note that in our main example, one can see that $\mc V^{\mr{ref}}$ is $\hh^*/W$ as expected. That is, if $\{x_i\}$ is the set of linear functions on $\hh^*[2]/W$, now the ideal $(x_i u^{-2})$ defines a corresponding actual point in $\hh^*/W$ (in Fig. \ref{refined_vacua_fig} this is the red point). 

\begin{remark}
In the arguments we made in Section \ref{main_SS_section} the parameter $t$ was a parameter of $\ZZ/2\ZZ \times \ZZ$ bidegree $(0,2)$ rather than a parameter of bidegree $(1,1)$.  We can pass to this purely bosonic situation by taking the bosonic part of the algebra of twisted local observables (meaning the bosonic part with respect to the super grading). Namely, if $u$ has bidegree $(1,1)$ then the bosonic part of $\OO(\hh^*[2]/W)[u,u^{-1}]$ is nothing but $\OO(\hh^*[2]/W)[t,t^{-1}]$ where now $t=u^2$ is a bosonic parameter of degree 2.  From now on we'll restrict the action of the full algebra of twisted local observables to the action of its bosonic subalgebra.
\end{remark}

Now, we can use the action of the algebra $\obsq_t(B^4)$ of twisted local observables on the category $\mc B[t,t^{-1}]$ of twisted boundary conditions to define the category compatible with a closed point $v \in \mc V^{\mr{ref}}$ in a natural way.  However, one must take care: it's a somewhat subtle issue to define a subcategory of the usual category $\mc B$ of boundary condition compatible with a closed point in $\mc V^{\mr{ref}}$.  The definition we'll give below will give a sensible dg subcategory of $\mc B[t,t^{-1}]$. To proceed further we'll use the following fact.

\begin{lemma}\label{lemma:2-periodic}
A cochain complex with an invertible degree 2 automorphism is equivalent to a $\bb Z/2$-graded cochain complex.
\end{lemma}

\begin{proof}
A $\bb Z/2$-graded cochain complex is equivalent to a 2-periodic complex, which in turn is equivalent to a complex with an invertible degree 2 automorphism.
\end{proof}

Using this, any dg category enriched over $\bb C[t,t^{-1}]$ can be canonically identified with a $\bb Z/2$-graded category, since one can think of $t$ as giving an invertible degree 2 automorphism. In particular, in the context of topologically twisted field theory where our categories of twisted boundary conditions are naturally enriched over $\bb C[t,t^{-1}]$, we can apply the lemma to obtain a $\bb Z/2$-graded category. 

\begin{definition} \label{twist_restriction_def}
The \emph{restricted category of twisted boundary conditions} $\mc B[t,t^{-1}]_{v}$ with respect to a closed point $v \in \mc V^{\mr{ref}}$ is the restriction of $\mc B[t,t^{-1}]$ at $v$ with respect to the composition $\mc V^{\mr{ref}}\to H^\bullet(\obsq_t(B^4) )\to \mc B[t,t^{-1}]$ via the embedding $\mc V^{\mr{ref}}\to H^\bullet(\obsq_t(B^4) )$ discussed in Remark \ref{conical_t_remark}. Then the \emph{category $\mc B_v$ of boundary conditions compatible with $v\in \mc V^{\mr{ref}}$} is defined to be the $\ZZ/2$-graded category associated to $\mc B[t,t^{-1}]_{v}$ using the above discussion.
\end{definition}

This definition deserves a few remarks.  The first remark is aimed at readers approaching this paper from a string-theoretic background, and can safely be ignored by readers with different interests.

\begin{remark} \label{IIB_string_remark}
One perspective on supersymmetric 4d $\mc N=4$ gauge theory with gauge group $\mr{U}_n$ is as the worldsheet theory of $n$ D3 branes in type IIB string theory, with the choice of vacuum arising as the position of these branes in the remaining six dimensions. After twisting, there are two  directions left which the branes to move in: those directions which are invariant for the topological supercharge.  This twisted theory on a complex surface $X$ with group $\GL_n$ corresponds to the worldsheet theory of $n$ D3 branes wrapping $X$ in a certain twist of type IIB string theory on the background $T^*X \times\bb C$ (for the notion of a twist of a supergravity theory, one should refer to Costello and Li \cite{CostelloLiSUGRA}) and the two invariant directions are the two real directions corresponding to the factor of $\bb C$. However, the residual R-symmetry $\CC^\times$ acts non-trivially by rotating this complex plane $\CC$.  As such taking the full R-symmetry invariants would freeze any possible motions of D3 branes, which is unnatural from the point of view of string theory. The R-symmetry $\CC^\times$ is the unique copy of $\CC^\times$ which makes the theory as $\bb Z$-graded theory, as tautologically it is responsible for cohomological grading of the theory.

In sum, as long as we only care about one particular theory, which corresponds to looking at a given configuration of D3 branes, there is a unique way of making things $\bb Z$-graded using the original $\CC^\times$-action, but once we start to consider several different theories on an equal footing by moving D branes around, which is inevitable from the perspective of string theory, we should work with $\bb Z/2$-graded theories.
\end{remark}

\begin{remark}
On the other hand the geometric Langlands program (and indeed derived algebraic geometry as a whole) is developed centered around $\bb Z$-graded objects. There are two complementary solutions we take. One is that for a fixed theory, one can just work with $\mc V$ rather than $\mc V^{\mr{ref}}$ to obtain a $\ZZ$-graded theory. For instance, if we work with $0\in \hh^*[2]/W$, this would define a well-defined dg category, and this is what we are going to do to recover the category Arinkin and Gaitsgory introduced for the geometric Langlands correspondence. Another is just to work with the dg category $\mc B[t,t^{-1}]_v$ for $v\in \mc V^{\mr{ref}}$ throughout.
\end{remark}

We'll conclude this section with the key result describing the category of boundary conditions compatible with an arbitrary vacuum, and in particular recovering Arinkin and Gaitsgory's categories of sheaves with nilpotent singular support from the compatibility with the zero vacuum state.  In order to state this we'll first introduce a piece of notation.

\begin{definition}
We write $\mr{Arth}^v_G(C) \to \mr{Arth}_G(C)$ for the pullback $\mr{Arth}_G(C) \times_{\gg^*/G} \tilde{v}$, where $\tilde v$ is the image in $\gg^*/G$ of the pullback of $v \in \hh^*/W$ under the Chevalley map $\gg^* \to \hh^*/W$ (so in particular if $v=0$ we recover the global nilpotent cone mod $G$)\footnote{If $G$ is semisimple we can define $\wt v$ as the pullback of $v$ under $\gg^*/G \to \hh^*/W$, but this map doesn't quite make sense if $G$ is abelian.}. 
\end{definition}

\begin{remark}
We'll see in the next section that the subspace $\mr{Arth}^v_G(C)$ is not so geometrically mysterious.  It's supported over the locus in $\flat_G(C)$ consisting of $G$-bundles that break to the Levi subgroup $L$ centralizing $v$.  For example, in the generic case where $v$ is a \emph{regular} element in $\hh^*/W$, the only flat $G$-bundle admitting a flat section conjugate to $v$ is the trivial bundle (that is, the unique completely reducible flat bundle, where the gauge symmetry breaks to the maximal torus $H$).  Over this locus, $\mr{Arth}^v_G(C)$ is just a copy of the global nilpotent cone for $L$ translated by the constant section $v$, viewing $v$ as an element of the center of the Lie algebra of $L$.
\end{remark}

\begin{theorem} \label{main_theorem}
The restriction $\IC(\flat_G(C))[t,t^{-1}]_v$ of the category of twisted boundary conditions to the refined vacuum $v \in \hh^*/W$ is equivalent to the category $\IC_{\arth^v_G(C)[t,t^{-1}]}(\flat_G(C))[t,t^{-1}]$.  Moreover, the restriction $\IC(\flat_G(C))_0$ of the category of boundary condition at the vacuum $0 \in \hh^*[2]/W$ (in the $\ZZ$-graded sense) is equivalent to the category $\IC_{\mc N_G}(\flat_G(C))$ of sheaves with nilpotent singular support. 
\end{theorem}

\begin{proof}
We've set things up so that we can apply Proposition \ref{SS_complete_intersection_prop}.  On one side of the claimed equivalence we have the restricted category of twisted boundary conditions:
\[\IC(\flat_G(C))[t,t^{-1}]_v = \IC(\flat_G(C))[t,t^{-1}] \otimes_{\OO(\hh^*/W)\text{-mod}} \OO(\hh^*/W)\text{-mod}_{v}.\]
Here, the action of the algebra $\OO(\hh^*/W)$ is defined using the action of the spectral Hecke category.  That is, it's defined through the embedding of $\OO(\hh^*/W)$ into $\OO(\hh^*[2]/W)[t,t^{-1}]$ as explained in Remark \ref{conical_t_remark}, where $\OO(\hh^*[2]/W)$ is identified with the endomorphism algebra for the unit object of the spectral Hecke category by Lemma \ref{endomorphism_unit_hecke_lemma}, (and where $t$ acts by multiplication).  In particular we can make the identification
\begin{align*}
&\IC(\flat_G(C))[t,t^{-1}] \otimes_{\OO(\hh^*/W)\text{-mod}} \OO(\hh^*/W)\text{-mod}_{v} \\
&\iso \IC(\flat_G(C))[t,t^{-1}]   \otimes_{\OO(\hh^*[2]/W)[t,t^{-1}]\text{-mod}} \OO(\hh^*[2]/W)[t,t^{-1}]\text{-mod} _{v[t,t^{-1}]}.
\end{align*}

On the other side of the equivalence we have the category of sheaves with a singular support condition.  In order to apply Proposition \ref{SS_complete_intersection_prop} we identify $\flat_G(C)$ as a global complete intersection stack as we explained in Section \ref{flat_G_section}.  Explicitly, in the setup for Proposition \ref{SS_complete_intersection_prop}
\[\xymatrix{
\mc Z \ar[r] \ar[d] &\mc U \ar[d]^*+<4em>{^{=}}  &\flat_G(C) \ar[r] \ar[d] &\flat^{\mr{RS}}(C) \ar[d] \\
\mc X\ar@<0.6ex>[r]&\mc V \ar@<0.6ex>[l] &BG \ar@<1ex>[r]&\mc \gg/G. \ar@<1ex>[l]
}\]
In the notation of the Proposition, the groupoid $\mc G_{\mc X/\mc V}$ is identified with the spectral Hecke stack $\mr{Hecke}_G^{\mr{spec}}$, the category $\mr{HC}^\bullet(\mc X / \mc V)[t,t^{-1}]^{\mr{op}}\text{-mod}$ is the spectral Hecke category -- isomorphic to ${\OO(\hh^*[2]/W)[t,t^{-1}]\text{-mod}}$.  Now, we'll start with the category of twisted boundary conditions.  By the above dicussion we can identify the restriction as a tensor product over the category $\mr{HC}^\bullet(\mc X / \mc V)[t,t^{-1}]^{\mr{op}}\text{-mod}$.  That is, so far we've established that
\begin{align*}
 \IC(\flat_G(C))[t,t^{-1}]_v &= \IC(\flat_G(C))[t,t^{-1}] \otimes_{\OO(\hh^*/W)\text{-mod}} \OO(\hh^*/W)\text{-mod}_{v}\\
& \iso \IC(\flat_G(C))[t,t^{-1}]   \otimes_{\OO(\hh^*[2]/W)[t,t^{-1}]\text{-mod}} \OO(\hh^*[2]/W)[t,t^{-1}]\text{-mod} _{v[t,t^{-1}]}\\
& \iso \IC(\flat_G(C))[t,t^{-1}]   \otimes_{\mr{HC}^{\bullet}(\mr{Hecke}_G^{\mr{spec}} )[t,t^{-1}]\text{-mod}} \mr{HC}^{\bullet}(\mr{Hecke}_G^{\mr{spec}})[t,t^{-1}]\text{-mod}_{\chi^{-1}(v)[t,t^{-1}]}.
\end{align*}
In the last step, $\chi$ is the Chevalley map $\gg^*/G \to \hh^*/W$, and we're using the isomorphism $\OO(\gg^*/G) \iso \OO(\hh^*/W)$ identifying the ideal associated to the closed subset $\chi^{-1}(v)$ inside of $V = \gg^*/G$ with the ideal associated to the closed subset $v$ inside of $\hh^*/W$.

We're now in the situation where Proposition \ref{SS_complete_intersection_prop} can be applied.  We use support conditions corresponding the subset $\wt Y = \chi^{-1}(v)$ inside of $V = \gg^*/G$, so that $Y = \mr{Arth}_G^v(C)$.  This yields the equivalence
\begin{align*}
 &\IC(\flat_G(C))[t,t^{-1}]   \otimes_{\mr{HC}^{\bullet}(\mr{Hecke}_G^{\mr{spec}} )[t,t^{-1}]\text{-mod}} \mr{HC}^{\bullet}(\mr{Hecke}_G^{\mr{spec}})[t,t^{-1}]\text{-mod}_{\chi^{-1}(v)[t,t^{-1}]} \\
 &\quad \iso \IC_{ \mr{Arth}^v_G[t,t^{-1}] } ( \flat_G(C) )[t,t^{-1}],
\end{align*}
and so
\[\IC(\flat_G(C))[t,t^{-1}]_v \iso \IC_{ \mr{Arth}^v_G[t,t^{-1}] } ( \flat_G(C) )[t,t^{-1}],\]
as required.  In the case of $v=0\in \hh^*[2]/W$, by considering the natural map $\OO(\hh^*[2]/W)\to \mr{HC}^\bullet( \Flat_G(C) )$, we obtain a dg category $\IC_{\mc N_G}(\flat_G(C))$ as required. 
\end{proof}

\begin{remark}\label{rmk:supportandvacuum}
The fact that nilpotent singular support condition arises from compatibility with the point $0\in \hh^*/W$ was known to Ben-Zvi and Nadler up to an even grading shift in the Betti context \cite[Remark 3.16]{BZNBettiLanglands}. 
\end{remark}

\begin{remark}\label{rmk:A-twist}
What can we say for the A-twisted theory?  All our definitions, and many of our calculations still make sense in that context; for instance one can still check that the algebra of local observables in the A-twisted theory with gauge group $G$ is equivalent to $\OO(\hh[2]/W)$ (this is easiest to see as the endomorphism algebra of the unit in the category $\LL$ of line operators, which is the spherical Hecke category $D_{G[[z]]}(\Gr_G)$ of equivariant D-modules on the affine Grassmannian).  The degree shifting trick still makes sense, and one can still abstractly define the category of boundary conditions restricted to a point in $\hh/W$.  However, there's now a technical issue that one must deal with: the ambient category $\mc B$ one works in on the A-side -- the category of boundary conditions without imposing a vacuum condition -- should be equivalent (under geometric Langlands) to the \emph{full} category of ind-coherent sheaves on the B-side.  In particular it should be strictly larger than the category of D-modules on $\bun_G(C)$.  \vspace{-0.5em}

In the local story -- that of geometric Satake -- the category one considers is the ind-completion of the category of equivariant D-modules on $\Gr_G$ which are coherent after forgetting the equivariant structure (see Bezrukavnikov--Finkelberg \cite{BezrukavnikovFinkelberg}).  In Arinkin--Gaitsgory this is called the ``renormalized'' spherical Hecke category.  There are several ways one might try to globalize this definition, and it's not clear which globalization is correct.  \vspace{-0.5em}

Another difficulty is that, whatever choice one makes for this renormalization, even for ordinary D-modules there isn't a nice way of accessing the Hochschild cohomology $\mr{HH}^\bullet(D(\bun_G(C)))$ as there was on the B-side, where we had access to the Hochchild cohomology via the scheme of singularities (at least if $G$ is non-abelian -- in the abelian case one can use a result of Koppensteiner \cite{Koppensteiner}).  One can abstractly define the category restricted to a vacuum using the Hecke action, but actually computing it is more complicated without this explicit description.  We hope in the future to make some sensible calculations at least for the vacuum $0 \in \hh^*/W$, where we should recover the usual category of D-modules on $\bun_G(C)$ as the restriction.

In this direction, since the first version of this paper was released, Beraldo has developed a theory of singular support through ``$\bb H(\mc Y)$-module categories'' \cite{Beraldo1, Beraldo2}, where $\bb H(\mc Y)$ for a quasi-smooth derived stack $\mc Y$ is a monoidal dg category that acts as a refined version of the category of modules for Hochschild cochains.  The geometric Satake equivalence will make a version of the category $D^{\mr{ren}}(\bun_G(C))$ of renormalized $D$-modules on $\bun_G(C)$ into a $\bb H(\flat_{G^\vee}(C))$-module, and therefore allow one to make sense of singular support conditions on the A-side of geometric Langlands.  It's possible that one could pursue a field-theoretic understanding of singular support conditions on the A-side using these techniques.
\end{remark}

\section{Consequences}\label{consequences_section}
\subsection{Gauge Symmetry Breaking Conjecture} \label{sym_breaking_section}

We've argued that requiring compatibility with the vacuum $0 \in \hh^*[2]/W$ in the B-twisted $\mc N=4$ theory naturally gives us the category $\IC_{\mc N_G}(\flat_G(C))$ considered by Arinkin and Gaitsgory and described the category of boundary conditions compatible with arbitrary vacuum $v\in \hh^*/W$. However, the latter category is less geometric, possibly involving some $\ZZ/2$-graded derived algebraic geometry, and it is natural to ask to find a more geometric description of the category in terms of $\bb Z$-graded derived algebraic geometry.

We conjecture that there is a clean answer in terms of the usual objects appearing in the geometric Langlands theory, but with gauge symmetry broken to a Levi subgroup depending on $v$.

\begin{conjecture} \label{gauge_sym_breaking_conjecture}
The category $\IC_{\arth^v_G(C)[t,t^{-1}]}(\flat_G(C))[t,t^{-1}]$ obtained by restricting $\IC(\flat_G(C))[t,t^{-1}]$ at the vacuum $v \in \hh^*/W$ is equivalent to $\IC_{\mc N_L}(\flat_L(C))[t,t^{-1}]$, where $L$ is the centralizer of its semisimple lift $x$ in $\gg^*$.
\end{conjecture}

\begin{remark}
Here $\IC_{\mc N_L}(\flat_L(C))[t,t^{-1}]$ can be specialized to $\IC_{\mc N_L}(\flat_L(C))$ as a dg category with a geometric object $\mc N_L$, as one obtains the actual dg category with nilpotent singular support condition, whereas there's no immediate way to extract such a geometric object from $\IC_{\arth^v_G(C)[t,t^{-1}]}(\flat_G(C))[t,t^{-1}]$. In this sense, the conjecture in particular tells us a way to find the restriction of the original category $\mc B$ at the refined vacuum $v$ as a dg category. Note that this implies that the conjectural equivalence will not generally be $\CC[t,t^{-1}]$-linear.
\end{remark}

The idea of symmetry breaking according to the stratification of $\hh^*/W$ was previously described in classical $\mc N=4$ gauge theory in work of Balasubramanian \cite{Balasubramanian} in terms of the Higgs mechanism for the $\gg_{\RR}$-valued scalars in the $\mc N=4$ vector multiplet (although rather than explicitly considering a twisted $\mc N=4$ theory he considers the untwisted theory in a certain subspace of the full space $\gg^3$ of untwisted vacua -- the ``aligned'' vacua).  Balasumramanian also discussed the geometry of this stratification and its preimage in the dual Lie algebra $\gg^*$, and the correspondence between these strata under Langlands duality.  It would be interesting to understand the precise connection between those results and the present work.

\begin{remark}
The situation for geometric Langlands categories -- where the full category of boundary conditions ``decomposes'' according to parameters in the moduli space of vacua -- is comparable to other related stories in geometric representation theory.  For instance, the character theory of Ben-Zvi, Gunningham, and Nadler with gauge group $G$ \cite{BZNCharacter, BZGN} (in particular, see Section 1.4.2 in \cite{BZNCharacter}) lies in a family parameterized by points in the quotient stack $H/W$.  The Hecke categories that these TQFTs assign to the circle arise as a family depending on this parameter (according to the arguments in \cite{BZNLoopSpaces2}) by a twisting procedure essentially equivalent to the twisting procedure we discussed in Section \ref{twisting_section}.
\end{remark}

\begin{example}
If the semisimple lift of $v \in \hh^*/W$ is a central element in $\gg^*$, e.g., $v=0$, then the singular support of a sheaf compatible with the vacuum $v$ will now be contained in the derived pullback $ \arth^v_G(C)$ which is isomorphic to the global nilpotent cone $\mc N_G$ by addition of the constant section of the coadjoint bundle. This is a plausible evidence for the conjecture in this simple case.
\end{example}

There's a nice point of view motivating this conjecture coming from string theory in the case where $G = \GL_n$ (this perspective is not essential; for readers not interested in the connection to string theory it is safe to ignore this brief discussion).  For the gauge group $\GL_n$ we can view the space $\hh^*/W$ of vacua as the configuration space of $n$ points in $\CC$. As mentioned in Remark \ref{IIB_string_remark}, a string-theoretic origin for B-twisted $\mc N=4$ supersymmetric gauge theory for $\GL_n$ is as the effective theory of $n$ coincident D3 branes in a certain twist of type IIB string theory, with two direction left for D3 branes to move in.

Now, we can apply a version of the Higgs mechanism to this setup.  Turning on a vacuum expectation value in $\hh^*/W$ corresponds to moving these $n$ D3 branes away from the origin in this complex plane, which in particular breaks the gauge symmetry group to a Levi subgroup (which one depends on which branes are still coincident after this motion). Hence we expect that motion away from 0 in the moduli of vacua corresponds to this gauge symmetry breaking, motivating our conjecture at the level of categories of boundary conditions. In sum, one can say that both sides try to describe the situation for the same configuration of D branes, but $\IC_{\arth^v_G(C)[t,t^{-1}]}(\flat_G(C))[t,t^{-1}]$ is from the point of view of moving branes from the original configuration, where one cannot geometrically capture the $\ZZ$-grading any more, and $\IC_{\mc N_L}(\flat_L(C))[t,t^{-1}]$ is from the point of view of static branes, where one thinks of the configuration as defining a single particular theory for which one can find a $\ZZ$-grading.

From a different perspective, let's explain some mathematical evidence for this claim. To begin with, we'll explain what we can say for sure about the category restricted to $v$.  By Theorem \ref{main_theorem}, we know that the restriction of the category  $\mc B[t,t^{-1}] = \IC(\flat_G(C))[t,t^{-1}]$ of twisted boundary conditions at the vacuum $v$ is equivalent to the category $\IC_{\mr{Arth}^v_G(C)[t,t^{-1}]}(\flat_G(C))[t,t^{-1}]$ of sheaves with singular support in $\mr{Arth}^v_G(C)$.  What isn't yet known is whether this category is equivalent to the category of sheaves with singular support in $\mc N_L$ (after adjoining $t$).  We have some geometric evidence for this conjecture however.  We know that singular support is defined as support within the space of singularities, in this case $\mr{Arth}_G(C)$.  In fact the loci $\mr{Arth}^v_G(C)$ and $\mc N_L$ inside the space $\mr{Arth}_G(C)$ are not just equivalent, but so are their formal neighborhoods.

\begin{prop}
The formal completions $\mr{Arth}_G(C)^\wedge_{\mr{Arth}^v_G(C)}$ and $\mr{Arth}_L(C)^\wedge_{\mc N_L}$ are equivalent, where $L$ is the centralizer of a semisimple lift $v$ in $\gg^*$ (we won't distinguish $v$ and its lift notationally).
\end{prop}

\begin{proof}
We use a result from derived deformation theory (see Gaitsgory--Rozenblyum {\cite[IV.1 1.6.4, IV.3 3.1.4]{GRvol2}}) which identifies pointed formal moduli problems over a fixed base stack with sheaves of dg Lie algebras on that stack.  Concretely, if $\pi: \, \mc Y \leftrightarrows \mc X:\,\sigma$ is a pointed formal moduli problem over $\mc X$, the associated sheaf of dg Lie algebras over $\mc X$ is the relative restricted shifted tangent complex $\sigma^*\bb T_{\mc Y/\mc X}[-1]$.  

We need to show that there exists an identification $\mc N_L \iso \mr{Arth}^v_G(C)$ so that the two formal completions can be compared.  Using the fact that a central element $v$ in the Lie algebra $\mf l$ of $L$ (identifying $\mf l$ with $\mf l^*$ using an invariant pairing) gives a section of the map $\arth_L(C)\to \Flat_L(C)$, there is a canonical isomorphism $\mr{Arth}_L^v(C) \iso \mc N_L$, by translating the Arthur parameter by $v$.

Moreover, if $L$ is the centralizer of $v$ then the map $\mr{Arth}_L^v(C) \to \mr{Arth}_G^v(C)$ is an equivalence. For this, we observe that for a group scheme $H$ acting on a derived stack $X$ and the orbit $H \cdot Y$ of a substack $Y \inj X$, there is a natural isomorphism
\[H \cdot Y \iso H \times Y / H_Y\]
where $H_Y$ is the stabilizer of $Y$ in $H$, and where $H_Y$ acts on $H \times Y$ by $h \cdot (g,y) = (gh^{-1}, h \cdot y)$. Now by choosing a point $p \in C $, we view $\arth^v_G(C)$ as the quotient $\arth^{v,p}_G(C)/G$, where $\arth^{v,p}_G(C)$ is the derived scheme of Arthur parameters equipped with a framing at the point $p$.  This is directly analogous to the presentation of moduli stacks of Higgs bundles and flat bundles as global quotients due to Simpson \cite[Theorems 9.6 and 9.10]{SimpsonModuliII}. Using the observation with $X = \arth_G^{p}$, $Y = \arth_L^{v,p}$, and $H=G$, so that $H \cdot Y = \arth_G^{v,p}$ holds, we obtain the equivalence we want as follows:
\[\arth^{v,p}_G(C)/G \iso L \backslash (\arth^{v,p}_L(C) \times G)/G \iso L \backslash \arth^{v,p}_L(C).\]

Now it remains to compare the $(-1)$-shifted relative tangent complexes. It suffices to work with the full spaces $T^*[-1] \Flat_G(C)$ and $T^*[-1]\Flat_L(C)$ instead of the Arthur stacks. First of all, the shifted tangent complex at a closed point $(P,\nabla,\phi) \in T^*[-1]\flat_G(C)$ is equivalent to
\[\Omega^\bullet(C ; \mf m_P) \oplus \Omega^\bullet(C ; \mf m^*_P)[-1]\]
where $\mf m$ is the Lie algebra centralizing the value of the section $\phi$ at any point (independent of the choice of point by flatness). In particular, if $(P,\nabla,\phi) \in \mr{Arth}^v_G(C)$ (which should correspond to $ (P_L ,\nabla_L,\phi_L ) \in \mc N_L$) so that $P$ lands in the image of $\Flat_L(C)\to \Flat_G(C)$, then the shifted tangent complex at a point $(P,\nabla,\phi)\in T^*[-1] \Flat_G(C)$ and $ (P_L ,\nabla_L,\phi_L )\in T^*[-1]\Flat_L(C)$ coincide. This yields the desired equivalence of dg Lie algebras, which establishes the equivalence of formal completions.
\end{proof}

At a heuristic level therefore, where the singular support of an ind-coherent sheaf on $X$ is thought of in terms of the ordinary support of an associated object living over $\mr{Sing}(X)$, this suggests that categories of sheaves on $\flat_G(C)$ with singular support in $\mr{Arth}^v_G(C)$ and of sheaves on $\flat_L(C)$ with singular support in $\mc N_L$ should be closely related.

Indeed, one can directly prove Conjecture \ref{gauge_sym_breaking_conjecture} in the case of $C=\bb P^1$ using this proposition. Let us first note a simple consequence of the proposition.

\begin{corollary} \label{support_corollary}
There is an equivalence $\QC(\gg^*/G)_{\tilde{v}} \simeq \QC(\mathfrak{l}^*/L)_{\mc N_L}$ of categories of quasi-coherent sheaves on the coadjoint quotients for $\gg$ and $\mf l$ supported on the orbit $\tilde v$ and the nilpotent cone respectively.
\end{corollary}

To concretely see what is going on here, let's work out an example.

\begin{example}
Suppose $G=\GL_2$ so that we can write $v=(x_1,x_2)$.
\begin{itemize}
\item When $x_1 = x_2$, we have $L=G$ for which the corollary is a tautology.
\item When $x_1 \neq x_2 $, we have $L=H$, the maximal torus.  The corresponding orbit $\tilde{v}$ is a hyperboloid with a transitive action of $\GL_2$; as a stack this is isomorphic to $\{v\}/\stab(v) \iso BH$.  Likewise, the global nilpotent cone for the torus is just $\mc N_H \iso BH$. Their relative tangent complexes are both equivalent to $\hh^*$; This is immediate for $\hh^*/H$ since the adjoint $H$-action is trivial, and for $\gg^*/G$ it follows by considering deformations of the point $v$: the tangent complex $\bb T_{v, \gg^*/G}$ is equivalent to $\hh[1] \oplus \hh^*$.
\end{itemize}
\end{example}

In this example of the curve $C=\bb P^1$, since $\flat_G(C) \simeq \gg[-1]/G$ has no underived directions (i.e. trivial classical truncation) the degree shifting doesn't play any significant role. Then the desired equivalence would follow from Corollary \ref{support_corollary}, which up to this shift is the Koszul dual statement of Conjecture \ref{gauge_sym_breaking_conjecture} for the curve $\bb P^1$.  That is:

\begin{corollary}
Conjecture \ref{gauge_sym_breaking_conjecture} holds for the curve $C = \bb P^1$.  That is, there is an equivalence of categories
\[\IC_{\arth^v_G(\bb P^1) [t,t^{-1}]}(\flat_G(\bb P^1))[t,t^{-1}] \iso \IC_{\mc N_L}(\flat_L(\bb P^1))[t,t^{-1}]\] 
where $L$ is the centralizer of its semisimple lift $x$ in $\gg^*$.
\end{corollary}

\begin{proof}
We can identify $\IC(\flat_G(\bb P^1))[t,t^{-1}]$ with $\QC(\gg^*[2]/G)[t,t^{-1}]$ by Koszul duality (see \cite[Proposition 12.4.2]{ArinkinGaitsgory}).  Singular support conditions for ind-coherent sheaves become ordinary support conditions under this identification, so it suffices to show that
\[\QC(\gg^*[2]/G)[t,t^{-1}]_{\tilde{v}[t,t^{-1}]} \iso \QC(\mf l^*[2]/L)[t,t^{-1}]_{\mc N_L[t,t^{-1}]}\]
or equivalently that
\[\QC(\gg^*/G)_{\tilde{v}} \otimes \CC[t,t^{-1}]\text{-mod} \iso \QC(\mathfrak{l}^*/L)_{\mc N_L} \otimes \CC[t,t^{-1}]\text{-mod},\]
which is just Corollary \ref{support_corollary}. 
\end{proof}

\begin{remark}
Arinkin and Gaitsgory motivated the choice of the nilpotent singular support condition in part by proving that the category of ind-coherent sheaves with nilpotent singular support was the minimal subcategory of $\IC(\flat_G(C))$ where the pullback functors induced from the projection $P \to L$ from a parabolic subgroup to its Levi preserve compact objects.  From the point of view of the geometric Langlands program this condition is important because one expects the geometric Langlands functor to commute with the so-called \emph{geometric Eisenstein series} functors.  If $P \sub G$ is a parabolic subgroup with Levi $L$ then the (spectral) geometric Eisenstein series functor $\mr{Eis}^{\mr{spec}}_P$ is the pull-push functor $p_*q^!$ on ind-coherent sheaves induced from the span
\[\xymatrix{
&\flat_P(C) \ar[dl]_q \ar[dr]^p &\\
\flat_L(C) &&\flat_G(C).
}\]
This functor preserves the condition of having nilpotent singular support.  One can define a pull-push functor of D-modules identically on the automorphic side, and the geometric Langlands equivalence is required to intertwine these functors (possibly up to tensoring with a line bundle).

It's natural to ask what the relationship is between the gauge symmetry breaking we've discussed here and these geometric Eisenstein series functors, since they both involve a relationship between the Arinkin--Gaitsgory category at a group $G$ and a Levi subgroup $L$. Physically these two relationships have quite different origins: on the one hand by motion in the moduli of vacua and on the other by the existence of a parabolic domain wall \footnote{We should remark that the question of why the relationship of boundary conditions on either side of these parabolic domain walls should be preserved by S-duality is beyond the scope of this paper.}.  Conjecture \ref{gauge_sym_breaking_conjecture} implies that these two different notions will be compatible in the following sense.
\begin{corollary}
Let $P \sub G$ be a parabolic subgroup with Levi $L$, and let $H_L = L \cap H$ be the maximal torus in $L$ associated to a choice of maximal torus $H$ in $G$.  Choose a point $v \in \hh_L^*/W_L \sub \hh^*/W$ with stabilizer $M$ in $G$ and $M \cap L$ in $L$.  Then the geometric Eisenstein series functor
\[\mr{Eis}^{\mr{spec}}_P \colon \IC(\flat_L(C)) \to \IC(\flat_G(C))\]
preserves the full subcategory of objects compatible with the vacuum $v$.  That is, it restricts to the geometric Eisenstein series functor
\[\mr{Eis}^{\mr{spec}}_{P \cap M} \colon \IC_{\mc N_{L \cap M}}(\flat_{L \cap M}(C)) \to \IC_{\mc N_M}(\flat_M(C))\]
on the subcategory $\IC_{\mc N_{L \cap M}}(\flat_{L \cap M}(C))$ from 
\[\IC_{\arth_L^v(C) [t,t^{-1}]}(\flat_L(C))[t,t^{-1}] \iso \IC_{\mc N_{L \cap M}}(\flat_{L \cap M}(C)) [t,t^{-1}].\]
\end{corollary}
The fact that the geometric Eisenstein series functor preserves the condition of having singular support in $\arth^v$ is straightforward, by an argument completely identical to that of \cite[Lemma 13.2.5, Proposition 13.2.6]{ArinkinGaitsgory}: one argues verbatim, replacing the condition of nilpotence with the condition of having image $v$ under the evaluation map $\arth_K(C) \to \hh_K^*/W_K$ (where $K = G$ or $L$) and including an invertible degree 2 parameter in order to make sense of the relevant support conditions.  One obtains the conjecture by combining this fact with Conjecture \ref{gauge_sym_breaking_conjecture}, along with the additional claim that the equivalence of Conjecture \ref{gauge_sym_breaking_conjecture} itself commutes with the geometric Eisenstein series functor.
\end{remark}

\subsection{Factorization Conjecture} \label{fact_categories_section}

As a consequence of Conjecture \ref{gauge_sym_breaking_conjecture}, the geometric Langlands categories appear to admit an additional structure: that of \emph{factorization}. For the gauge group $G=\GL_n$, since we can identify $x \in \hh^*/W$ as a point $\{x_1,\cdots,x_n\}$ in the $n$-th symmetric power $\sym^n(\CC)$ of $\CC$,  the category of boundary conditions compatible with the vacuum $x$ can be factorized as
\[\IC_{\mc N_{L_{\underline{x}_{1\cdots n}} }}(\flat_{L_{\underline{x}_{1\cdots n}}}(C) ) \simeq \IC_{\mc N_{L_{\underline{x}_{1\cdots m}}}}(\flat_{L_{\underline{x}_{1\cdots m}}}(C) ) \otimes \IC_{\mc N_{L_{\underline{x}_{m+1 \cdots n}} }}(\flat_{L_{\underline{x}_{m+1\cdots n}}} (C) ) \] 
for every decomposition $\{x_1,\cdots,x_n\}  = \{x_1,\cdots,x_m\} \sqcup \{x_{m+1} , \cdots,x_n\}$ with $x_i \neq x_j$ for $x_i \in \{x_1,\cdots,x_m\}$ and $x_j\in \{x_{m+1} ,\cdots, x_n\}$, which we denote by $\underline{x}_{1\cdots n} =\underline{x}_{1\cdots m}\sqcup \underline{x}_{m+1 \cdots n}$.

In fact, by considering a system 
\[\GL_1\hookrightarrow \GL_2 \hookrightarrow \cdots \hookrightarrow  \GL_n \hookrightarrow \cdots \]
over the base $\sym(\CC) = \coprod_n \sym^n(\CC)$, we should be able to assemble the $\GL_k$ geometric Langlands categories into a \emph{factorization category} (as introduced by Raskin \cite{RaskinChiral}: we refer the reader to that work for details on the theory of factorization categories). Note that we obtain a factorization object with an additional grading: a family depending on the whole configuration space of $\CC$ rather than depending just on the Ran space $\mr{Ran}(\CC)$, as discussed in \cite{FinkelbergSchechtman} and more recently revisited from a different perspective by Mirkovi\'c \cite{MirkovicLocal}.

\begin{remark}
As we discussed in Remark \ref{IIB_string_remark}, a certain twist of type IIB string theory on the background $T^*X \times \CC$ for a complex surface $X$ with $n$ D3 branes along $X \times \{x_1,\cdots, x_n\}$ is responsible for the B-twisted supersymmetric gauge theory with gauge group $\GL_n$. Therefore, the belief that the categorical geometric Langlands correspondence factorizes in the above sense corresponds to the motion of configurations of D3 branes in its string theoretic origin. The idea that one considers varying number of D branes is natural from the string theoretic point of view, because D branes are dynamic objects and in particular once we introduce D branes to the stage, the most natural thing to do is to let the number of them vary without any restriction, leading to our configuration. The idea that the resulting category of boundary conditions for the worldvolume theory should admit a factorization structure coming from vacuum compatibility seems new even in the physics literature, as far as we are aware.
\end{remark}

In order to formulate a precise conjecture, we'll keep the technical details manageable by first thinking of the geometric Langlands categories after composition with the map $\sym(\CC)\to \mr{Ran}(\CC)$ and second by considering an algebraic rather than categorical version of the statement by taking Hochschild cochains of the categories to obtain \textit{bona fide} factorization algebras.

\begin{conjecture} \label{factorization_conjecture}
There exists a factorization algebra $\mc C$ over $\CC$ with the stalk at a point $x\in \mr{Ran}(\CC)$ being 
\[ \mc C_{x} = \bigoplus_{n \geq 1 , \; \tilde{x} \in \sym^n(\CC) } \mr{HC}^\bullet ( \IC_{\mc N_{L_{ \tilde{x}}}}  (\flat_{L_{\tilde{x}}}(C) )),\] 
where $\tilde{x} \in \sym^n(\CC)$ is a lift of $x \in \mr{Ran}(\CC)$.
\end{conjecture}

The operator product expansion here would carry fundamental information on how to relate geometric Langlands correspondences for different gauge groups. We hope to further investigate this conjecture and its variants in future work. 

\begin{remark}
A factorization algebra will define in particular a constructible sheaf of vector spaces over $\CC^k$ for each $k$. The provided stalk data determines the restriction of this conjectural sheaf to each stratum in $\CC^k$ with respect to the natural stratification by partitions of $k$, as a constant sheaf with the given value.  The conjecture says that it's possible to construct canonical exit path data assembling these constant sheaves into a constructible sheaf on each $\CC^k$.  This being done, the resulting sheaf on $\mr{Ran}(\CC)$ will automatically satisfy factorization.
\end{remark}

We'll conclude with some motivation for this conjecture coming from factorization structures appearing in a related construction: the context of cohomological Hall algebras in the work of Kontsevich and Soibelman \cite{KontsevichSoibelmanCHA}. We claim that the underlying reason why this factorization structure appears is similar to that of our setup. To explain this, let's consider a certain twisted type IIA string theory on a background of the form $\bb R^2 \times X\times  \bb C$, where $X$ is a Calabi--Yau 3-fold. Here the string theory is twisted (in the sense of Costello and Li \cite{CostelloLiSUGRA}) in such a way that it is the A-model on the background $\bb R^2$ and the B-model on $X\times \bb C$ with a linear superpotential on $\bb C$. A typical $\text{D}_{2k}$ brane of the IIA string theory in this background is of the form $\bb R \times Y \times \{0\}$, where $\bb R\subset \bb R^2$ is Lagrangian and $Y\subset X $ is a submanifold of complex dimension $k$.

We think of $X$ as defining a Calabi--Yau 3-category. On the other hand, a quiver with superpotential yields a Calabi--Yau 3-category of representations of the quiver, where an edge corresponds to a class in $\mr{Ext}^1$ and Serre duality determines the other Ext groups. Indeed, for some Calabi--Yau 3-folds, by thinking of spherical generators, one can explicitly construct such a quiver. Then the dimension vector of a given quiver is precisely the number of corresponding branes in the configuration.

\begin{example}
The Calabi--Yau 3-fold $X=\bb C^3$ corresponds to a quiver with a single vertex, three self-loops denoted by $x,y,z$, and a nontrivial superpotential $W =  xyz - xzy$. Here having a single vertex represents the fact that there is a unique compact cycle in $\bb C^3$, namely, a point. On the other hand, from the perspective of string theory, this point corresponds to a position $ \bb R\times \pt$ on which a D0 brane lives. For a fixed number of coincident D0 branes, say $N$, the Hilbert space of the 1-dimensional worldline theory is the equivariant cohomology $\mr H^\bullet_{\GL_N}(\mr{Crit}(W)_{X,Y,Z\in \mathfrak{gl}_N}, \mc P)$, where $\mc P$ is a sheaf of vanishing cycles.

Now the theorem of Kontsevich and Soibelman asserts that $\bigoplus_{N\geq 1 } \mr H_{\GL_N}^\bullet( \mr{Crit}(W) _{X,Y,Z\in \mathfrak{gl}_N}, \mc P)$ forms a factorization system over $\bigoplus_{N\geq 1} \mr H_{\GL_N}^\bullet( \pt )$. From our interpretation, we expect such a factorization structure to appear by varying the positions of the D0 branes; this is possible only in the direction of remaining B-model background which is identified as $\mr H_{\GL_N}^\bullet(\pt) \simeq \OO( \sym^N(\bb C))$. This is completely analogous to our story for D3 branes in type IIB string theory.
\end{example}

In other words, we propose a new source of factorization structures in terms of string theory. Incidentally, one can take this proposal as an answer to the question of Kontsevich and Soibelman to identify the origin of factorization system appearing in their work \cite[Section 2.5]{Soibelman}. We believe that there are many other surprising instances of factorization structures to be discovered from this perspective. 

\appendix
\appendixpage
\addappheadtotoc
\section{Singular Support for Global Complete Intersections} \label{GCI_appendix}

In this appendix we'll describe the proof of Proposition \ref{SS_complete_intersection_prop} from Section \ref{main_SS_section}.  The argument we'll give is only a slight modification of the arguments in \cite[Section 9.2]{ArinkinGaitsgory}.  We give a proof here, following Arinkin and Gaitsgory closely, in order to justify the modification where the auxiliary degree 2 parameter is included.

Recall the setup: $\mc Z$ was a global complete intersection stack, which meant that it could be written as a fiber product smooth stacks 
\[\xymatrix{
\mc Z \ar[r] \ar[d] &\mc U \ar[d] \\
\mc X\ar@<0.6ex>[r]&\mc V \ar@<0.6ex>[l]
}\]
where $\mc X \to \mc V$ is a section of a smooth schematic map $\mc V\to \mc X$.  One can associate to such a global complete intersection a pair of groupoids which we'll use heavily: namely $\mc G_{\mc Z/\mc U} := \mc Z \times_{\mc U} \mc Z$ as a derived group scheme over $\mc Z$ and $\mc G_{\mc X/\mc V} := \mc X \times_{\mc V} \mc X$, a derived group scheme over $\mc X$.

There is a canonical morphism $f \colon \sing(\mc Z) \to \sing(\mc G_{\mc X/\mc V} \times_{\mc X} \mc U) \iso V \times_{\mc X} \mc U$, obtained by first embedding $\sing(\mc Z)$ into the total space $V \times_{\mc X} \mc Z$ of the relative tangent complex $\bb T_{\mc Z /\mc X}$, where $V$ is the total space of the restriction of $\bb T_{\mc Z /\mc X}$ to $\mc X$, and then composing with the defining map $\mc Z \to \mc U$.  This definition of $V$ is equivalent to the definition of $V$ given in Section \ref{main_SS_section} by \cite[9.1.4]{ArinkinGaitsgory}.

\begin{prop}
If $Y \sub \sing(\mc Z)$ is a closed subset of the form $f^{-1}(\wt Y \times_{\mc X}\mc U)$, then there is a canonical equivalence
\[\IC_{Y[t,t^{-1}]}(\mc Z)[t,t^{-1}] \iso  \IC(\mc Z)[t,t^{-1}]  \otimes_{\mr{HC}^\bullet(\mc X / \mc V)[t,t^{-1}]^{\mr{op}}\text{-mod}}   \mr{HC}^\bullet(\mc X / \mc V)[t,t^{-1}]^{\mr{op}}\text{-mod}_{\wt Y[t,t^{-1}]} ,\]
where the category $\mr{HC}^\bullet(\mc X / \mc V)^{\mr{op}}\text{-mod}$ acts on $\IC(Z)$ by first using Koszul duality to identify it with the category $\IC(\mc G_{\mc X/\mc V})$, then using the monoidal functor
\[\IC(\mc G_{\mc X/\mc V}) \otimes_{\QC(\mc X)} \QC(\mc U) \to \IC(\mc G_{\mc Z/\mc U})\]
to define an action of $\IC(\mc G_{\mc X/\mc V})$ on $\IC(\mc Z)$. 
\end{prop}

\begin{proof}
We'll start by recalling the main ideas exploited by Arinkin and Gaitsgory, then introduce the parameter $t$.  According to \cite[Section 9.2]{ArinkinGaitsgory}, in this context we can understand singular support conditions for $\IC(\mc Z)$ in terms of an action of the monoidal category $\IC(\mc G_{\mc X/\mc V})$.  Indeed, the groupoid $\mc G_{\mc Z/\mc U}$ acts on $\mc Z$, and so the category $\IC(\mc Z)$ is a module over the monoidal category $\IC(\mc G_{\mc Z/\mc U})$.  There's a natural monoidal pullback functor 
\[\IC(\mc G_{\mc X/\mc V}) \otimes_{\QC(\mc X)} \QC(\mc U) \to \IC(\mc G_{\mc Z/\mc U})\]
which makes $\IC(\mc Z)$ into a module over $\IC(\mc G_{\mc X/\mc V}) \otimes_{\IC(\mc X)} \IC(\mc U)$, or indeed over $\IC(\mc G_{\mc X/\mc V})$.  In order to understand singular support in terms of this action, we first identify the category on the left hand side as
\begin{align*}
\IC(\mc G_{\mc X/\mc V}) \otimes_{\QC(\mc X)} \QC(\mc U) &\iso \IC(\mc G_{\mc X/\mc V}) \otimes_{\QC(\mc G_{\mc X/\mc V})} \left(\QC(\mc G_{\mc X/\mc V}) \otimes_{\QC(\mc X)} \QC(\mc U) \right) \\
&\iso \IC(\mc G_{\mc X/\mc V}) \otimes_{\QC(\mc G_{\mc X/\mc V})} \QC(\mc G_{\mc X/\mc V} \times_{\mc X} \mc U)\\
&\iso \IC(\mc G_{\mc X/\mc V} \times_{\mc X} \mc U).
\end{align*}

By \cite[Lemma 9.2.6]{ArinkinGaitsgory} the category $\IC_Y(\mc Z)$ is equivalent to the restricted category
\[\IC(\mc Z) \otimes_{\IC(\mc G_{\mc X/\mc V} \times_{\mc X} \mc U)} \IC_{f(Y)}(\mc G_{\mc X/\mc V} \times_{\mc X} \mc U).\]
We'll want a slightly modified version of this statement.  

Let us bring the auxiliary parameter $t$ into the story.  We claim that there is an equivalence
\[\IC_{f^{-1}(Y)[t,t^{-1}]}(\mc Z)[t,t^{-1}] \iso \IC(\mc Z)[t,t^{-1}] \otimes_{\IC(\mc G_{\mc X/\mc V} \times_{\mc X} \mc U)[t,t^{-1}]} \IC_{Y[t,t^{-1}]}(\mc G_{\mc X/\mc V} \times_{\mc X} \mc U)[t,t^{-1}]\]
where now $Y \sub \sing(\mc G_{\mc X/\mc V} \times_{\mc X} \mc U)$ is a Zariski-closed subset.  

This is proven in exactly the same way as \cite[Lemma 9.2.6]{ArinkinGaitsgory}.  As they argue, it suffices to prove the claim in the case where $\mc Z$ is affine.  This argument is not affected by the introduction of the parameter $t$ since the descent takes place independently of this parameter (heuristically we're performing descent for maps $Z \times \spec \CC[t,t^{-1}] \to \mc Z \times \spec \CC[t,t^{-1}]$ which are constant in the second factor).  The embedding $f \colon \sing(\mc Z) \to V \times_{\mc X} \mc U$ induces a map $f^*\colon \sym(V)[t,t^{-1}] \otimes_{\OO(\mc X)} \OO(\mc U) \to \OO(\sing(\mc Z))[t,t^{-1}]$.  We then note that the triangle
\begin{equation} \label{triangle_diagram} \xymatrix{
\sym(V)[t,t^{-1}] \otimes_{\OO(\mc X)} \OO(\mc U) \ar[rr]^{f^*} \ar[dr] &&\OO(\sing(\mc Z))[t,t^{-1}] \ar[dl] \\
&\mr{HH}^{\mr{even}}(\IC(\mc Z)[t,t^{-1}]) &
} \tag{*} \end{equation}
defined using this map and the action of $\IC(\mc G_{\mc X/\mc V} \times \mc U)[t,t^{-1}]$ on $\IC(\mc Z)[t,t^{-1}]$ commutes.  This implies the claim that we want since we can equivalently take the restriction to $f^{-1}(Y)[t,t^{-1}]$ with respect to the $\OO(\sing(\mc Z))[t,t^{-1}]$-action, or the restriction to its preimage under $f^*$, i.e. the restriction to $f(f^{-1}(Y))[t,t^{-1}] = Y[t,t^{-1}]$ with respect to the  $\sym(V)[t,t^{-1}] \otimes \OO(\mc U)$-action.

In particular, suppose $Y \sub V \times_{\mc X} \mc Z$ is a subset of the form $\wt Y \times_{\mc X} \mc Z$ for some closed subset $\wt Y \sub V$.  Then according to \cite[Proposition 8.4.14]{ArinkinGaitsgory} there is an equivalence
\[\IC_Y(\mc G_{\mc X/\mc V} \times_{\mc X} \mc U) \iso \IC_{\wt Y}(\mc G_{\mc X/\mc V}) \otimes_{\QC(\mc X)} \QC(\mc U).\]
We can make the same argument identically including the parameter $t$ in order to say that there is an equivalence
\[\IC_{Y[t,t^{-1}]}(\mc G_{\mc X/\mc V} \times_{\mc X} \mc U)[t,t^{-1}] \iso \IC_{\wt Y[t,t^{-1}]}(\mc G_{\mc X/\mc V})[t,t^{-1}] \otimes_{\QC(\mc X)[t,t^{-1}]} \QC(\mc U)[t,t^{-1}].\]

The upshot of all this is that if $Y \sub \sing(\mc Z)$ is a closed subset of the form $f^{-1}(\wt Y \times_{\mc X}\mc U)$ then we can understand the restricted category $\IC_{Y[t,t^{-1}]}(\mc Z)[t,t^{-1}]$ as the restriction along $\wt Y[t,t^{-1}]$ of $\IC(\mc Z)[t,t^{-1}]$ with respect to the action of $\IC(\mc G_{\mc X/\mc V})[t,t^{-1}]$.  Even better, by Koszul duality \cite[9.1.2]{ArinkinGaitsgory} (this is just the observation that the counit of the adjunction \ref{E_2_adjunction_prop} is an equivalence), we can identify $\IC(\mc G_{\mc X/\mc V})$ with the category $\mr{HC}^\bullet(\mc X / \mc V)^{\mr{op}}\text{-mod}$ and thus \cite[Corollary 9.1.7]{ArinkinGaitsgory} we are able to identify $\IC_{Y[t,t^{-1}]}(\mc Z)[t,t^{-1}]$ as the restriction along $\wt Y[t,t^{-1}]$ of $\IC(\mc Z)[t,t^{-1}]$ with respect to the action of $\mr{HC}^\bullet(\mc X / \mc V)[t,t^{-1}]^{\mr{op}}\text{-mod}$.  Specifically here we're using Koszul duality \emph{before} taking the tensor product with $\CC[t,t^{-1}]\text{-mod}$; it's still true that this induces an equivalence of categories
\[\mr{HC}^\bullet(\mc X / \mc V)[t,t^{-1}]^{\mr{op}}\text{-mod}_{Y[t,t^{-1}]} \iso \IC_{Y[t,t^{-1}]}(\mc G_{\mc X/\mc V})[t,t^{-1}]\]
with specified support conditions.  It suffices to check this smoothly locally, so we can assume $\mc X$ is affine in which case it follows from commutativity of the triangle (\ref{triangle_diagram}).
\end{proof}

\bibliographystyle{alpha}
\bibliography{GL-KW}

\textsc{Institut des Hautes \'Etudes Scientifiques}\\
\textsc{35 Route de Chartres, Bures-sur-Yvette, 91440, France}\\
\texttt{celliott@ihes.fr}\\
\vspace{5pt}\\
\textsc{Department of Mathematics, Yale University}\\
\textsc{10 Hillhouse Ave., New Haven, CT 06511, USA} \\
\texttt{philsang.yoo@yale.edu}
\end{document}